%% file: main_journal.tex
\documentclass[abstract=true, DIV=14]{scrartcl}

\usepackage[noEnd,commentColor=black]{algpseudocodex}
\usepackage{amsmath}
\usepackage{amssymb}
\usepackage{amsthm}
\usepackage[mathscr]{euscript}
\usepackage[shortlabels]{enumitem}
\usepackage{graphicx} 
\usepackage{subcaption}
\usepackage{thmtools}
\usepackage{todonotes}
\usepackage{xcolor}
\usepackage{microtype}

\makeatletter
\def\@fnsymbol#1{\ensuremath{\ifcase#1\or \!\;\or \!\;\or \ddagger\or
   \mathsection\or \mathparagraph\or \|\or **\or \dagger\dagger
   \or \ddagger\ddagger \else\@ctrerr\fi}}
\makeatother

\usepackage[normalem]{ulem}

\usepackage{tcolorbox}

\usepackage{upgreek}
\usepackage[colorlinks]{hyperref}
\usepackage[capitalize]{cleveref}

\crefname{equation}{eq.}{eqs.} 
\crefname{enumi}{}{} 
\crefname{icase}{case}{cases}
\crefname{ipart}{part}{parts}
\crefname{iprop}{property}{properties}
\crefname{iinv}{invariant}{invariants}

\crefformat{section}{\S\,#2#1#3}
\crefformat{subsection}{\S\,#2#1#3}
\crefrangeformat{section}{\S\S\,#3#1#4--#5#2#6}
\crefmultiformat{section}{\S\S\,#2#1#3}{,~#2#1#3}{,~#2#1#3}{,~#2#1#3}

\usepackage{authblk}
\usepackage[normalem]{ulem}

\usepackage{tikz}
\usetikzlibrary{calc}

\newcommand{\N}{\mathbb{N}}
\newcommand{\R}{\mathbb{R}}
\newcommand{\fO}{\mathcal{O}}

\DeclareMathOperator{\ex}{ex} 
\DeclareMathOperator{\tww}{tww} 
\DeclareMathOperator{\tw}{tww} 

\DeclareMathOperator{\OPT}{\mathsf{OPT}}
\DeclareMathOperator{\TSP}{\mathsf{TSP}} 
\DeclareMathOperator{\MST}{\mathsf{MST}} %
\DeclareMathOperator{\MStT}{\mathsf{MStT}} 
\DeclareMathOperator{\NN}{\mathsf{NN}} 
\newcommand{\MN}{\mathsf{SP}}
\newcommand{\SM}{\mathsf{SM}}

\renewcommand{\alpha}{\upalpha}

\newcommand{\cR}{\mathscr{R}}

\newcommand{\ceil}[1]{\lceil #1 \rceil}
\newcommand{\floor}[1]{\lfloor #1 \rfloor}

\newcommand{\connected}[1]{\def\temp{#1}\ifx\temp\empty\sim\else\overset{#1}{\sim}\fi}

\newcommand{\Av}{\mathrm{Av}}

\tikzset{
	point/.style={circle, fill, inner sep=1.5pt},
	smallpoint/.style={point, inner sep=1.2pt},
	tinypoint/.style={point, inner sep=1pt},
	hlbox/.style={fill, {white!90!black}},
	subrect/.style={draw, fill={white!80!cyan}}, 
	msrect/.style=subrect 
}

\newtheorem{theorem}{Theorem}[section]
\newtheorem{lemma}[theorem]{Lemma}

\newtheorem{corollary}[theorem]{Corollary}
\newtheorem{observation}[theorem]{Observation}
\theoremstyle{definition}

\title{Optimization with pattern-avoiding input 
}

\author[1]{Benjamin Aram Berendsohn\thanks{$^{1,2}$Supported by DFG Grant KO 6140/1-2. Work done while at Freie Universität Berlin.}}
\author[2]{L\'aszl\'o Kozma\protect\footnotemark[1]}
\affil[1]{Max Planck Institute for Informatics, Saarbrücken, Germany}
\affil[2]{Faculty of Computer Science, TU Dresden, Germany}
\author[3]{Michal Opler\thanks{$^3$This work was co-funded by the European Union under the project Robotics and advanced industrial production (reg. no. CZ.02.01.01/00/22\_008/0004590).}}
\affil[3]{Czech Technical University in Prague, Czech Republic}

\date{}

\begin{document}

\maketitle

\begin{abstract}
Permutation pattern-avoidance is a central concept of both enumerative and extremal combinatorics. In this paper we study the effect of permutation pattern-avoidance on the complexity of optimization problems. 

In the context of the \emph{dynamic optimality conjecture} (Sleator, Tarjan, STOC 1983), Chalermsook, Goswami, Kozma, Mehlhorn, and Saranurak (FOCS 2015) conjectured that the amortized search cost of an optimal binary search tree (BST) is \emph{constant} 
whenever the search sequence 
is pattern-avoiding, generalizing earlier conjectures and results (\emph{sequential}-, \emph{traversal-}, \emph{deque-} access) that can be seen as avoiding concrete small patterns.  
The best known bound to date is 
$2^{\alpha{(n)}(1+o(1))}$ recently obtained by Chalermsook, Pettie, and Yingchareonthawornchai (SODA 2024); here $n$ is the BST size and $\alpha(\cdot)$ the inverse-Ackermann function.
In this paper we resolve the conjecture, showing a tight $\fO(1)$ bound.
This indicates a barrier to dynamic optimality: any candidate online BST (e.g., splay trees or greedy trees) must match this optimum, but current analysis techniques only give superconstant bounds.

More broadly, we argue that the \emph{easiness} of pattern-avoiding input is a general phenomenon, not limited to BSTs or even to data structures. To illustrate this, we show that when the input avoids an \emph{arbitrary, fixed, a priori unknown} pattern, one can efficiently compute: 

\begin{itemize}
\item a \emph{$k$-server solution} of $n$ requests from a unit interval, with total cost $n^{\fO(1/\log k)}$, in contrast to the worst-case $\Theta(n/k)$ bound, and
\item a \emph{traveling salesman tour} of $n$ points from a unit box, of  length $\fO(\log{n})$, in contrast to the worst-case $\Theta(\sqrt{n})$ bound; similar results hold for the euclidean \emph{minimum spanning tree}, \emph{Steiner tree}, and \emph{nearest-neighbor} graphs.

\end{itemize}

We show the above results to be tight. Our techniques build on the Marcus-Tardos proof of the Stanley-Wilf conjecture, and on the recently emerging concept of \emph{twin-width}.

\end{abstract}

\newpage

\tableofcontents
\newpage

\section{Introduction}\label{sec1}

Modeling mathematical structure by a \emph{finite list of substructures} that an infinite family of objects \emph{must avoid} is a cornerstone of modern combinatorics. For instance, the \emph{graph minor theory} of Robertson and Seymour characterizes graph properties closed under edge-contractions and edge- and vertex deletions, by a finite list of forbidden minors (e.g., see~\cite{RS, Lovasz,Diestel}). On the algorithmic side, forbidden minors imply the existence of small separators~\cite{Alon_sep, Reed_sep}, which restrict the solution-structure of various optimization problems, leading to algorithmic improvements (e.g., see~\cite{epp_tw, Grohe03, reed2003algorithmic, Dem_tw, param_book}). 

In permutations a rich theory of \emph{pattern-avoidance} has developed (e.g., see the surveys~\cite[\S\,12]{bona2015handbook}, \cite{ kitaev2011patterns,bona2022combinatorics}), however, this theory has so far focused on enumeration and extremal questions, and algorithmic consequences have been much less studied. Two notable exceptions are algorithmic questions related to \emph{binary search trees} (BST), and \emph{permutation pattern matching} (PPM). In this paper we characterize the complexities of three fundamental optimization problems when the input is pattern-avoiding: BST, $k$-server on the line, and euclidean TSP. We build on the \emph{twin-width} decomposition of permutations initially developed by Guillemot and Marx~\cite{GM_PPM} for the PPM problem. We start by defining permutation pattern avoidance, the central notion of our paper.
\smallskip

\begin{tcolorbox}[boxsep=4pt,left=0pt,right=0pt,top=0pt,bottom=0pt, colframe=gray,colback=white,boxrule=1pt,arc=2.4pt]
A permutation $\tau = \tau_1, \dots, \tau_n$ \textbf{contains} a permutation pattern $\pi = \pi_1, \dots, \pi_k$, if there are indices $i_1 < \cdots < i_k$ for which $\tau_{i_j} < \tau_{i_\ell} \Leftrightarrow \pi_j < \pi_\ell$, for all $j,\ell \in [k]$. In words, $\tau$ has a subsequence (not necessarily contiguous) that is \emph{order-isomorphic} to $\pi$; otherwise  $\tau$ \textbf{avoids} $\pi$ (see Fig.~\ref{fig:perm-reprs}). 
\end{tcolorbox}

\paragraph{Binary search trees.} The BST problem
concerns one of the most fundamental data structuring tasks: serving a sequence of accesses\footnote{An \emph{access} is a successful search, where the node storing the searched key is visited. The model can be easily extended to accommodate unsuccessful searches, insertions, deletions, and other operations; the setting with accesses, however, already captures the essential complexity of the problem, e.g., see~\cite{DHIKP, FOCS15}. } $X = (x_1, \dots, x_m) \in [n]^m$ in a binary search tree that stores the set $[n]$. Standard balanced trees achieve $\fO(m\log{n})$ total cost, which is, in general, best possible~\cite{Knuth3}. For some structured sequences $X$, however, the cost can be reduced by adjusting the BST between accesses, via rotations. 
More precisely, each access starts with a pointer at the root, and we may do rotations at the pointer, or move the pointer along an edge (both being unit cost operations), so that the pointer visits the accessed node at least once. 
Denote by $\OPT(X)$ the lowest cost with which $X$ can be served in such a BST. We call $\OPT(X)$ the \emph{offline} optimum, since it can be computed with a priori knowledge of $X$ (we make this standard model more precise in~\S\,\ref{sec3.1}).

The \emph{dynamic optimality conjecture}~\cite{ST85} asks whether the \emph{splay tree}, 
a simple and elegant \emph{online} strategy for BST adjustment, achieves cost $\fO(\OPT(X))$ for all $X$ (online means that access $x_i$ is revealed only when $x_1, \dots, x_{i-1}$ have already been completed). The conjecture is more recently understood as asking whether \emph{any} online BST can achieve this; a prime candidate besides splay is the \emph{greedy BST}~\cite{DHIKP}. The question has inspired four decades of research and remains open.\footnote{For more background and related work on dynamic optimality, we refer to the surveys~\cite{iacono_survey, LT19}.} 

The quantity $\OPT(X)$ is poorly understood; much of what we know about it comes from analysing \emph{online algorithms}.\footnote{Computing $\OPT(X)$ exactly is not known to be either in P or NP-hard; the best known polynomial-time approximation ratio is $\fO(\log\log{n})$ achieved by Tango trees~\cite{Tango}, an online algorithm.}
For almost all $X \in [n]^m$, $\OPT(X) \in \Omega(m\log{n})$~\cite{BlumCK, greedy_deque}, which can be matched by a simple balanced tree. Dynamic optimality is thus mainly about sequences $X$ whose structure allows BSTs to \emph{adapt to them} via rotations, enabling $\OPT(X) \in o(m \log{n})$ or even $\OPT(X) \in \Theta(m)$.

In~\cite{FOCS15} it was observed that several classical structured sequences can be described in terms of \emph{permutation pattern avoidance},  
and well-studied corollaries of dynamic optimality, e.g., deque-, traversal-, sequential- access~\cite{ST85}, correspond to the avoidance of concrete small patterns. Thus, as a broad generalization, ~\cite{FOCS15} conjectured that $\OPT(X) \in \fO(m)$ whenever $X \in [n]^m$ is an access sequence\footnote{We mostly focus on the case when $X$ is a permutation, in particular $m=n$. This assumption is customary and not a significant restriction: As shown in~\cite{FOCS15}, bounds can be transferred from permutations to general access sequences, and dynamic optimality for permutations implies dynamic optimality in full generality.} that avoids an \emph{arbitrary} constant-sized pattern.

It is a priori not obvious why the avoidance of an arbitrary pattern should make BST access easier. Our current understanding comes from special cases, and from 
analysing the online \emph{greedy BST}~\cite{DHIKP}, which appears particularly well-behaved in this setting. In~\cite{FOCS15} it was shown that the total cost of greedy for serving an access sequence $X$ that avoids a $k$-permutation is $2^{\alpha{(n)}^{\fO(k)}} \cdot m$, which was recently improved to $2^{\fO(k^2) + (1+o(1))\alpha{(n)}} \cdot m$ in~\cite{ChalermsookPettieY}; $\alpha(\cdot)$ is the slowly growing inverse-Ackermann function. These were to date the best bounds on $\OPT(X)$, leaving open whether $\fO(1)$ cost per access, i.e., independent of the tree size $n$, is possible. 
Our first main result answers this question, settling the complexity of BST access with pattern-avoiding input: 

\begin{theorem}[informal]\label{thm1}
A BST access sequence of length $m$ that avoids a permutation pattern of size $k$ can be served with total cost $2^{\fO(k)} \cdot m$.
\end{theorem}

The result can also be interpreted as a barrier towards dynamic optimality. Any candidate online BST that is to be competitive with $\OPT(X)$ must achieve $\fO(m)$ cost when $X$ avoids a fixed pattern. 

For splay, no improvement over the trivial $\fO(m \log{n})$ is known, apart from very special cases, e.g., when $X = 1,2,3,\dots$, so $X$ avoids $(2,1)$~\cite{Tarjan85, LTpre}. For greedy BST, the above bounds are obtained by using a certain \emph{input-revealing property} that allows relating the cost to a problem in extremal combinatorics~\cite{FOCS15}. 
It is known that this type of analysis cannot give a linear bound for greedy~\cite{FOCS15}, and that the best known superlinear bound for the extremal problem is essentially tight~\cite{ChalermsookPettieY}.  


Our contribution is a different, perhaps more general approach: constructing and analysing an offline BST solution through the \emph{twin-width} decomposition of permutations introduced by Guillemot and Marx~\cite{GM_PPM}.



\paragraph{Encoding permutations.}
A landmark result of combinatorics is the proof of the Stanley-Wilf
 conjecture by Marcus and Tardos~\cite{MarcusTardos}, also building on results by Füredi, Hajnal~\cite{FurediHajnal}, and Klazar~\cite{Klazar}. The result states that the number of permutations of length $n$ that avoid an (arbitrary) permutation $\pi$ (the pattern) is at most ${s_\pi^n}$, i.e., single-exponential in $n$, with a base $s_\pi$ that depends only on the pattern.\footnote{The asymptotic behavior of $s_\pi$ is far from fully understood, see~\S\,\ref{sec2}.} 


Serving an $n$-permutation $\tau$ as an access sequence in a BST with total cost $\mathcal{C}$ gives an \emph{encoding} of $\tau$ with $\fO(\mathcal{C}+n)$ bits~\cite{BlumCK, greedy_deque}. To see this, observe that each unit cost step in the execution corresponds to one of constantly many choices (e.g., move pointer up/left/right, rotate at the pointer, etc.). We can reconstruct $\tau$ by replaying the sequence of operations starting from the initial tree (encoded with another $2n$ bits). Theorem~\ref{thm1} implies a cost $\mathcal{C} \in \fO_{\pi}(n)$, and thus an encoding of $\tau$ with $\fO_{\pi}(n)$ bits, which in turn implies that there are only $2^{\fO_{\pi}(n)}$ $\pi$-avoiding $n$-permutations, giving an ``algorithmization'' of the Stanley-Wilf / Marcus-Tardos theorem.\footnote{Different ``algorithmizations'' have been obtained previously in other settings~\cite{BonnetBourneufEtAl2024,PilipczukSokolowskiEtAl2022}.}

\paragraph{Permutation pattern matching.} The PPM problem
asks whether a given $n$-permutation $\tau$ avoids a given $k$-permutation $\pi$. It is a well-studied algorithmic question (e.g., see~\cite{Bose_PPM, AR_PPM, GM_PPM, BKM_PPM, GR_PPM}) and a breakthrough result by Guillemot and Marx~\cite{GM_PPM} achieved the runtime $2^{\fO(k^2\log{k})}\cdot n$, with an improvement to $2^{\fO(k^2)} \cdot n$ by Fox~\cite{jfox}, showing that PPM is \emph{fixed-parameter tractable} in terms of the pattern-size $k$.

\paragraph{Twin-width.} The PPM algorithm of Guillemot and Marx works by dynamic programming over a decomposition of permutations that they introduce. This decomposition has later been generalized to the concept of \emph{twin-width} in permutations, graphs, as well as in more general structures~\cite{bonnet2021twin, bonnet2021twinperm}. The Guillemot-Marx result (with the improvement by Fox) implies that if an $n$-permutation $\tau$ avoids a $k$-permutation $\pi$, then $\tww(\tau) \leq 2^{\fO(k)}$, where $\tw(\tau)$ is the twin-width of~$\tau$. 
We adapt the decomposition to compute a solution to the BST problem (i.e., a way to serve the access sequence $\tau$ via rotations and pointer moves) of cost $\fO{(n \cdot \tww^2(\tau))}$, implying Theorem~\ref{thm1}. 
The discussion leading to Theorem~\ref{thm1} suggests the following question:

\smallskip
\textit{Is the easiness of pattern-avoiding inputs specific to BSTs, or a broader phenomenon?}

\smallskip
In this paper we emphatically argue for the latter, illustrating it via two representative problems from different areas: the \emph{$k$-server} and the \emph{traveling salesman} (TSP) problems. 

\paragraph{$k$-server.} This problem
is one of the central problems in online optimization~\cite{Manasse, Borodin, K09}. It asks to serve $n$ requests (points in a metric space $\mathcal{M}$) revealed one-by-one, by moving one of the $k$ servers (also located at points of $\mathcal{M}$) to the request. The cost is the total movement of servers (as distance in $\mathcal{M}$); denote the optimum for a sequence $X \in \mathcal{M}^n$ of requests by $\OPT_k(X)$. 

The famous \emph{$k$-server conjecture} asks whether an online algorithm can achieve cost at most $k \cdot \OPT_k(X)$ for all $X$. In general, this is still open (e.g., see~\cite{K09}), but when $\mathcal{M}$ is the real line, a $k$-competitive algorithm is known: the simple and elegant \emph{double coverage}\,(DC)~\cite{Chrobak}. 
Our focus is this well-studied case of the $k$-server problem. 
In contrast to the BST problem, the input now 
is a sequence $X$ of $n$ \emph{reals} (for simplicity, we let $X \in[0,1]^n$). Simple examples show that in the worst case $\OPT_k(X)$ is linear, i.e., $\OPT_k(X) \in \Theta(n/k)$. Pattern-avoidance in the sequence $X$ of requests can be defined in the natural way. Our second main result shows that $\OPT_k(X)$ is significantly smaller if $X$
is pattern-avoiding:

\begin{theorem}[informal]\label{thm2}
A sequence of $n$ requests from the interval $[0,1]$ that avoids a permutation pattern $\pi$ of size $t$ can be served with $k$ servers at total cost
$k^{\fO(1)}n^{\fO_t(1/\log k)}$.
\end{theorem}
We show that the exponent of $n$ is best possible for almost all $\pi$.
On the other hand, for certain avoided patterns $\pi$, we show a stronger upper bound of roughly $n^{\fO(1/k)}$.
Note that our bound on the cost also applies to an efficient online algorithm (DC) with a cost-increase of at most a factor $k$. In our view, this exemplifies the usefulness of competitive analysis of algorithms: nontrivial structural properties shown for the offline optimum are automatically exploited by simple and general online algorithms not necessarily tailored to this structure.

Our first two theorems concern online problems with patterns defined in the temporal order of the accesses/requests. 
Key to both results is a geometric view where the input is mapped to points in the plane, with the two axes being the time-, and the key-value of the access/request sequence. It is natural to ask whether the approach can be adapted to \emph{purely geometric} optimization problems that take planar point sets as input. 

\paragraph{Traveling salesman problem.} TSP
is such a geometric problem. Let $X$ be a set of $n$ points in the plane (for simplicity, $X \subset [0,1]^2$) and denote by $\TSP(X)$ the (euclidean) length of the shortest tour visiting all $X$. 
Assume $X$ is in general position, i.e., no two points share an $x$- or $y$-coordinate.  Then, $X$ can be seen as a permutation by reading out the $y$-coordinates of the points in a left-to-right order, and pattern-avoidance for permutations naturally extends to the point set $X$; see \S\,\ref{sec2} for precise definitions.\footnote{This notion of geometric pattern can be used, e.g., to express the Erd\H{o}s-Szekeres theorem on monotone subsets. Connections to other types of geometric patterns are explored in~\cite[\S\,14]{Eppstein_book}.} 


A simple and well-known observation is that $\TSP(X) \in \fO(\sqrt{n})$ for all $X$, and this is best possible; take a uniform $\sqrt{n} \times \sqrt{n}$ grid spanning $[0,1]^2$. For point sets $X$ that avoid an arbitrary permutation pattern, the situation changes dramatically, as our third main result shows.  

\begin{theorem}[informal]\label{thm3}
A set of $n$ points in the unit box $[0,1]^2$ that avoids a permutation pattern of size $k$ admits a TSP solution of length $\fO_k(\log{n})$.
\end{theorem}

\Cref{thm3} implies a similar bound for related problems like \emph{euclidean minimum spanning tree} and \emph{Steiner tree}, whose optima are at most a constant factor away from $\TSP(X)$. 

\paragraph{Discussion.}

We emphasize that our results (Theorems~\ref{thm1}, \ref{thm2}, \ref{thm3}) hold for \emph{arbitrary} avoided patterns. The three problems we consider are representative of their areas (data structures, online algorithms, geometric optimization). While adaptivity to structure has been extensively studied for all three, to our knowledge, pattern-avoidance was considered before only for BSTs. 
We study all three in a similar geometric view, but the transfer of techniques is not automatic, and the problems highlight specific technical challenges. (The \emph{effect} of pattern-avoidance also differs across the problems, with the cost changing from $\fO(n\log{n})$ to $\fO(n)$ in the first, from $\fO(n)$ to $\fO(n^\varepsilon)$ in the second, and from $\fO(\sqrt{n})$ to $\fO(\log{n})$ in the third.) 

In the BST problem, in both time- and access-key-dimensions, only the ordering matters, with no meaningful concept of distance. In $k$-server, this is true for the time-dimension, but the requests are distance-sensitive (the specific distances between neighboring requests can affect the cost, the input cannot be freely ``stretched''). The TSP problem is inherently geometric, with \emph{both} dimensions distance-sensitive. 
This distinction is reflected in our solutions: for the BST problem the Guillemot-Marx decomposition can be used more or less directly; for the other two problems we develop a refined, ``distance-balanced'' decomposition (\S\,\ref{sec4}), which may be of independent interest. 

\medskip

Pattern-avoidance can be seen as a natural restriction of the input-generating process, to which algorithms can adapt for efficiency. But why does pattern-avoidance help in the first place? In online tasks it is intuitive that \emph{reducing the entropy} of 
%
the distribution for the next access/request can reduce cost. 
However, directly taking advantage of this in an online manner seems very difficult (even when the avoided pattern is known). Another intuition is that pattern-avoidance imposes a global \emph{sparsity} on the input -- this is often studied via properties of \emph{random} pattern-avoiding permutations~\cite{Scaling}. Our results, e.g., for TSP, can be seen as capturing a certain sparsity at the level of individual pattern-avoiding point sets.


 The effect of pattern-avoidance on algorithmic complexity is mysterious and fragile (in~\S\,\ref{sec3.1} we also exhibit a natural optimization problem for which pattern-avoidance does \emph{not} help). 

\paragraph{Further developments.}

 Since the preliminary version of our paper\footnote{Appeared at the STOC 2024 conference.}, a number of follow-up works have appeared. One of the motivations for studying BST access sequences is their connection to \emph{sorting}. In particular, if splay or greedy attain the $\fO(n)$ bound for serving an $n$-permutation $\tau$ that avoids a constant-sized $\pi$ (this is an important question left open by our work), then $\tau$ can be \emph{sorted} in time $\fO(n)$ (by simple insertion sort, with insertions into a BST according to the splay or greedy strategy, or by simple heapsort, using the \emph{smooth heap} data structure~\cite{smooth1, smooth2, smooth3}, whose cost matches that of greedy~\cite{smooth1}). By a result of Fredman~\cite{Fredman}, a permutation from a family of size $2^{\fO(n)}$ can be sorted with $\fO(n)$ \emph{comparisons} (in contrast to general sorting that requires $\Omega(n\log{n})$ comparisons). Fredman's argument, however, only implies the existence of a decision tree of depth $\fO(n)$ and does not give an efficient algorithm for finding it. 
 
 Theorem~\ref{thm1} shows that the decision tree of pattern-avoiding sorting can be efficiently implemented in the BST model, making linear-time sorting of such inputs a corollary of dynamic optimality.
Subsequently, Opler~\cite{Opler24} has found such a sorting algorithm with an optimal runtime of $\fO(n\log{s_\pi})$ via different means (an approach based on \emph{mergesort}).

Building on this sorting algorithm, as well as on an adaptation of the \emph{distance-balanced merge sequence} introduced in \S\,\ref{sec:dist-bal} of this work, Kozma and Opler~\cite{KozmaOpler25} obtain a further follow-up: a \emph{compact data structure} for storing any $\pi$-avoiding $n$-permutation, requiring $\fO(n\log{s_\pi})$ bits of space, $\fO(n\log{s_\pi})$ build time, and supporting $\fO(1)$ time \emph{rank} and \emph{unrank} queries. 


\paragraph{Structure of the paper.} In \S\,\ref{sec2} we review some important definitions related to permutations and pattern-avoidance. The reader may prefer to skim these and refer to \S\,\ref{sec2} when needed. In \S\,\ref{sec3} we define the three main problems we study and overview the results at a high level. In \S\,\ref{sec4} we give our distance-balanced decomposition used in the solutions. In \S\,\ref{sec:ass}, \S\,\ref{sec:k-server}, and \S\,\ref{sec:mst} we give the detailed proofs for the BST, $k$-server, and TSP problems respectively. In \S\,\ref{sec8} we conclude with open questions. 


\section{Preliminaries}
\label{sec2}


\paragraph{Permutations, matrices, and point sets.}

An $n$-\emph{permutation} $\pi$ is an ordering $\pi_1 \pi_2 \dots \pi_n$ of $[n] = \{1,2,\dots,n\}$ (as customary, we omit commas when there is no ambiguity).
Two 
alternative representations of $\pi$ (see Figs.~\ref{fig:perm-reprs-a},~\ref{fig:perm-reprs-b}) 
are (i) as an $n \times n$ \emph{permutation matrix} 
$M_\pi$ where the $\pi_i$-th entry of the $i$-th column
\footnote{We read columns from left to right and rows from bottom to top.}
is $1$ for all $i\in[n]$, and all other entries are $0$; and (ii)
as a \emph{set $P_\pi \subset \R^2$ of $n$ points} such that there is a bijection $f \colon [n] \rightarrow P_\pi$ with $f(i).x < f(j).x \Leftrightarrow i < j$ and $f(i).y < f(j).y \Leftrightarrow \pi_i < \pi_j$ for all $i, j \in [n]$, where $p.x$ and $p.y$ denote the x- and y- coordinates of point $p$.

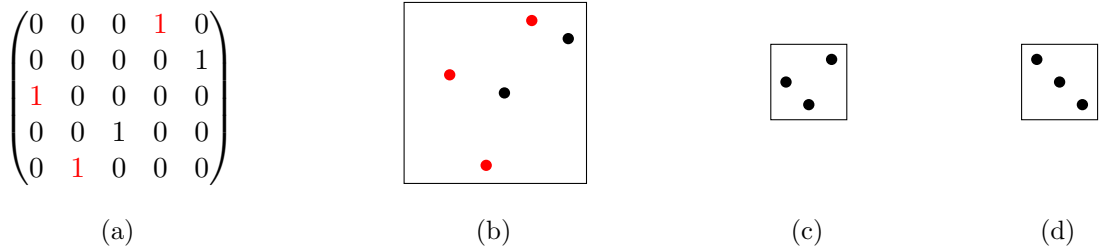
\begin{figure}[h]
	\begin{subfigure}[c]{0.3\textwidth}%
		\centering%
		\[	\begin{pmatrix}
				0&0&0&\textcolor{red}{1}&0\\
				0&0&0&0&1\\				\textcolor{red}{1}&0&0&0&0\\
				0&0&1&0&0\\
				0&\textcolor{red}{1}&0&0&0
			\end{pmatrix}
		\]
	\end{subfigure}%
	\begin{subfigure}[c]{0.3\textwidth}
 \vspace{0.1in}\ \\
		\centering
		\begin{tikzpicture}[scale=2.4]
			\draw (0,0) rectangle (1,1);
			\node[point, red] at (.25,.6) {};
			\node[point] at (.9,.8) {};
			\node[point, red] at (.45,.1) {};
			\node[point] at (.55,.5) {};
			\node[point, red] at (.7,.9) {};
		\end{tikzpicture}
	\end{subfigure}%
	\begin{subfigure}[c]{0.2\textwidth}
		\centering
		\begin{tikzpicture}[scale=1]
			\draw (0,0) rectangle (1,1);
			\node[point] at (.2,.5) {};
			\node[point] at (.5,.2) {};
			\node[point] at (.8,.8) {};
		\end{tikzpicture}
	\end{subfigure}%
	\begin{subfigure}[c]{0.2\textwidth}
		\centering
		\begin{tikzpicture}[scale=1]
			\draw (0,0) rectangle (1,1);
			\node[point] at (.2,.8) {};
			\node[point] at (.5,.5) {};
			\node[point] at (.8,.2) {};
		\end{tikzpicture}
	\end{subfigure}
	\vspace{1mm}

 \begin{subfigure}{0.3\textwidth}
    \caption{}\label{fig:perm-reprs-a}
	\end{subfigure}%
	\begin{subfigure}{0.3\textwidth}
      \caption{}\label{fig:perm-reprs-b}
	\end{subfigure}%
	\begin{subfigure}{0.2\textwidth}
      \caption{}\label{fig:perm-reprs-c}
	\end{subfigure}%
	\begin{subfigure}{0.2\textwidth}
     \caption{}\label{fig:perm-reprs-d}
	\end{subfigure}
	
	\caption{The permutation 31254 as a matrix (a) and as a point set (b); the highlighted points/ones correspond to an occurrence of the pattern 213 (c). Pattern 321~(d) is avoided by 31254.}\label{fig:perm-reprs}
\end{figure}

Observe that $n$-permutations and $n \times n$ permutation matrices are in bijection, while infinitely many different point sets correspond to a single permutation.
Note that in a point set that corresponds to a permutation, no two points can have the same x- or y-coordinate. We call this property \emph{general position}. A point set in general position represents a unique permutation. We similarly say that a sequence of reals is in \emph{general position} if its values are pairwise distinct.

\paragraph{Pattern-avoidance.}
Two sequences $x_1 x_2 \dots x_m \in \R^m$ and $y_1 y_2 \dots y_m \in \R^m$ are \emph{order-isomorphic} if for each $i, j \in [m]$, we have $x_i < x_j$ if and only if $y_i < y_j$.
We say a sequence $X$ \emph{contains} another sequence $Y$ if some subsequence (not necessarily contiguous) of $X$ is order-isomorphic to $Y$. Otherwise, we say $X$ \emph{avoids} $Y$.

We focus on the case when sequence $Y$ is a permutation; for $X$ we consider both permutations and general sequences. The definition easily extends to point sets: If $P$ and $Q$ are point sets in general position, then $P$ \emph{contains} $Q$ (or a permutation $\pi$) if the permutation corresponding to $P$ contains the permutation corresponding to $Q$ (or $\pi$). Extending the definition to permutation matrices is immediate. 

We also define pattern containment/avoidance for \emph{general} 0-1 matrices. A 0-1 matrix $M$ \emph{contains} a 0-1 matrix $P$ if $P$ can be obtained from $M$ by removing rows or columns (i.e., taking a submatrix) and turning zero or more entries from 1 to 0.
If both $M$ and $P$ are permutation matrices, we recover pattern containment for permutations. 

\paragraph{Operations on permutations.}

Let $\pi$ be a $k$-permutation. The \emph{reversal} of $\pi$ is obtained by reading $\pi$ backwards. The \emph{complement} of $\pi$ is obtained by replacing each $\pi_i$ with $k+1-\pi_i$.

Let $\pi = x_1 x_2 \dots x_k$ and $\rho = y_1 y_2 \dots y_\ell$ be permutations.
The \emph{sum} $\pi \oplus \rho$ is the permutation $x_1 x_2 \dots x_k (y_1+k) (y_2+k) \dots (y_\ell + k)$.
The \emph{skew sum} $\pi \ominus \rho$ is the permutation $(x_1+\ell)(x_2+\ell) \dots (x_k+\ell) y_1 y_2 \dots y_\ell$ (see \cref{sfig:sum,sfig:skew-sum}).

\begin{figure}
	\def\mar{0.2}
	\centering
	\begin{subfigure}{.2\textwidth}
		\centering
		\begin{tikzpicture}[scale=0.27]
			\draw (0,0) rectangle (9,9);
			\begin{scope}[shift={(0,0)}]
				\fill[hlbox] (\mar,\mar) rectangle (4-\mar,4-\mar);
				\node[point] at (1,2) {};
				\node[point] at (2,1) {};
				\node[point] at (3,3) {};
			\end{scope}
			
			\begin{scope}[shift={(4,4)}]
				\fill[hlbox] (\mar,\mar) rectangle (5-\mar,5-\mar);
				\node[point] at (1,3) {};
				\node[point] at (2,4) {};
				\node[point] at (3,1) {};
				\node[point] at (4,2) {};
			\end{scope}
		\end{tikzpicture}
		\caption{Sum}\label{sfig:sum}
	\end{subfigure}%
	\begin{subfigure}{.2\textwidth}
		\centering
		\begin{tikzpicture}[scale=0.27]
			\draw (0,0) rectangle (9,9);
			\begin{scope}[shift={(0,5)}]
				\fill[hlbox] (\mar,\mar) rectangle (4-\mar,4-\mar);
				\node[point] at (1,2) {};
				\node[point] at (2,1) {};
				\node[point] at (3,3) {};
			\end{scope}
			
			\begin{scope}[shift={(4,0)}]
				\fill[hlbox] (\mar,\mar) rectangle (5-\mar,5-\mar);
				\node[point] at (1,3) {};
				\node[point] at (2,4) {};
				\node[point] at (3,1) {};
				\node[point] at (4,2) {};
			\end{scope}
		\end{tikzpicture}
		\caption{Skew sum}\label{sfig:skew-sum}
	\end{subfigure}%
	\begin{subfigure}{.2\textwidth}
		\centering
		\begin{tikzpicture}[scale=0.27]
			\draw (0,0) rectangle (9,9);
			\begin{scope}[shift={(0,2)}]
				\fill[hlbox] (\mar,\mar) rectangle (4-\mar,4-\mar);
				\node[point] at (1,1) {};
				\node[point] at (2,2) {};
				\node[point] at (3,3) {};
			\end{scope}
			\begin{scope}[shift={(4,0)}]
				\fill[hlbox] (\mar,\mar) rectangle (2-\mar,2-\mar);
				\node[point] at (1,1) {};
			\end{scope}
			\begin{scope}[shift={(6,6)}]
				\fill[hlbox] (\mar,\mar) rectangle (3-\mar,3-\mar);
				\node[point] at (1,2) {};
				\node[point] at (2,1) {};
			\end{scope}
		\end{tikzpicture}
		\caption{Inflation}\label{sfig:infl}
	\end{subfigure}%
	\begin{subfigure}{.2\textwidth}
		\centering
		\begin{tikzpicture}[scale=0.203, pa/.style={point, black}, pb/.style={point, red}, gridline/.style={gray, very thin}]
			\draw (0.5,0.5) rectangle (12.5,12.5);
			\node[pa] at (1,2) {};
			\node[pa] at (5,3) {};
			\node[pa] at (10,1) {};
			\node[pa] at (3,4) {};
			\node[pa] at (6,6) {};
			\node[pa] at (9,5) {};
			\node[pa] at (2,9) {};
			\node[pa] at (8,8) {};
			\node[pa] at (12,7) {};
			\node[pa] at (4,10) {};
			\node[pa] at (7,11) {};
			\node[pa] at (11,12) {};
			
			\draw[gridline] (4.5,0.5) -- (4.5,12.5);
			\draw[gridline] (8.5,0.5) -- (8.5,12.5);
			\draw[gridline] (0.5,3.5) -- (12.5,3.5);
			\draw[gridline] (0.5,6.5) -- (12.5,6.5);
			\draw[gridline] (0.5,9.5) -- (12.5,9.5);
		\end{tikzpicture}
		\caption{$3 \times 4$ Grid}\label{sfig:grid}
	\end{subfigure}%
	\begin{subfigure}{.2\textwidth}
		\centering
		\begin{tikzpicture}[scale=0.203, pa/.style={point, black}, pb/.style={point, red}, gridline/.style={gray, very thin}]
			\draw (0.5,0.5) rectangle (12.5,12.5);
			\node[pa] at (4,1) {};
			\node[pa] at (8,2) {};
			\node[pa] at (12,3) {};
			\node[pa] at (3,4) {};
			\node[pa] at (7,5) {};
			\node[pa] at (11,6) {};
			\node[pa] at (2,7) {};
			\node[pa] at (6,8) {};
			\node[pa] at (10,9) {};
			\node[pa] at (1,10) {};
			\node[pa] at (5,11) {};
			\node[pa] at (9,12) {};
			
			\draw[gridline] (4.5,0.5) -- (4.5,12.5);
			\draw[gridline] (8.5,0.5) -- (8.5,12.5);
			\draw[gridline] (0.5,3.5) -- (12.5,3.5);        
			\draw[gridline] (0.5,6.5) -- (12.5,6.5);    
			\draw[gridline] (0.5,9.5) -- (12.5,9.5);   
		\end{tikzpicture}
		\caption{Canonical grid}\label{sfig:can-grid}
	\end{subfigure}%
	\caption{(a) Sum- and (b) skew sum of permutations 213 and 3412; (c) the inflation of 213 by 123, 1, and 21; 
    (d) a $3 \times 4$ grid permutation, and (e) canonical grid.}\label{fig:sum-infl-merge}
\end{figure}
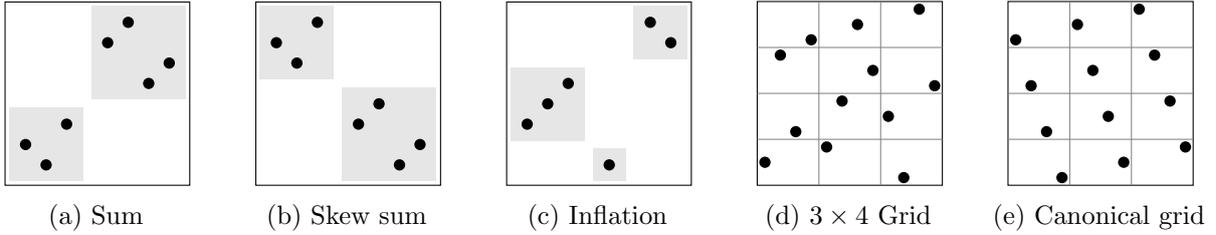

%
%

A \emph{shift by $t$} of a permutation $\pi$ is the sequence $\pi_1+t, \dots, \pi_k+t$.
The \emph{inflation} of $\pi$ by permutations $\rho_1, \rho_2, \dots, \rho_k$ is the unique permutation obtained by replacing each entry $\pi_i$ with a shift of $\rho_i$ by~$t_i$, for an appropriate integer sequence $t_1,\dots,t_k$ that is order-isomorphic to $\pi$ (see \cref{sfig:infl}). In the point set view, this amounts to replacing each point in $\pi$ by a small copy of $\rho_i$. Observe that the sum (skew sum) of permutations $\pi$ and $\rho$ is the inflation of 12 (21) by $\pi$ and $\rho$.


\paragraph{Special permutations.}

We call the unique 1-permutation \emph{trivial} and all other permutations \emph{non-trivial}.
The \emph{identity} or \emph{increasing permutation} of length $k$ is $I_k = 12\dots k$. The \emph{decreasing permutation} of length $k$ is $D_k = k (k-1) \dots 1$.

A \emph{$k \times \ell$ grid permutation} is a $k \ell$-permutation that has a corresponding point set $P \subset [0,1]^2$ such that every ``cell'' $C = [(i-1)/k, i/k] \times [(j-1)/\ell, j/\ell]$ for $i \in [k]$, $j \in [\ell]$ contains precisely one point of $P$ (see \cref{sfig:grid}). Observe that a $k \times k$ grid permutation contains all $k$-permutations. 
The \emph{canonical $k \times \ell$ grid permutation} corresponds to the point set with coordinates $((i-1)\ell+(\ell-j+1), (j-1)k + i)$ for $i\in [k], j\in [\ell]$ (see \cref{sfig:can-grid}).
Each ``row'' of this permutation is increasing, and each ``column'' is decreasing.

\paragraph{Permutation classes.}

A \emph{permutation class} is a set of permutations that is hereditary, i.e., closed under containment. It is easy to see that each permutation class can be characterized by \emph{avoidance} of a (not necessarily finite) set of permutations. Given a set $\Pi$ of permutations, we write $\Av(\Pi)$ for the class of permutations avoiding each $\pi \in \Pi$.
A permutation class $\Av(\pi)$ that avoids a single permutation $\pi$ is called a \emph{principal} 
permutation class.

We describe some simple and important permutation classes. A permutation is \emph{$k$-increasing} (\emph{$k$-decreasing}) if it avoids the permutation $D_{k+1}$ ($I_{k+1}$).
Observe that a permutation is $k$-increasing ($k$-decreasing) if and only if it is formed by interleaving 
$k$ increasing (decreasing) sequences.

A permutation is \emph{separable} if it is the trivial 1-permutation, or the sum or skew-sum of two separable permutations (for example, the permutations depicted in \cref{sfig:sum,sfig:skew-sum,sfig:infl} are separable).
The set of separable permutations is known to equal $\Av(3142,2413)$.

Yet another characterization of separable permutations is as the set of inflations of $12$ or $21$ with two separable (or trivial) permutations. This definition can be generalized: A permutation is \emph{$k$-separable} if it is the trivial 1-permutation or the inflation of a permutation of length at most $k$ with $k$-separable permutations.

The classes $\Av(132)$, $\Av(213)$, $\Av(231)$, and $\Av(312)$ are important special cases, due to their recursive structure (see \cref{sec:k-server-231-ub}, Fig.~\ref{fig:231-av-structure}). Observe that all four are subclasses of the separable permutations. As the problems we study are invariant under reversal and complement, we often consider only one of the four classes, namely $\Av(231)$.

\paragraph{Füredi-Hajnal and Stanley-Wilf limits.} 
Given a permutation $\pi$, let $\ex_\pi(n)$ be the maximum number of ones in an $n \times n$ 0-1 matrix that avoids $\pi$.
Marcus and Tardos~\cite{MarcusTardos} proved that $\ex_\pi(n) \in \fO(n)$ for each fixed $\pi$ (the \emph{Füredi-Hajnal conjecture}~\cite{FurediHajnal}). The \emph{Füredi-Hajnal limit} $c_\pi$ of $\pi$ is the constant hidden in the $\fO$-notation, i.e.,
\begin{align*}
	c_\pi = \lim_{n \rightarrow \infty} \tfrac 1n \ex_\pi(n).
\end{align*}

Cibulka~\cite{Cibulka2009} observed that, due to superadditivity of $\ex_\pi(\cdot)$, the limit $c_\pi$ indeed exists and $\ex_\pi(n) \le c_\pi \cdot n$ for all $n$. The extension to non-square matrices is immediate:

\begin{lemma}\label{p:marcus-tardos-with-limit-non-square}
	Every $m \times n$ 0-1 matrix with strictly more than $c_\pi \cdot \max(m,n)$ entries contains $\pi$.
\end{lemma}

The \emph{Stanley-Wilf conjecture} states that the number of $n$-permutations avoiding a fixed pattern is single-exponential in $n$ (i.e., $2^{\fO(n)}$ of the $2^{\Theta(n\log{n})}$ $n$-permutations).
Let~$\Av_n(\pi)$ be the set of $n$-permutations that avoid the pattern $\pi$. The \emph{Stanley-Wilf limit} $s_\pi$ is defined as
\[
	s_\pi = \lim_{n \rightarrow \infty} \sqrt[n]{|\Av_n(\pi)|}.
\]

As observed by Klazar~\cite{Klazar}, the Füredi-Hajnal conjecture implies the Stanley-Wilf conjecture, and hence the existence of Stanley-Wilf limits.

Fox~\cite{jfox} showed that $c_\pi \in 2^{\fO(k)}$ for all $k$-permutations $\pi$ (the Marcus-Tardos proof implied $c_\pi \in 2^{\fO(k\log{k})}$), and that $c_\pi \in 2^{\Omega(k^{1/4})}$ for \emph{some} $k$-permutations $\pi$ (an unpublished improvement to $2^{\Omega(k^{1/2})}$ is often cited, e.g., in \cite{ChalermsookPettieY, cibulka2017better}). Cibulka~\cite{Cibulka2009} showed that $s_\pi$ and $c_\pi$ are polynomially related. More precisely, $s_\pi \in \Omega({c_\pi^{2/9}}) \cap \fO({c_\pi^2})$.
A finer characterization of $c_\pi$ and $s_\pi$, in particular, of the cases when they are polynomial in $k$, is still lacking. 

\paragraph{Merge sequences and twin-width.}

A \emph{rectangle family} $\cR$ is a set of axis-parallel rectangles.
Note that a single point is a special kind of a rectangle and thus, any point set can be interpreted as a rectangle family.
Let $S, T$ be two different rectangles of a rectangle family $\cR$.
The \emph{merging} of rectangles $S, T$ is an operation that replaces $S, T$ by the smallest axis-parallel rectangle enclosing $S \cup T$, i.e., their bounding box.

A \emph{merge sequence} of a point set $P \subset \mathbb{R}^2$ of size $n$ is a sequence $\cR_1, \cR_2, \dots ,\allowbreak\cR_n$ of rectangle families where
$\cR_1 = P$ is the original point set,  $\cR_n$ contains a single rectangle, and each $\cR_{i+1}$ is obtained by merging two rectangles of $\cR_i$. Notice that each $\cR_i$ consists of exactly $n-i+1$ rectangles.

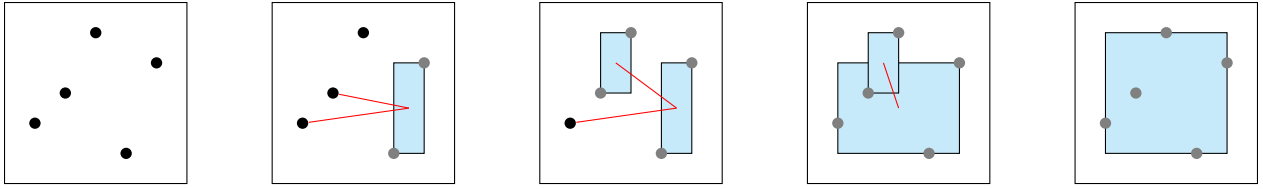
\begin{figure}
	\def\xs{{1,2,3,4,5}}
	\def\ys{{2,3,5,1,4}}
	\def\tscale{.4}
	
	\tikzset{
		disabled point/.style={point, gray},
		p0/.style={disabled point},
		p1/.style={disabled point},
		p2/.style={disabled point},
		p3/.style={disabled point},
		p4/.style={disabled point},
		red edge/.style={red}
	}
	
	\newcommand{\drawPoints}{
		\draw (0,0) rectangle (6,6);
		\foreach \i in {0,1,2,3,4} {
			\node[p\i] (p\i) at (\xs[\i], \ys[\i]) {};
		}
	}
	\newcommand{\msrect}[4]{
		\draw[msrect] (\xs[#1],\ys[#3]) rectangle (\xs[#2],\ys[#4]);
	}
	
	\begin{tikzpicture}[scale=\tscale, p0/.style={point}, p1/.style={point}, p2/.style={point}, p3/.style={point}, p4/.style={point}]
		\drawPoints
	\end{tikzpicture}
	\hfill
	\begin{tikzpicture}[scale=\tscale, p0/.style={point}, p1/.style={point}, p2/.style={point}]
		\msrect{3}{4}{3}{4}
		\drawPoints
		\draw[red edge] (p0) -- (4.5, 2.5);
		\draw[red edge] (p1) -- (4.5, 2.5);
	\end{tikzpicture}
	\hfill
	\begin{tikzpicture}[scale=\tscale, p0/.style={point}]
		\msrect{3}{4}{3}{4}
		\msrect{1}{2}{1}{2}
		\drawPoints
		\draw[red edge] (p0) -- (4.5, 2.5);
		\draw[red edge] (2.5,4) -- (4.5, 2.5);
	\end{tikzpicture}
	\hfill
	\begin{tikzpicture}[scale=\tscale]
		\msrect{0}{4}{3}{4}
		\msrect{1}{2}{1}{2}
		\draw[red edge] (2.5,4) -- (3.0, 2.5);
		\drawPoints
	\end{tikzpicture}
	\hfill
	\begin{tikzpicture}[scale=\tscale]
		\msrect{0}{4}{3}{2}
		\drawPoints
	\end{tikzpicture}
	\caption{A 3-wide merge sequence of 23514, along with its red graphs. Black points are also degenerate rectangles. 
 }\label{fig:merge-seq}
\end{figure}







%
Two rectangles $S$ and $T$ are called \emph{homogeneous} if their projections onto both $x$- and $y$-axis are disjoint.
Given a rectangle family $\cR_i$, we consider an auxiliary graph~$R_i$, called \emph{red graph}, where the rectangles of $\cR_i$ are vertices, and there is a (red) edge between every pair of distinct \emph{non-homogeneous} rectangles $S,T$.
\Cref{fig:merge-seq} shows an example merge sequence and its red graph.

Let $d$ be a positive integer.
We say that a merge sequence $\cR_1, \dots ,\cR_n$ is \emph{$d$-wide} if $\max_{i \in [n]}\Delta(R_i) < d$, i.e., the maximum degree over all red graphs associated to this sequence is strictly less than $d$.
The \emph{twin-width} $\tww(P)$ of a point set $P$ is then the minimum integer $d$ such that there exists a $d$-wide merge sequence of $P$.\footnote{Originally, Guillemot and Marx define width slightly differently, separately considering width on the horizontal and vertical projections. The current definition became more common in the literature on twin-width, and differs from the original by a factor of at most two.}
For a permutation $\pi$, we let $\tww(\pi) = \tww(P_\pi)$ for any point set $P_\pi$ corresponding to $\pi$. 

We make some simple observations about the behavior of twin-width.
\begin{observation}\label{obs:tww}~
	\begin{itemize}
		\item The trivial $1$-permutation has twin-width $1$.
		\item If $\pi$ contains $\rho$, then $\tww(\rho) \le \tww(\pi)$.
		\item $\tww(\pi \oplus \rho) = \tww(\pi \ominus \rho) = \max( \tww(\pi), \tww(\rho) )$. In particular, the separable permutations are precisely the permutations of twin-width 1.
		\item More generally, if $\sigma$ is the inflation of $\pi$ by $\rho_1, \rho_2, \dots, \rho_k$, then
			\[\tww(\sigma) = \max( \tww(\pi), \tww(\rho_1), \tww(\rho_2), \dots, \tww(\rho_k) ).\] In particular, a $k$-separable permutation has twin-width at most $k$.

	\end{itemize}
\end{observation}

We refer to Guillemot and Marx \cite[Lemma~3.2, Propositions~3.2 and~3.3]{GM_PPM} for proofs of \cref{obs:tww}.\footnote{Note that Guillemot and Marx use a slightly different definition of twin-width but their arguments readily extend to our version.}

\begin{lemma}\label{lem:canonical-grid-tww}
	The canonical $k \times \ell$ grid permutation has twin-width $\min(k,\ell)$.
\end{lemma}
\begin{proof}
	For the lower bound, observe that after the very first merge, the obtained rectangle is already non-homogeneous to at least $\min(k,\ell)-1$ points.
	
	For the upper bound, suppose $k \le \ell$. Merge the lowest two points in a column~$i$ into a rectangle $R_i$, from left to right. Then merge each third lowest point in column~$i$ into $R_i$, again from left to right, and so on. The rectangles stay homogeneous with all remaining points, so the degree of the red graph is always at most $k-1$. After each column is a single rectangle, merge these $k$ rectangles in an arbitrary order.
\end{proof}


The importance of twin-width for our work is derived from the following theorem.

\begin{theorem}[Guillemot and Marx~\cite{GM_PPM}]\label{thm:tw}
	A point set that avoids a pattern $\pi$ has twin-width at most $\fO(c_\pi)$.
\end{theorem}

Guillemot and Marx~\cite{GM_PPM} originally stated the slightly weaker bound $2^{\fO(k\log{k})}$ for a $k$-permutation pattern $\pi$, but the bound of $\fO(c_\pi)$ is implicit in their proof. The upper bound of Fox~\cite{jfox} on $c_\pi$ then implies an upper bound of $2^{\fO(k)}$ on the twin-width. In \S\,\ref{sec4} we describe a decomposition with additional properties that also implies this claim. 

The merge sequence implied by Theorem~\ref{thm:tw} can be found in time $\fO(n)$, if the points are accessible in sorted order by x- and y- coordinates. If this is not the case, and we view the coordinates as arbitrary comparable elements, an initial sorting step of $\fO(n\log{n})$ time is necessary.\footnote{This term is for the cost of \emph{finding the decomposition}, and does not affect the cost of the solution (e.g., for BST) that we compute with its help.}



The bound on twin-width in terms of $c_{\pi}$ cannot be asymptotically improved: For $\pi = D_k$, we have $c_\pi = 2 \cdot (k-1)$ (Füredi and Hajnal~\cite{FurediHajnal}); on the other hand, the canonical $(k-1) \times n$ grid permutation is $\pi$-avoiding and has twin-width $k-1$. Whether for \emph{some} pattern $\pi$, the twin-width of \emph{all} $\pi$-avoiding permutations can be significantly smaller than $c_\pi$  is an intriguing open question. 

\Cref{fig:classes} shows an overview of several relevant permutation classes and their inclusions.



\begin{figure}
	\centering
	\small
	\begin{tikzpicture}[cls/.style={align=center, draw}, x=20mm, y=15mm, scale=0.75]
		\node[cls] (inc)  at (1,0) {Increasing\\or $\Av(21)$};
		\node[cls] (231)  at (3,0) {$\Av(231)$};
		\node[cls] (kinc) at (0,1.5) {$k$-increasing\\or $\Av(k(k-1)\dots1)$};
		\node[cls] (sep)  at (2,1) {separable\\or $\tww = 1$};
		\node[cls] (ksep) at (2,2) {$k$-separable};
		\node[cls] (piav) at (1,3) {$\Av(\pi)$ \small(for some $\pi$)};
		
		\draw (inc) -- (kinc) -- (piav) ; 
		\draw (inc) -- (sep) -- (ksep) -- (piav);
		\draw (231) -- (sep);
	\end{tikzpicture}
	\caption{A hierarchy of important permutation classes.\label{fig:classes}
 }
\end{figure}
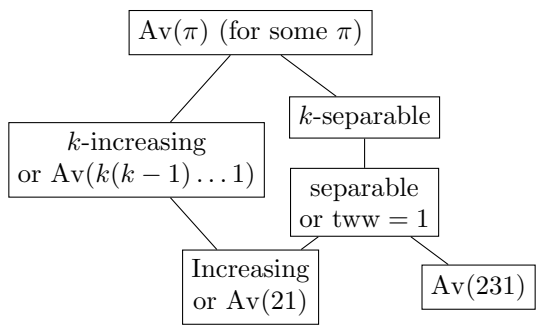

\section{Overview of results}
\label{sec3}

In this section we present our main results and sketch some of the proofs, deferring the details to the later sections. 

\subsection{BST and arborally satisfied superset}\label{sec3.1}

Let $T$ be a binary search tree (BST) with nodes identified with the elements of $[n]$ with the usual search tree property~\cite{Knuth3}. \emph{Serving an access} $x \in [n]$ in $T$ means \emph{visiting a subtree} $T'$ of $T$ such that $T'$ contains both the root of $T$ and the node $x$. Then, $T'$ is replaced with another BST $T''$ on the same node set as $T'$ (all subtrees of $T$ hanging off $T'$ are linked in the unique location to $T''$ given by the search tree property). We can think of replacing $T'$ by $T''$ as \emph{re-arranging the tree}, in preparation for future accesses.

The cost of the access is $|T'| = |T''|$, i.e., the number of nodes \emph{touched}. Note that $T'$ necessarily contains the \emph{search path} of $x$ in $T$. This common formulation of the BST model~\cite{DHIKP, LT19} allows us to ignore individual pointer moves and rotations. It is well-known that any BST $T'$ can be transformed into any 
$T''$ (on the same nodes) with $O(|T'|)$ rotations and pointer moves~\cite{STT}), therefore the model is equivalent (up to a small constant factor) with most other reasonable models (see \cref{fig:bst-a} for illustration).

\emph{Serving a sequence} $X = (x_1,\dots,x_m) \in [n]^m$, starting from an initial tree $T_0$, means serving access $x_i$ in tree $T_{i-1}$, and replacing subtree $T'_{i}$ of $T_{i-1}$ by $T''_{i}$, resulting in tree $T_i$, for all $i \in [n]$. The total cost of serving $X$ is $|T'_1| + \cdots + |T'_n|$.

\begin{figure}[h]
\begin{subfigure}[T]{.65\textwidth}  
\vskip 0pt
		\includegraphics[scale=0.25]{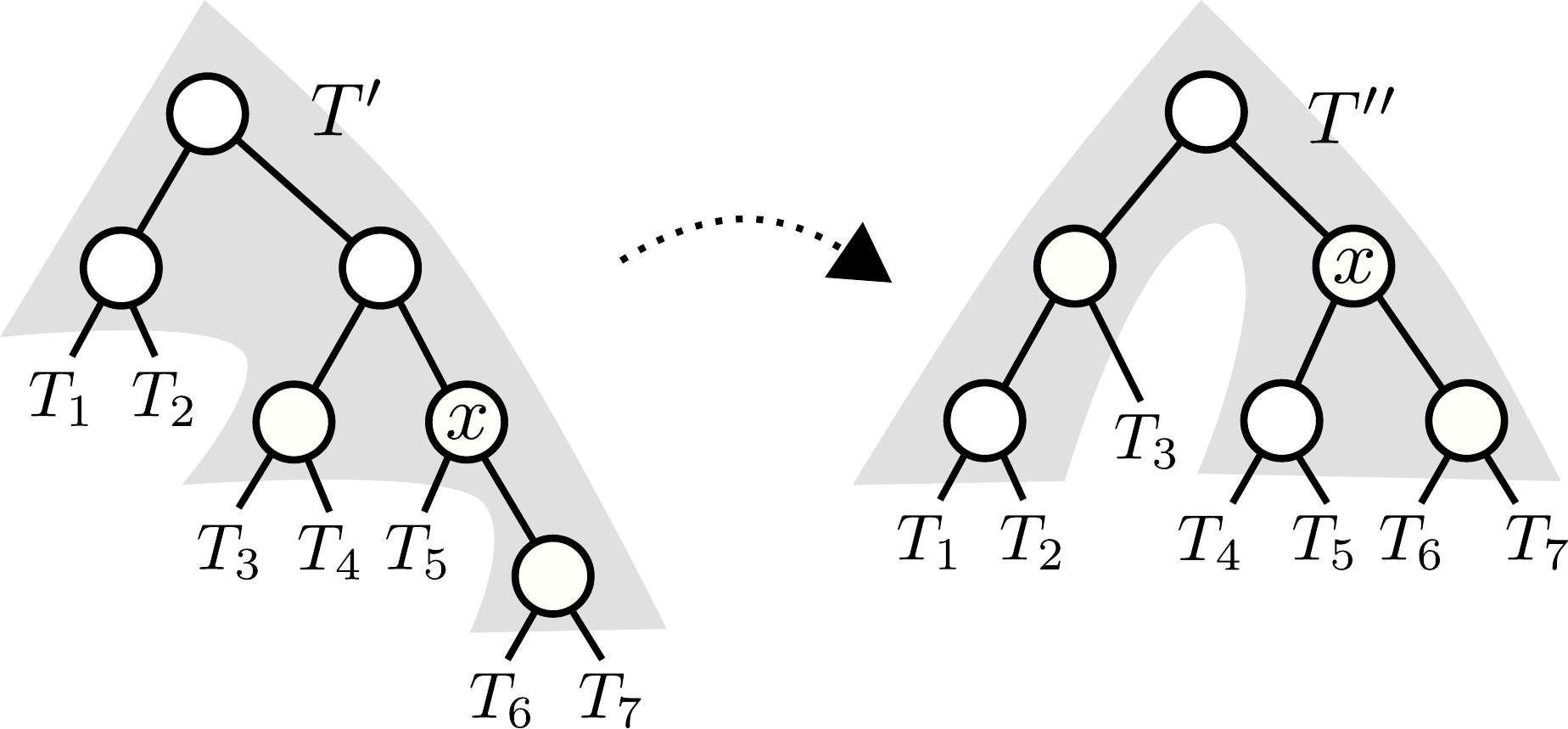}
  \end{subfigure}\hfill
\begin{subfigure}[T]{.25\textwidth} 
\vskip 0.2in
		\includegraphics[scale=0.25]{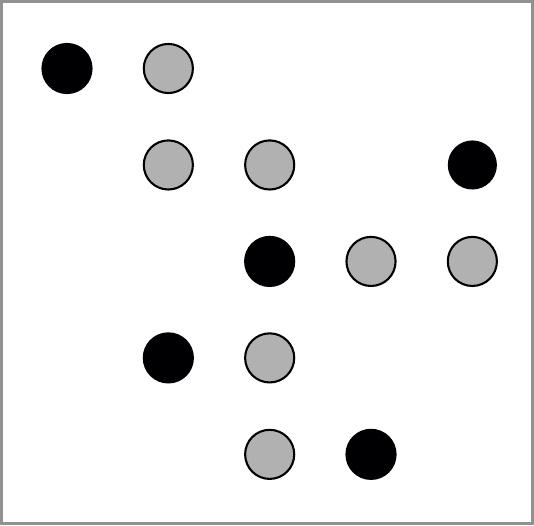}
  \end{subfigure}~~~~~
	\vspace{-4mm}

 \begin{subfigure}{0.65\textwidth}
    \caption{BST access}\label{fig:bst-a}
	\end{subfigure}%
	\begin{subfigure}{0.25\textwidth}
      \caption{satisfied superset}\label{fig:bst-b}
	\end{subfigure}%
	\caption{(a) Accessing node $x$ in a BST with subtree $T'$ re-arranged into $T''$ (shown in grey) at cost 6 and subtrees $T_1, \dots, T_7$ unchanged. (b) A satisfied superset for access sequence $5,2,3,1,4$ (black points denote accesses).\label{fig:bst}}
	\end{figure}

\vspace{-0.2in}

\paragraph{Arborally satisfied superset.} Demaine, Harmon, Iacono, Kane, and P\v{a}tra\c{s}cu~\cite{DHIKP}
give an elegant geometric characterization of the BST problem. A point set $P$ is called \emph{(arborally) satisfied}, if for any two points $p,q \in P$ one of the following holds: (a) $p.x = q.x$ or $p.y = q.y$, or (b) the axis-parallel rectangle $R$ with corners $p,q$ contains some point in $P \setminus \{p,q\}$ (possibly on the boundary of $R$). The \emph{(arborally) satisfied superset} problem asks to find, given a point set $P$, the smallest point set $P' \supseteq P$ that is arborally satisfied (\cref{fig:bst-b}).

The following theorem shows the equivalence between this geometric problem and the BST model defined earlier. 

\begin{theorem}[\cite{DHIKP}]
	Let $X \in [n]^n$ be an $n$-permutation access sequence and let $P_X$ be a corresponding point set.
	Then $X$ can be served at cost $\mathcal{C}$ if and only if $P_X$ admits a satisfied superset of size $\mathcal{C}$.
\end{theorem}

 More strongly, the satisfied superset can be assumed to consist of points aligned vertically and horizontally with $P_X$ and points on the $i$-th vertical line in the solution exactly correspond to the nodes $T'_{i}$ touched during the $i$-th access.\footnote{To be precise, this equivalence assumes a certain initial tree $T_0$. A different initial tree, however, may only add $\fO(n)$ to the cost, as an offline algorithm may start by rotating into an arbitrary initial tree.}

In the following we focus on the satisfied superset problem, and denote by $\OPT(P)$ the size of the smallest satisfied superset of $P$. We now state our main result. 
\begin{restatable}{theorem}{restateASSTwinWidth}\label{p:ass-tww}
	Let $P$ be a set of $n$ points in general position, of twin-width $d$. Then $\OPT(P) \in \fO( n d^2)$.
\end{restatable}

By Theorems~\ref{thm:tw} and \ref{p:ass-tww} it follows that if $P$ avoids a permutation $\pi$, then $\OPT(P) \in \fO(n \cdot c_{\pi}^2)$, implying Theorem~\ref{thm1} (recall that $c_\pi \in 2^{\fO(|\pi|)}$).
From the coding argument in \S\,\ref{sec1}, the lower bound $\OPT (P) \in \Omega(\log{|\Av_n(\pi)|}) \subseteq \Omega(n \cdot \log{c_\pi})$ follows. Sharpening the bounds in terms of $c_\pi$ or $|\pi|$ for all $\pi$ 
is a challenging open question. 
Previously, linear bounds on $\OPT(P)$ for pattern-avoiding $P$ were known only in special cases, e.g., 
when $P$ is $k$-increasing, $k$-decreasing, or $k$-separable~\cite{FOCS15, GoyalGupta, Isaac18, BSTdecomp}. Curiously, these cases are all \emph{low twin-width}, $\tww(P) \in \fO(k)$, but earlier studies did not use this fact.

\paragraph{Proof sketch.} In the following we sketch the proof of Theorem~\ref{p:ass-tww}, 
deferring the full details 
to \cref{sec:ass}. At a high level, we use a merge sequence $\cR_1, \cR_2, \dots, \cR_n$ of $P$ of width $d$ to construct a satisfied superset $P' \supseteq P$ (initially $P'=P$).

We maintain the invariant $(\mathtt{I_1})$ that $P' \cap Q$ is satisfied for each $Q \in \cR$, for the current set of rectangles $\cR$. This is trivially true initially for $\cR = \cR_1$, and as $\cR_n$ consists of a single rectangle containing the entire point set $P$, it is implied by $(\mathtt{I_1})$ that we obtain a valid solution $P'$.

We also maintain two technical invariants: ($\mathtt{I_2}$) states that $(Q \cup Q') \cap P'$ is satisfied for each non-homogeneous pair $Q, Q' \in \cR$, and $(\mathtt{I_3})$ states that for each $Q \in \cR$, the point set $P'$ contains the intersections of $Q$ with the grid formed by extending every side of every rectangle in $\cR$. Notice a subtletly here: as the rectangles grow, invariant ($\mathtt{I_3}$) applies to larger areas, but at the same time, $\cR$ shrinks, so the invariant becomes less stringent. 

The key step is merging two rectangles $Q_1, Q_2 \in \cR$ into a new rectangle $Q$ while maintaining the invariants $(\mathtt{I_1})$, $(\mathtt{I_2})$, $(\mathtt{I_3})$. We can achieve this by adding to $P'$ all points inside $Q$ of the grid induced by sides of rectangles in $\cR$ (including $Q_1,Q_2$);  
see Fig.~\ref{fig:ass-constr} on page~\pageref{fig:ass-constr}. 
The validity of the invariants can be verified through a case-analysis. 

Observe that $Q$ can ``see'' (is non-homogeneous with) at most $d$ rectangles (including itself), so at most $2d$ side-extensions intersect $Q$. Therefore, we add at most $\fO(d^2)$ points in every step, for a total of $\fO(nd^2)$ during the entire merge sequence. 
We defer the detailed proof to \S\,\ref{sec:ass}.

\subsubsection{Sparse Manhattan network}\label{sec311}

We briefly mention a connection between the BST/satisfied superset problem, and a well-studied network design task. 

If $P$ is a set of points, then we call $x, y \in P$ \emph{connected} and write $x \connected{P} y$ if there is a monotone path $T$ connecting $x$ and $y$ that consists of axis-parallel line segments, and each corner of $T$ is contained in $P$. For two sets $X, Y \subseteq P$ we write $X \connected{P} Y$ if $x \connected{P} y$ for all $x \in X$, $y \in Y$.
The following equivalence was observed by Harmon.

\begin{observation}[\cite{Harmon}]
A set of points $P$ is \emph{arborally satisfied} if and only if 
$P \connected{P} P$.
\end{observation}

In the satisfied superset problem we ask, given a point set $X$, for a point set $Y \supseteq X$ with $Y \connected{Y} Y$. A relaxation is to find a set $Y \supseteq X$ with $X \connected{Y} X$, that is, to connect all pairs of the input point set, but not necessarily the pairs involving newly added points. It is natural to ask for the smallest set $Y$ with this property. We call this the \emph{sparse Manhattan network} (\texttt{sparseMN}) problem, and denote its optimum as $\MN(X)$ (\cref{fig:mn-a}).
 
Let $P$ be a set of $n$ points in the plane. It follows from the definition that $\MN(P) \leq \OPT(P)$; in fact, $\MN(P)$ is equivalent with the well-known \emph{independent rectangle lower bound}~\cite{Harmon, DHIKP} (this is the best lower bound known for $\OPT(P)$, subsuming earlier bounds by Wilber~\cite{Wilber}). 
A central conjecture is that $\MN(P) \in \Theta(\OPT(P))$~\cite{DHIKP}. Its importance for dynamic optimality is given by the fact that for $\MN(P)$ a polynomial-time $2$-approximation is known, whereas for $\OPT(P)$ the best known approximation ratio is $\fO(\log\log{n})$.

The quantity $\MN(P)$ has been studied in computational geometry~\cite{GudmundEM,Rectang}; it is known that in the worst case, $\MN(P) \in \Theta(n\log{n})$. 
From Theorem~\ref{p:ass-tww} it immediately follows that $\MN(P) \in \fO_k(n)$ whenever $P$ avoids a $k$-permutation. Such a bound was known before~\cite{FOCS15}, interestingly, by a different algorithm (geometric sweepline) than the one implicit in our current work. 

\begin{figure}[h]
	\begin{subfigure}[T]{.5\textwidth}  
		\centering\includegraphics[scale=0.13]{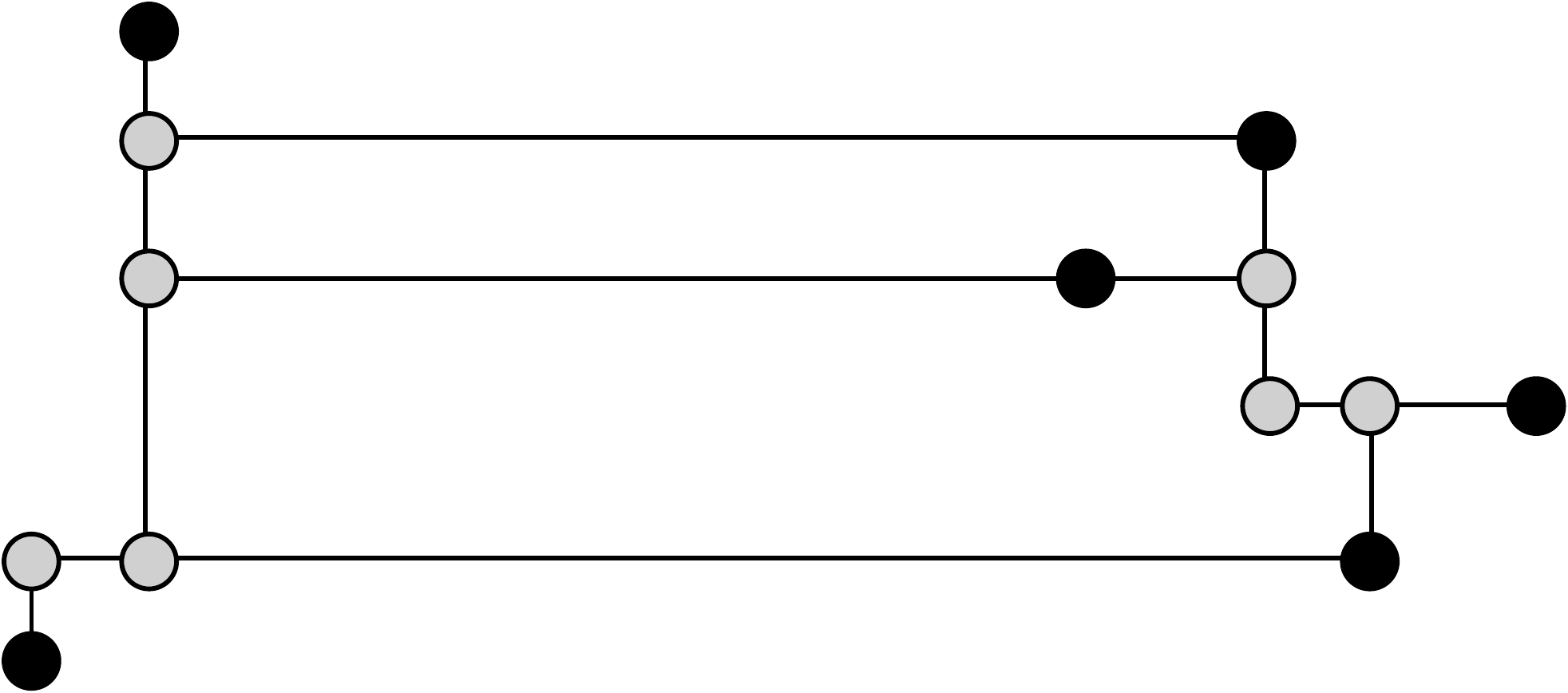}
		\caption{Sparse Manhattan network $\MN(X)$}\label{fig:mn-a}
	\end{subfigure}%
	\begin{subfigure}[T]{.5\textwidth} 
		\centering\includegraphics[scale=0.13]{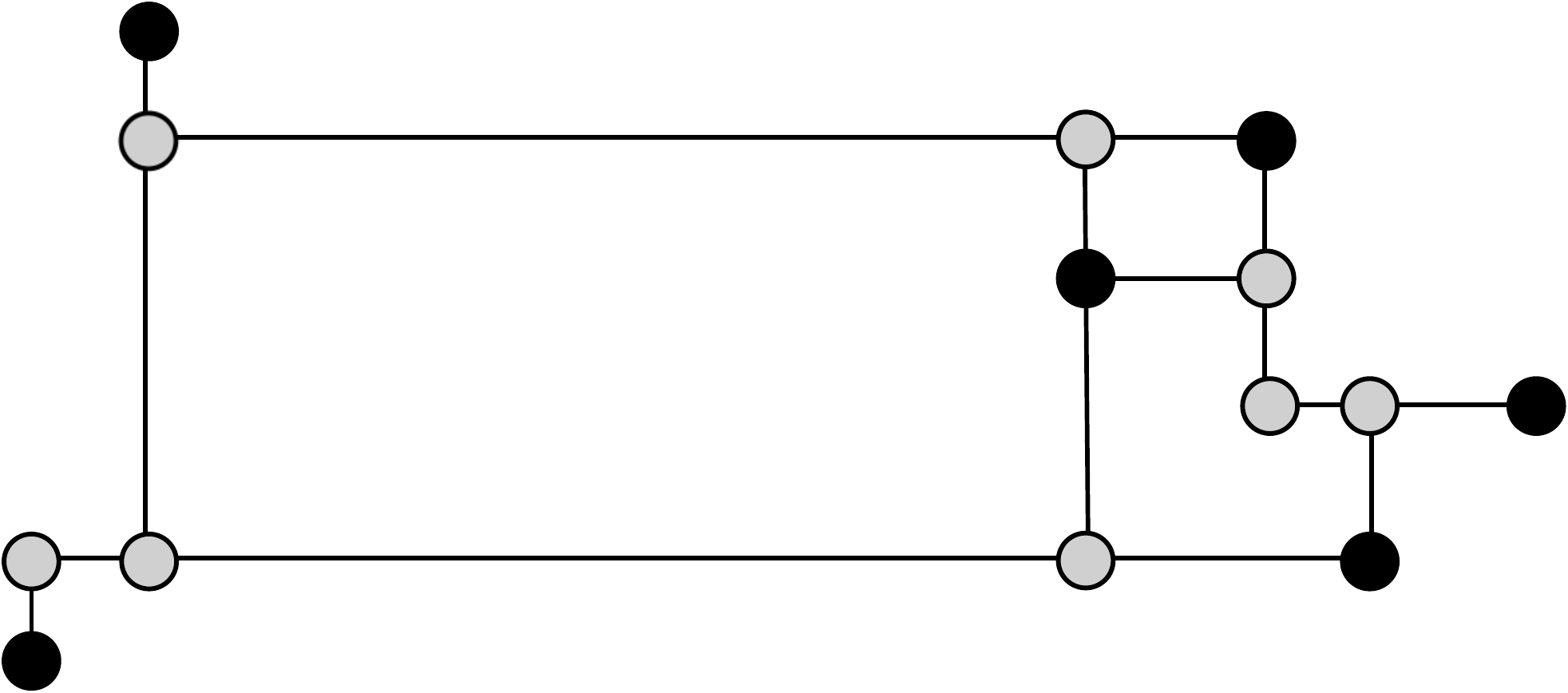}
		\caption{Small Manhattan network $\SM(X)$}\label{fig:mn-b}
	\end{subfigure}
	\caption{Manhattan networks on input $X$ (black dots) and newly added points (grey). \label{fig:mn}}
\end{figure}

\vspace{-0.2in}

\subsubsection{Small Manhattan network}

A closely related network design task is the \emph{small Manhattan network} (\texttt{smallMN}) problem. Whereas \texttt{sparseMN} aims to minimize the number of corner points, in \texttt{smallMN} the goal is to minimize the total length of the network. 

More precisely, given a set $P$ of $n$ points in the plane, we look for a collection $N$ of axis-parallel line segments so that each pair of points in $P$ can be connected by a monotone path contained in the union of the segments in $N$. The task is to minimize the total length of the segments in $N$. 
The problem is well studied~\cite{GudmundssonLN, ChinGuoSun} and known to admit a $2$-approximation~\cite{Chepoi, GuoSunZhu}.

Surprisingly, despite the similarity of \texttt{smallMN} to the other problems studied in the paper, in this case, pattern-avoidance does not lead to an asymptotic improvement. 

Given a set $P$ of $n$ points in $[0,1]^2$, let $\SM(P)$ denote the length of the smallest MN for $P$ (see \cref{fig:mn-b}). It is easy to see that $\SM(P) \in \fO(n)$: Extend each point $p \in P$ by horizontal and vertical lines until the boundaries of the box $[0,1]^2$. The resulting grid contains monotone paths between any pair $p,q \in P$, and is of total length $2n$.
The next observation indicates that \texttt{smallMN} is unaffected by non-trivial single-pattern avoidance.
\begin{observation}
	For each $n \in \N_+$, there is a $321$-avoiding set $P \subset [0,1]^2$ of $2n$ points in general position where any Manhattan network of $P$ has length at least $n$.
\end{observation}
\begin{proof}
	For $i \in [n]$, let $p_i = (\frac{i}{2n^2}, \frac{2i-2}{2n})$ and $q_i = (\frac{2n^2-n+i}{2n^2}, \frac{2i-1}{2n})$. Let $P = \{ p_i \mid i \in [n] \} \cup \{ q_i \mid i \in [n] \}$. See \cref{sfig:sm-lb-321} for a sketch. Clearly, $P$ is in general position and 321-avoiding.
	
	Each pair $p_i, q_i$ has to be connected by a path of length $\frac{1}{2n} + \frac{2n^2-n}{2n^2} = 1$. These paths cannot share common line segments (or even intersect), hence $\SM(P) \geq n$.
\end{proof}

Perturbing differently, one can make the instance avoid 231 or its symmetries (\cref{sfig:sm-lb-132}). 

\begin{figure}[h]
	\begin{subfigure}[T]{.5\textwidth}  
		\centering\includegraphics[scale=0.10]{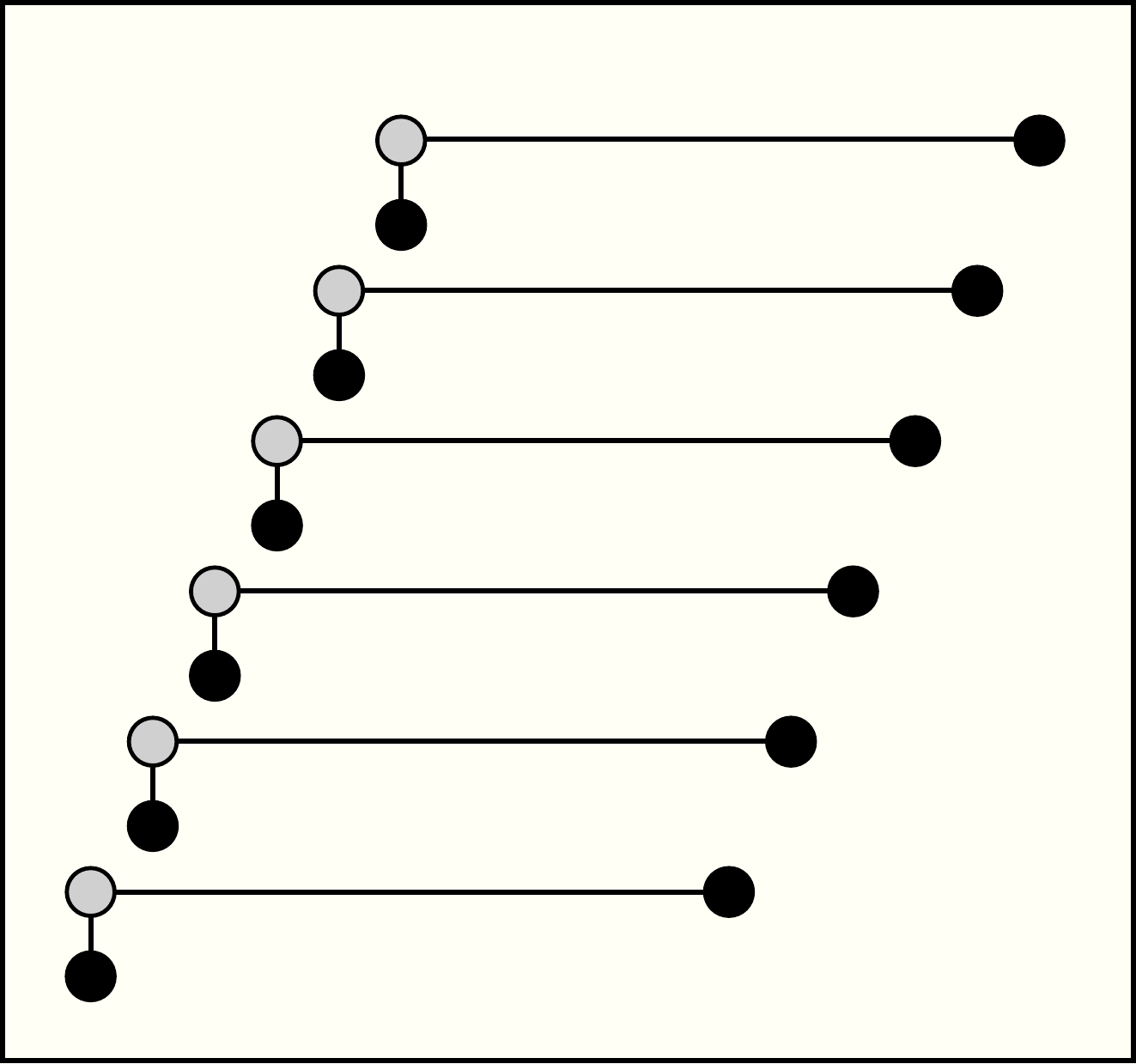}
		\caption{321-avoiding}\label{sfig:sm-lb-321}
	\end{subfigure}%
	\begin{subfigure}[T]{.5\textwidth} 
		\centering\includegraphics[scale=0.10]{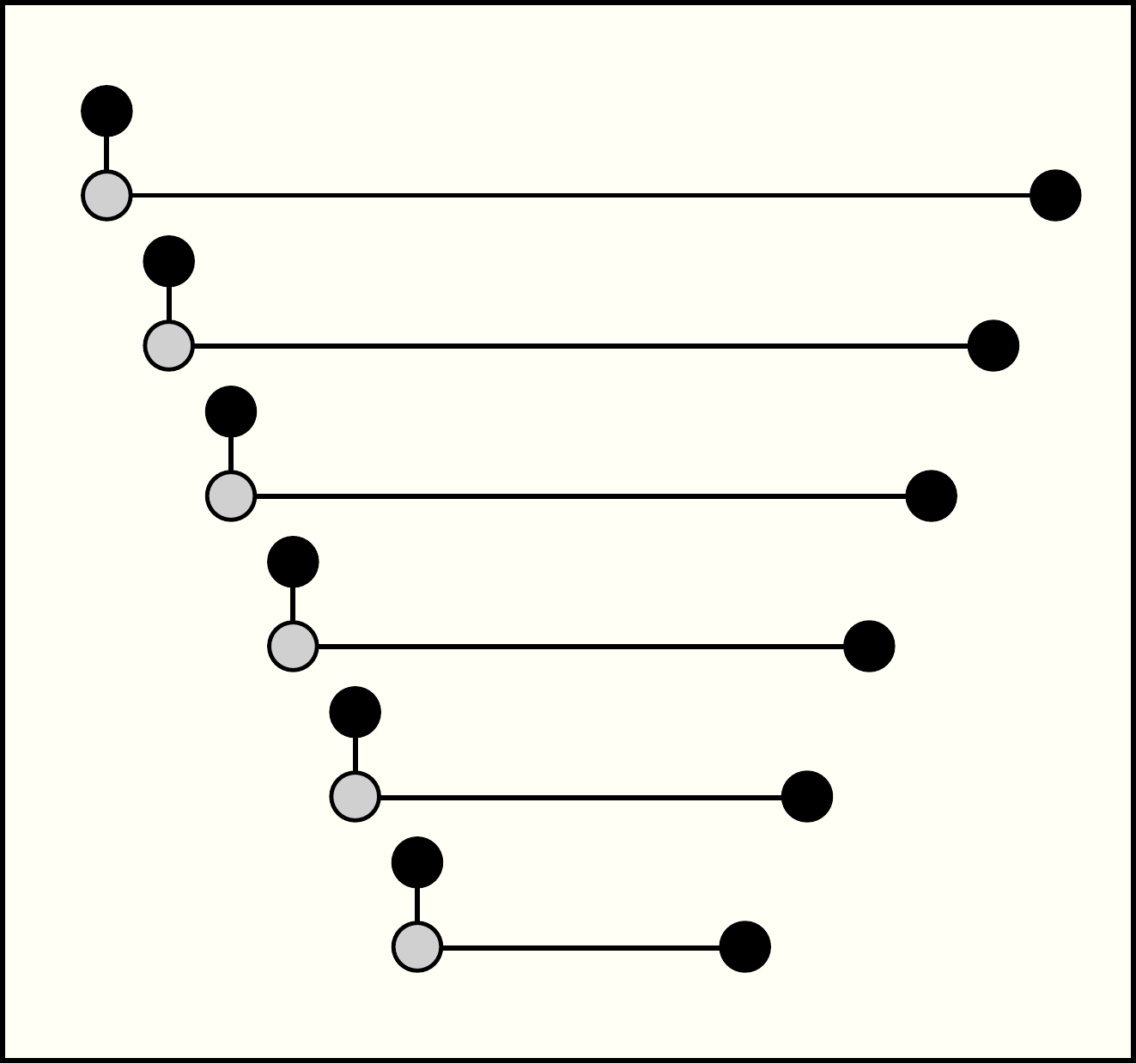}
		\caption{132- and 231-avoiding}\label{sfig:sm-lb-132}
	\end{subfigure}
	\caption{Two point sets with partial, but already long, Manhattan networks. \label{fig:mn2}}
\end{figure}

\subsection{\texorpdfstring{$k$}{k}-server on the line}


We consider the offline $k$-server problem on the line, in the following called simply \emph{$k$-server}. The input is a sequence of $n$ reals (\emph{requests}) in $[0,1]$, which we need to \emph{serve} with $k$ servers. Each of the servers has a position in $[0,1]$ throughout time, initially $0$. We serve the $n$ requests one-by-one by moving the servers, requiring that at least one server visits the requested value. 

More precisely, let $X=(x_1, \dots,x_n) \in [0,1]^n$ be the sequence of requests and let $p_i^j$ denote the position of server $j$ after serving the request $x_i$, for $j \in [k]$ and $i \in [n]$.
We fix $p_0^j = 0$ and require $p_i^j \in [0,1]$ for all $i,j$. 
The values $p_i^j$ represent a \emph{valid solution} for $X$ if for all $i\in[n]$ there is some $j \in [k]$ such that $p_i^j = x_i$. 
The \emph{cost} of the solution is the total  movement of all servers, i.e., $$\sum_{i \in [n]} \sum_{j \in [k]}{|p_i^j - p_{i-1}^j|}.$$

Let $\OPT_k(X)$ denote the minimum cost of a valid solution for $X$ using $k$ servers; this is called the \emph{offline} optimum, as it can be computed with full knowledge of $X$. 
It is well known that the offline optimum solution can be assumed to move at most one server for each request~\cite{Borodin}. (Other moves may be postponed to future requests.) Nonetheless, it will be helpful for us to also make use of solutions that move multiple servers at once.





\medskip

In general, $\OPT_k(X) \leq \frac{n}{k}$. To see this, 
assign an interval of size $\frac1k$ to each server, with each server responsible for requests in its interval, so each individual move costs at most~$\frac1k$.\footnote{Note that this is actually an \emph{online} algorithm.} This upper bound is tight up to a constant factor: Consider $m = \frac{n}{2k}$ repetitions of the sequence $(\frac1{2k}, \frac2{2k}, \dots, 1)$. In every repetition, there must be at least $k$ movements of at least $\frac1{2k}$, for a total cost of $\frac{n}{4k}$.

\begin{observation}
	$\OPT_k(X) \le \frac n k$ for every input $X \in [0,1]^n$, and for every $n$ there exists an input $X \in [0,1]^n$ with $\OPT_k(X) \ge \frac{n}{4k}$.
\end{observation}


A slight modification of the lower bound sequence involves $m = \frac{n}{k+1}$ repetitions of $(0,\frac{1}{k}, \frac{2}{k}, \dots, 1)$ with a total cost of at least $\frac{n}{k(k+1)}$. Observe that this sequence can be made order-isomorphic to a $(k+1)$-increasing (or $(k+1)$-decreasing) permutation by slight perturbation. This shows that avoiding the pattern $D_{k+2}$ (or $I_{k+2}$) does not help when we have $k$ servers. On the other hand, only one additional server will bring down the cost to virtually zero. In general, we can serve a $k$-increasing input $X$ with $k$ servers by partitioning $X$ into $k$ increasing subsequences and assigning a server to each; as each server moves only in one direction, their individual cost will be at most~1. (This is of course more generally true for all sequences obtained by interleaving $k$ \emph{monotone} sequences.)

\begin{observation}\label{p:k-server-k-inc}
	For each $k$-increasing input $X \in [0,1]^n$, we have $\OPT_k(X) \le k$.
	On the other hand, for each $n$, there is a $(k+1)$-increasing input $X$ with $\OPT_k(X) \in \Omega(\frac{n}{k^2})$.
\end{observation}

The upper bound $\fO(k)$ is tight for all $k$: Scale down a hard $k$-increasing sequence to $[\frac{1}{2},1]$; this forces all $k$ servers to move at least distance $\frac{1}{2}$ before serving their first request.

In the following, we consider inputs that avoid one or more patterns. In the spirit of \cref{p:k-server-k-inc}, we assume that $k$ is ``large enough'', i.e., at least some constant that we may choose depending on the avoided pattern(s). On the other hand, we also require $k$ to be not too large (as a function of $n$), as otherwise the overhead of our techniques dominates the cost.

Our main result, stated in the introduction as \cref{thm2}, is (roughly) that $\OPT_k(X) \in \fO(n^{\varepsilon})$ if $X \in [0,1]^n$ avoids a fixed pattern $\pi$, where $\varepsilon$ depends on $\pi$ and $k$.
We further provide an almost complete characterization of $\OPT_k$ when the input is in a principal permutation class (i.e., with a single avoided pattern), which we now summarize. 

\paragraph{Overview for principal classes.} For a fixed pattern $\pi$, the worst-case cost of serving a $\pi$-avoiding sequence of $n$ requests with $k$ servers is:
\begin{itemize}
	\item $\Theta(k)$ if $\pi$ is monotone and $k \ge |\pi|$. (\Cref{p:k-server-k-inc})
	\item  $n^{\Theta(1/\log k)}$ if $\pi$ is not separable and $2 \le k \le 2^{\sqrt{\log n}}$. (\Cref{p:k-server-tww,p:k-server-tww-lb})
	\item  $n^{\Theta(1/k)}$ if $\pi \in \{132, 213, 231, 312\}$ and $k \le \frac12 \sqrt{\log n}$. (\Cref{p:k-server-231-ub,p:k-server-312-lb})
	\item  $n^{\Omega(1/k)}$ and $n^{\fO(1/\log k)}$ otherwise, if $2 \le k \le \frac12 \sqrt{\log n}$. (\Cref{p:k-server-tww,p:k-server-312-lb})
\end{itemize}

The constants in the asymptotic notation all depend on $\pi$. Hence, as mentioned above, most of our results are non-trivial when $k$ is at least some constant depending on~$\pi$.


Observe that the starting position of the servers (determined as zero above) changes the overall cost by at most $k$, and hence is irrelevant for our asymptotic results.
Further observe that the cost of the $k$-server problem is (essentially) invariant under reversal and complement of the input sequence. Here, complement means replacing each value $x$ with $1-x$, which clearly changes the overall cost by at most $k$ (because of the starting positions). A reversed input sequence can be handled by a reversed solution, which again does not change the cost up to starting positions.

\paragraph{Upper bounds.} 
%
We leverage our distance-balanced decomposition (see \cref{sec:dist-bal}) to obtain the following result (proven in \cref{sec:k-server}) that implies Theorem~\ref{thm2}.

\begin{restatable}{theorem}{restateKServerTwinWidth}\label{p:k-server-tww}
	Let $\pi$ be a permutation with Füredi-Hajnal limit $c_\pi$.
	Then, for every $\pi$-avoiding input $X \in [0,1]^n$, we have $\OPT_k(X) \in \fO( k^{1 + \log_{5c_\pi} 120} \cdot n^{1/(\floor{\log_{5c_\pi} k} + 1)})$.
\end{restatable}

For more restricted classes, we can prove better upper bounds (\cref{sec:k-server-sep-ub,sec:k-server-231-ub,sec:k-server-sep-subclasses}).


\begin{restatable}{theorem}{restateKServerAvUB}\label{p:k-server-231-ub}
	For every input $X \in [0,1]^n$ that avoids 231 (or its symmetries), we have $\OPT_k(X) \in \fO(k^2 + k \cdot n^{1/k})$.
\end{restatable}

\begin{restatable}{theorem}{restateKServerSepSubUB}\label{p:k-server-sep-subclasses-ind}
	Let $\pi \in \Av(231)$ and $k = 2^{|\pi|+1}$.
	Then, for every input $X \in [0,1]^n$ that avoids both $231$ and $\pi$, we have $\OPT_k(X) \le 2^{|\pi| + 2}$.
\end{restatable}

The proofs for \cref{p:k-server-231-ub,p:k-server-sep-subclasses-ind} do not rely on distance-balanced merge sequences. In fact, they are not based on twin-width at all, but use the decomposition inherent to $231$-avoiding permutations (see \cref{sec:k-server}). 
The idea is to take the top-level decomposition, and serve each part of the decomposition with a number of servers depending (mainly) on its size. A similar approach can be used when $X$ is $t$-separable, sharpening the bound of Theorem~\ref{p:k-server-tww} in this special case.

\begin{restatable}{theorem}{restateKServerSepUB}\label{p:k-server-sep-ub}
	Let $t \ge 2$.
	Then, for every $t$-separable input $X \in [0,1]^n$, we have $\OPT_k(X) \in \fO( k \cdot t \cdot n^{1/(\floor{\log k} + 1)})$.
\end{restatable}

\paragraph{Lower bounds.} We present two constructions based on a similar general idea.

Start with an appropriate input that can be efficiently served by our $k$ servers (or even less), but restricts movement somehow.
For example, in an efficient solution, some of the servers must remain in certain intervals most of the time.
Then, repeatedly \emph{inflate} some values in the input with a smaller copy of the original input, further restricting efficient solutions.
At some point, the servers find themselves in a lose-lose situation: Either they stick with the restrictions imposed by the ``global'' permutation, lacking enough servers in place for some of the small copies; or they make sure there are enough servers for most of the small copies, causing too much movement globally.


\begin{restatable}{theorem}{restateKServerTwwLB}\label{p:k-server-tww-lb}
	For each $n, k, d$, there is an input $X \in [0,1]^n$ of twin-width $d$ such that $\OPT_k(X) \in \Omega(d^{-2}k^{-3} \cdot n^{1/(\floor{\log_{2d} k} + 1)})$.
\end{restatable}

Since separable inputs are precisely those of twin-width 1, \cref{p:k-server-tww-lb} with $d=1$ implies that \cref{p:k-server-sep-ub} is tight up to a factor of $t \cdot k^4$.
Moreover, \cref{p:k-server-tww} is tight if the avoided pattern~$\pi$ is \emph{non-separable} (up to $\pi$-dependent constants in the exponent). This is because the separable input obtained from \cref{p:k-server-tww-lb} with $d=1$ must avoid all non-separable patterns.
Since almost all patterns are non-separable~\cite{brignall2010survey}, 
\cref{p:k-server-tww} is tight in this way for almost all $\pi$.

Since the construction in \cref{p:k-server-tww-lb} has twin-width $d$, it avoids the canonical $(d+1) \times (d+1)$ grid permutation $\gamma$ (which has twin-width $d+1$). We know $c_\gamma \in 2^{\fO(d)}$~\cite{jfox}, which yields the following. 
\begin{corollary}\label{p:k-server-tww-conseq}
	For each $n, k, d$, there is a permutation $\pi$ of size $d^2$ and a $\pi$-avoiding input $X \in [0,1]^n$ such that $\OPT_k(X) \in \Omega(d^{-2}k^{-3}) \cdot n^{\Omega( \log\log c_\pi / \log k)}$.
\end{corollary}
Note that the corresponding upper bound (\cref{p:k-server-tww}) is $k^{\fO(1)} \cdot n^{\fO( \log c_\pi / \log k)}$. \Cref{p:k-server-tww-conseq} implies that the dependence on $\pi$ in the exponent of $n$ is necessary, and we cannot hope for an equivalent of \cref{p:k-server-sep-ub} for every avoided pattern.

\medskip

Our second construction is specific to $\Av(231)$ and its symmetries. 

\begin{restatable}{theorem}{restateKServerAvLB}\label{p:k-server-312-lb}
	For each $n, k$ there is an input $X \in [0,1]^n$ avoiding 231 or its symmetries such that $\OPT_k \in \Omega(4^{-k} \cdot n^{1/k})$.
\end{restatable}

This shows that \cref{p:k-server-231-ub} is tight up to a factor of $\fO(k^2 \cdot 4^k)$. Further, \cref{p:k-server-sep-subclasses-ind} reveals that \cref{p:k-server-312-lb} is, in a sense, best possible, since for every proper subclass of $\Av(231)$ or its symmetries, $\OPT_k(X)$ is already bounded.



\subsection{Euclidean TSP}

In this subsection we consider the following problem: Given a set $P \subset [0,1]^2$ of $n$ points, find a \emph{tour} that visits each point in $P$. Let us denote the shortest euclidean length of such a tour as $\TSP(P)$.

It is well-known that $\TSP(P) \in \fO(\sqrt{n})$~\cite{beardwood1959shortest}.
For an easy argument, consider a uniform $\sqrt{n} \times \sqrt{n}$ grid inside $[0,1]^2$, and observe that traversing the grid points and routing each input point to the nearest grid point has total cost $\fO(\sqrt{n})$.

This bound is tight as the example of the $n$ grid points itself shows, since each point must be routed to a neighbor at cost at least $1/\sqrt{n}$; in fact, a uniform random set of $n$ points has w.h.p.\ cost $\Theta(\sqrt{n})$~\cite{beardwood1959shortest}.

The following warm-up example motivates the study of pattern-avoiding point sets. 

\begin{observation}
	If $P \subset [0,1]^2$ is $k$-increasing, then $\TSP(P) \in \fO(k)$.
\end{observation}

To see this, partition $P$ into $k$ increasing subsets, connect each subset by a path (of length at most $2$), and connect the paths to each other at a further cost of at most $k$.


Let us attempt to generalize the observation to arbitrary avoided patterns. We sketch a simple but suboptimal argument first. Let $P \subset [0,1]^2$ be a set of $n$ points avoiding $\pi$. Partition $[0,1]^2$ into $\frac{\sqrt{n}}{c_\pi} \times \frac{\sqrt{n}}{c_\pi}$ equal square cells. Observe that at most $\sqrt{n}$ of these cells have any point in them. Otherwise, by \cref{p:marcus-tardos-with-limit-non-square} we would have $|\pi|$ nonempty cells forming the pattern $\pi$.

Construct a tour recursively, by constructing a tour of each nonempty cell, and connecting these mini-tours by a tour of a representative in each. 
Assume for simplicity that the $n$ points are distributed equally among the nonempty cells (this can be shown to maximize the bound via Jensen's inequality; in this proof sketch we omit the formal justification). 
The total length $t(n)$ can then be bounded as: $$t(n) ~~\leq~~ t\left(\sqrt{n}\right) + c_\pi \cdot t\left(\frac{n}{\sqrt{n}}\right) ~~=~~ (c_\pi + 1) \cdot t \left(\sqrt{n}\right).$$

The recurrence already gives a non-trivial bound of the form $\fO(\log^{\log_2{(c_\pi+1)}}{n})$. Our main theorem (all proofs in \cref{sec:mst}), however, will yield a significant improvement.


\begin{restatable}{theorem}{restateThmGeneralMST}
	\label{thm:general-mst}
	Let $\pi$ be a permutation with Füredi-Hajnal limit $c_\pi$. Then, for every $\pi$-avoiding point set $P \subset [0,1]^2$ in general position, we have $\TSP(P) \in \fO(c_\pi \log |P| )$.
\end{restatable}

We give a lower bound that shows the result to be asymptotically tight for essentially all avoided patterns. The proof is by a ``balanced tree'' point set.

\begin{restatable}{theorem}{restateMSTSepLowerBound}\label{p:mst-sep-lb}
	For each $n$, there is a 231-avoiding point set $P \subset [0,1]^2$ of size $n$ with $\TSP(P) \in \Omega( \log n )$.
\end{restatable}

Observe that each non-monontone pattern $\pi$ contains 231 or one of its symmetries (132, 213, 312). Hence, we have the dichotomy that if our point set $P$ avoids a monotone pattern, then $\TSP(P)$ is constant, and if $P$ avoids a non-monotone pattern $\pi$, then $\TSP(P) \in \Theta_\pi( \log n)$ in the worst case.

We also show a sharp dichotomy between $\Av(231)$ and its subclasses. 

\begin{restatable}{theorem}{restateMSTSepSubclasses}\label{p:mst-sep-subclasses-ind}
	Let $\pi \in \Av(231)$. Then, for every point set $P \subset [0,1]^2$ that avoids 231 and $\pi$, we have $\TSP(P) \le 12 |\pi|$.
\end{restatable}

Finally, we show a lower bound implying that in general, an (almost) linear dependence on twin-width cannot be avoided. 


\begin{restatable}{theorem}{restateMSTTwwLB}\label{p:mst-tww-lb}
	For each $n$ and $d \ge 2$, there is a set $P \subset [0,1]^2$ of $n$ points with twin-width~$d$ such that $\TSP(P) \in \Omega(\frac{d}{\log d} \log n)$.
\end{restatable}

In particular, if $\gamma_d$ is a $d \times d$ grid permutation, then, for a worst-case $\gamma_d$-avoiding point set $P$ of $n$ points, $\TSP(P)$ is between $\Omega(\frac{d}{\log d} \log n)$ and $2^{\fO(d)} \log n$. 

\subsection{Symmetries and general position}

We remark that all three problems are invariant under reversal and complement. For the BST and $k$-server problems, however, reversal implies running the operation sequence backwards, which is only allowed in the offline case. In addition, the TSP problem is invariant under rotation, and the BST problem (in the satisfied superset formulation) under 90-degree rotation; the latter implies a swap between time- and keyspace-dimensions, thus less easily interpretable in the BST view. The restriction of the input to $[0,1]$, resp., $[0,1]^2$ is without loss of generality, as the offline algorithm can conceptually scale/shift the input. 


In our bounds for all three problems, we usually require the input to be in general position.
We sketch a way of extending our results to inputs that are not in general position. 

Let $P \subset [0,1]^2$ be a point set, not necessarily in general position, let $D$ be the minimum distance between two points in $P$, and let $0 < \varepsilon < D$. 
Construct $P'$ from $P$ by replacing every point $(x,y)$ with $(x+\varepsilon \cdot y, y + \varepsilon \cdot x)$.
Clearly, $P'$ is in general position; every ``row'' and every ``column'' of points is replaced by an increasing point set.
Strict inequalities between point coordinates are preserved, but pattern-avoidance is not.
However, it can be shown that if $P$ avoids a permutation $\pi$, then $P'$ avoids a permutation $\sigma$, obtained by inflating $\pi$ with $|\pi|$ copies of the permutation 21.
Note that for BST and $k$-server, the $x$-coordinates describe the temporal order of accesses/requests, and thus are already distinct; here, only the transformation of $y$-coordinates has any effect.

If $P$ corresponds to a $k$-server or TSP input, by choosing $\varepsilon$ appropriately, we can keep the cost of the transformed point set $P'$ arbitrarily close to the cost of $P$.  
Hence, \cref{p:k-server-tww,thm:general-mst} generalize to inputs not in general position; only the constant $c_\pi$ in the bounds is replaced. 

For the BST (and satisfied superset) problem, we use the same transformation. 
Let $X' \in [m]^m$ be the permutation access sequence obtained by perturbing the point set corresponding to $X \in [n]^m$. It can be shown that $\OPT(X') \in \fO(m)$ if and only if $\OPT(X) \in \fO(m)$ (see~\cite[\S\,E]{FOCS15}), thus generalizing Theorem~\ref{p:ass-tww}.


\section{Arborally satisfied superset}\label{sec:ass}

In this section, we prove:
\restateASSTwinWidth*

As defined in \S\,\ref{sec311}, if $P$ is a set of points and $x, y \in P$, then we call $x$ and $y$ \emph{connected} and write $x \connected{P} y$ if there is a monotone path $T$ connecting $x$ and $y$ that consists of axis-parallel line segments, and each corner of $T$ is contained in $P$. For two sets $X, Y$ we write $X \connected{P} Y$ if $x \connected{P} y$ for all $x \in (X \cap P)$, $y \in (Y \cap P)$. A set of points is called \emph{arborally satisfied} if every pair of points in $P$ is connected in $P$ (i.e., $P \connected{P} P$).

Let $\cR_1, \cR_2, \dots, \cR_n$ be a merge sequence of a point set $P$. For each rectangle family~$\cR_i$, let $G(\cR_i)$ be the set of all intersection points of lines $v$ and $h$, where $v$ is a vertical line containing a left or right side of a rectangle in $\cR_i$, and $h$ is a horizontal line containing the top or bottom side of a rectangle in $\cR_i$. In particular, note that $G(\cR_i)$ contains the corners of each rectangle in $\cR_i$.



We construct an arborally satisfied superset $A \supseteq P$ as follows. Let $A_0 = P$. Consider a step from $\cR_i$ to $\cR_{i+1}$ in the merge sequence where two rectangles $Q_1$ and $Q_2$ are merged into a rectangle $Q$. Let $A_i = A_{i-1} \cup (Q \cap G(\cR_i))$, and finally, $A = A_n$.
In words, we add to the point set each point in $Q$ that lies on the ``grid'' induced by rectangles in $\cR_i$ (including $Q_1$ and $Q_2$), if that point is not already present (see \cref{fig:ass-constr}).

%
%
%
%

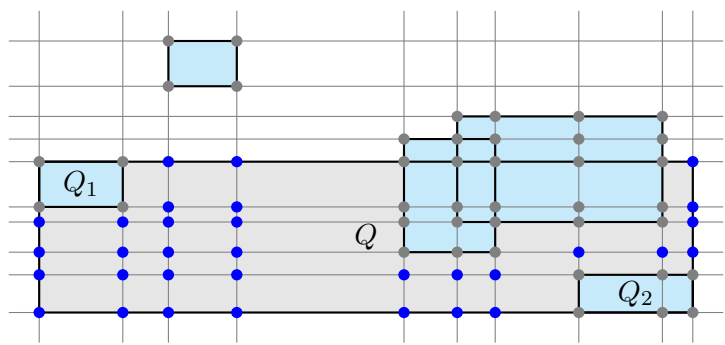
\begin{figure}[h]
	\centering
	\begin{tikzpicture}[
		boxBG/.style={msrect,draw=none},
		boxFG/.style={msrect,fill=none, thick},
		bigboxBG/.style={fill, {white!90!black}},
		bigboxFG/.style={draw, thick},
		gridline/.style={gray},
		newpoint/.style={point, blue},
		oldpoint/.style={point, gray}
	]
		\def\bboxDist{0.4}
		\input{figs/ass-big.tex}
	\end{tikzpicture}
	\caption{Sketch of points added in a single step in \cref{p:ass-tww}. Grey lines are the grid induced by rectangle sides, grey points are known to be present already, and blue points are new.}\label{fig:ass-constr}
\end{figure}

\begin{lemma}\label{p:ass-tww-invariants}
	For each $i \in \{0, 1, \dots, n\}$, the following two invariants hold:
	\begin{enumerate}[(i)]
		\item For each rectangle $Q \in \cR_i$, we have $(Q \cap G(\cR_i)) \subseteq A_i$.
		\label[iinv]{inv:corners}
		\item If two (not necessarily distinct) rectangles $Q, Q' \in \cR_i$ are not homogeneous, we have $Q \connected{A_i} Q'$.\label[iinv]{inv:sep-connected}
	\end{enumerate}
\end{lemma}
\begin{proof}
	\Cref{inv:corners} clearly holds initially. It is then maintained, since the required points are added whenever adding a new rectangle, points are never removed, and the set $G(\cR_i)$ does not gain new points with increasing $i$.
	
	To prove \cref{inv:sep-connected}, we proceed by induction on $i$. In the beginning, each rectangle $Q$ consists of a single point, hence trivially $Q \connected{A_0} Q$. Moreover, any two distinct rectangles are homogeneous, so the claim holds.
	
	Now consider the $i$-th step, where two rectangles $Q_1, Q_2 \in \cR_i$ are merged into a single rectangle $Q \in \cR_{i+1}$. Suppose our claim holds for $(A_i, \cR_i)$.
	Let $B = Q \cap G(\cR_i)$. Observe that $B$ contains all newly added points. 
	
	
	\begin{figure}
		\tikzset{conline/.style={thick, red}}
		\def\linemar{.7}
		\small
		\begin{subfigure}[c]{.33\textwidth}
			\centering
			\begin{tikzpicture}[scale=0.6]
				\def\yoff{1,.3}
				\draw[subrect] (2,0) rectangle node[below] {$Q'$} (4,1.8);
				
				\node[point] (x) at (.2,1) {};
				\node[above] at (x) {$x$};
				
				\node[point] (c) at (2,1) {};
				\node[above left] at (c) {$c$};
				
				\node[point] (y) at ($(c)+(\yoff)$) {};
				\node[right] at (y) {$y$};
								
				\draw ($(x)-(\linemar,0)$) -- (4+\linemar,1) node[right] {$h$};
				
				\draw[conline] (x) -- (c) -- ++(.33,0) -- ++(0,0.15) -- ++(.33,0) -- ++(0,.15) -- (y);
			\end{tikzpicture}
		\end{subfigure}%
		\begin{subfigure}[c]{.33\textwidth}
			\centering
			\begin{tikzpicture}[scale=0.6]
				\def\yoff{.6,-.8}
				\draw[subrect] (-.5,2.5) rectangle node[above=-1mm] {$Q$} (2.7,4);
				\draw[subrect] (2,0) rectangle node[below] {$Q'$} (3.5,2.1);
				
				\node[point] (x) at (0,3) {};
				\node[above] at (x) {$x$};
				
				\draw (x) -- (3.3,3) node[right] {$h'$};
				
				\node[point] (c) at (2,2.1) {};
				\node[left] at (c) {$c$};
				
				\draw (c) -- (2,4.7) node[above]{$v$};
				
				\node[point] (z) at (2,3) {};
				\node[above right] at (z) {$z$};
				
				\node[point] (y) at ($(c)+(\yoff)$) {};
				\node[right] at (y) {$y$};
				
				\draw[conline] (x) -- (z) -- (c) -- ++(0,-.3) -- ++(.3,0) -- ++(0,-.2) -- ++(.3,0) -- (y);
			\end{tikzpicture}
		\end{subfigure}%
		\begin{subfigure}[c]{.33\textwidth}
			\centering
			\begin{tikzpicture}[scale=0.8]
				\def\xoffset{-.4,.8}
				\def\yoffset{.4,-.4}
				
				\coordinate (cc) at (1.1,2);
				\coordinate (cc2) at (2.1,1.5);
				
				\draw[subrect] (-0.3,3) rectangle node[left=-1mm] {$Q_1$} (cc);
				\draw[subrect] (cc2) rectangle node[below] {$Q'$} (3.4,0);

				\node[point] (c) at (cc) {};
				\node[below] at (c) {$c$};
				
				\node[point] (x) at ($(c)+(\xoffset)$) {};
				\node[above=1mm] at (x) {$x$};
				
				\node[point] (y) at ($(cc2)+(\yoffset)$) {};
				\node[right] at (y) {$y$};
				
				\draw[conline] (x) -- ++(0,-.3) -- ++(.2,0) -- ++(0,-.3) -- ++(.2,0) -- (c)
						-- ++(.33,0) -- ++(0,-.25) -- ++(.33,0) -- ++(0,-.25) -- (cc2)
						-- ++(0,-.2) -- ++(.2,0) -- ++(0,-.2) -- (y);
			\end{tikzpicture}
		\end{subfigure}%
		
		\begin{subfigure}[c]{.33\textwidth}\caption{\Cref{case:grid-to-non-homog}.\label{sfig:ass-grid-to-non-homog}}\end{subfigure}%
		\begin{subfigure}[c]{.33\textwidth}\caption{\Cref{case:grid-to-homog}.\label{sfig:ass-grid-to-homog}}\end{subfigure}%
		\begin{subfigure}[c]{.33\textwidth}\caption{\Cref{case:rect-to-rect}.\label{sfig:ass-rect-to-rect}}\end{subfigure}%
		
		\caption{Illustrations for the proof of \cref{p:ass-tww-invariants}. The red lines indicate connections between $x$ and $y$. Note that in (b), $Q$ may be much larger and overlap or even contain $Q'$.}
	\end{figure}
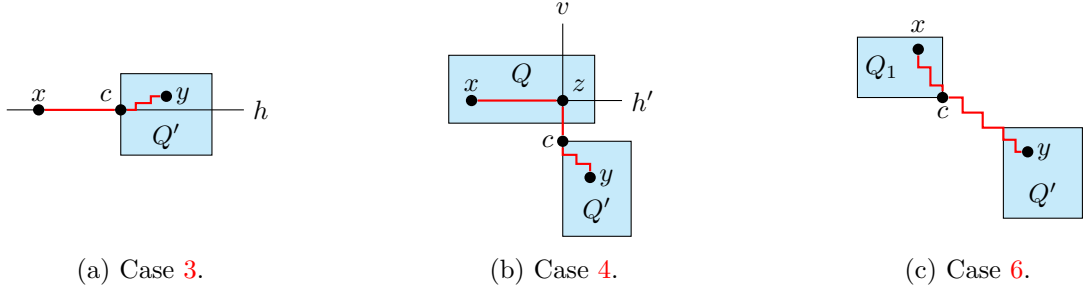
	
	We now prove \cref{inv:sep-connected} for $(A_{i+1}, \cR_{i+1})$. Let $x, y \in A_{i+1}$. We consider the following cases.
	\begin{enumerate}
		\item $x, y \in A_i \setminus Q$.  If $x \in Q'$ and $y \in Q''$ for two rectangles $Q', Q''$ that are not homogeneous, then $x \connected{A_i} y$ by induction, implying $x \connected{A_{i+1}} y$. Otherwise, there is nothing to prove.
		
		\item $x, y \in B$. Since $B$ forms a full grid inside $Q$, clearly $x$ and $y$ are connected, possibly via another point in $B$.
		
		\item $x \in B$ and $y \in Q' \setminus B$, where $Q' \in \cR_i$, and $x$ is not homogeneous with $Q'$. If $x \in Q'$, then $x \in A_i$ by \cref{inv:corners} and thus $x \connected{A_i} y$ by \cref{inv:sep-connected}.
		Otherwise, w.l.o.g.\ $Q'$ is to the right of $x$. (See \cref{sfig:ass-grid-to-non-homog}.) Let $h$ be the horizontal line through $x$. Since $x$ and $Q'$ are not homogeneous, $h$ must intersect $Q'$. Let $c$ be the leftmost intersection point of $h$ and $Q'$.
		Since $x \in B \subseteq G(\cR_i)$, we know that $h$ contains the top or bottom side of some rectangle in $\cR_i$, and so $c \in (G(\cR_i) \cap Q') \subseteq A_i \subseteq A_{i+1}$ by \cref{inv:corners}. Further, $c \connected{A_i} y$ by \cref{inv:sep-connected}, so $x$ is connected to $y$ via $c$.\label[icase]{case:grid-to-non-homog}
		
		\item $x \in B$ and $y \in Q' \setminus B$, where $Q' \in \cR_i$, and $x$ is homogeneous with $Q'$. W.l.o.g., $Q'$ is below and to the right of $x$. (See \cref{sfig:ass-grid-to-homog}.)
		
		If $Q'$ and $Q$ are also homogeneous, there is nothing to do. Otherwise, let $c$ be the top left corner of $Q'$. Since $c \connected{A_i} y$ by \cref{inv:sep-connected}, it suffices to show that $x \connected{A_{i+1}} c$.
		
		Let $v$ and $h$ be the vertical and horizontal line through $c$. One of $v$, $h$ intersects $Q$, say $v$. Let $h'$ be the horizontal line through $x$, and let $\{z\} = v \cap h'$. We have $z \in Q$ and $z \in G(\cR_i)$, hence $z \in B \subseteq A_{i+1}$ by \cref{inv:corners}. Thus, $x$ and $c$ are connected via $z$.\label[icase]{case:grid-to-homog}
		
		\item $x, y \in Q_1 \setminus B$. Then $x \connected{A_i} y$ by induction.
		
		\item $x \in Q_1 \setminus B$ and $y \in Q' \setminus B$, where $Q' \in \cR_i \setminus \{Q_1\}$. If $Q'$ and $Q_1$ are not homogeneous, then $x \connected{A_i} y$ by induction. If $Q'$ and $Q$ are homogeneous, then there is nothing to prove.
		
		Now suppose $Q'$ and $Q_1$ are homogeneous, but $Q'$ and $Q$ are not homogeneous. W.l.o.g., $Q'$ is below and to the right of $Q_1$. (See \cref{sfig:ass-rect-to-rect}.)
		Let $c$ be the bottom right corner of $Q_1$. We have $x \connected{A_i} c$ by \cref{inv:sep-connected}. Since $c \in B$, we have $c \connected{A_{i+1}} y$ (shown in \cref{case:grid-to-non-homog}), which implies $x \connected{A_{i+1}} y$.\label[icase]{case:rect-to-rect}
%
	\end{enumerate}
	Observe that $A_i \subset \bigcup_{Q \in \cR_i} Q$ for all $i$. This means that the list of cases is exhaustive, up to swapping $x$ with $y$ and $Q_1$ with $Q_2$.
\end{proof}

It remains to bound the number of points added in each step.

\begin{lemma}\label{p:ass-tww-count}
	If $\cR_1, \cR_2, \dots, \cR_n$ is $d$-wide, then each step in the algorithm adds at most $(2d+4)^2$ points.
\end{lemma}
\begin{proof}
	Consider a step of merging $Q_1, Q_2 \in \cR_i$ into a single rectangle $Q \in \cR_{i+1}$. By assumption, there are at most $d$ rectangles in $\cR_{i+1}$ that are not homogeneous with $Q$ (this includes $Q$). In $\cR_i$, there are two more, namely $Q_1$ and $Q_2$. Extending each side of these rectangles into lines yields at most $4d+8$ lines that intersect $Q$, creating at most $(2d+4)^2$ intersection points within $Q$.
\end{proof}

\Cref{p:ass-tww-invariants,p:ass-tww-count} together imply \cref{p:ass-tww}. 

\section{Distance-balanced merge sequences}
\label{sec4}\label{sec:dist-bal}

We say that a tuple of intervals $(I_1, \ldots, I_s)$ is a \emph{partition} of the real interval $[0,1]$ if (i) the intervals $I_j$ are disjoint, (ii) their union is equal to $[0,1]$, and (iii) $I_1 < I_2 < \dots < I_s$.
For two positive integers $r,s$, an $r \times s$-\emph{gridding} of the unit square $[0,1]^2$ is a pair $G = (P_1, P_2)$ where $P_1 = (I_1, \ldots, I_r)$ and $P_2 = (J_1, \ldots, J_s)$ are partitions of $[0,1]$.
We say that $I_x$ is the \emph{$x$-th column} of $G$, $J_y$ is the \emph{$y$-th row} of $G$, and $I_x \times J_y$ is the \emph{$(x,y)$-cell} of $G$.
A gridding $G$ is a \emph{coarsening} of a gridding $G'$ if $G$ can be obtained from $G'$ by repeated merging of consecutive rows and columns.

Recall the definition of a $d$-wide merge sequence from \S\,\ref{sec2}. 
We call a $d$-wide merge sequence $\cR_1, \dots, \cR_n$ of a point set $P \subset [0,1]^2$, augmented with a sequence of griddings $G_1, \dots, G_n$, \emph{distance-balanced} if the following properties hold. Here $p_i$ denotes the number of columns in $G_i$, and $q_i$ denotes the number of rows in $G_i$.
\begin{enumerate}[(a)]
	\item Each row (or column) of $G_i$ contains at most $\frac{d}{2}$ non-empty cells.\label[iprop]{decomp:gridding}
	\item Each rectangle of $\cR_i$ is contained in a single cell of $G_i$ and has width smaller than $\frac{20}{p_i}$ and height smaller than $\frac{20}{q_i}$.\label[iprop]{decomp:rect-size}
	\item Each column of $G_i$ wider than $\frac{40}{p_i}$ contains at most $\frac{d}{2}$ points of~$P$ and each row of $G_i$ taller than $\frac{40}{q_i}$ contains at most $\frac{d}{2}$  points of~$P$.\label[iprop]{decomp:super-wide}
	\item $\frac{9}{40} d \cdot (\max(p_i, q_i)-1) \le n-i+1 \le \frac12 d \cdot \min(p_i, q_i)$.\label[iprop]{decomp:grid-size}
\end{enumerate}

\Cref{decomp:gridding} implies that the 
grid is sparsely populated. 
\Cref{decomp:rect-size} implies that rectangles are never much wider or taller than the \emph{average} column width or row height.
Ideally, we would like to also avoid columns/rows that are significantly wider/taller.
This is not always possible, but \cref{decomp:super-wide} at least ensures that very wide/tall columns/rows contain only few points. Finally, \cref{decomp:grid-size} implies that $p_i, q_i \in \Theta(n-i+1) = \Theta(|\cR_i|)$, so the number of rows and columns is linear in the number of rectangles.

\begin{theorem}
\label{thm:decomposition}
Let $\pi$ be a permutation with Füredi-Hajnal limit $c_\pi$ and let $P$ be a $\pi$-avoiding set $P \subset [0,1]^2$ of $n$ points in general position. Then $P$ has a distance-balanced $10\,c_\pi$-wide merge sequence.
\end{theorem}
\begin{proof}
Let $t = 5 \, c_\pi$ and $C = 20$.
We say that a column $x$ of an $r \times s$-gridding $G$ is \emph{wide} if $|x| > \frac{C}{r}$, and a row $y$ is \emph{tall} if $|y| > \frac{C}{s}$ where $|\cdot|$ denotes the size of the interval, i.e., the width of the column $x$ (resp.~height of the row $y$).
Observe that in any gridding $G$ of the unit square, at most a constant fraction of columns can be wide and the same holds for tall rows. We call a cell of $G$ \emph{wide} (\emph{tall}) if it is in a wide column (tall row).

First, we describe an algorithm that constructs the desired merge sequence along with a sequence of griddings and then prove that \cref{decomp:gridding,decomp:rect-size,decomp:grid-size,decomp:super-wide} hold.
Throughout the execution the following invariants are maintained: 
\begin{enumerate}[(i)]
	\item Each rectangle $S \in \cR_i$ is contained in a single cell of $G_i$.\label{cond:cell}
	\item Each row (or column) of $G_i$ contains at most $t$ rectangles of $\cR_i$.\label{cond:single}
	\item The union of any two consecutive non-tall rows (or non-wide columns) contains strictly more than $t$ rectangles of $\cR$.\label{cond:pairs}
\end{enumerate}

Let $s = \lceil \frac{n}{t}\rceil$.
Initially, set $\cR_1 = P$ and let the gridding $G_1$ consist of columns $c_1, \ldots, c_s$ and rows $r_1, \ldots, r_s$ such that for every $i \in [s-1]$, the row $r_i$ and column $c_i$ both contain exactly $t$ points of $P$.
Observe that invariants \cref{cond:cell,cond:single,cond:pairs} hold trivially for $\cR_1$ and $G_1$.

The algorithm repeats the following main step as long as there are at least two rectangles.
If the rectangle family $\cR_i$ contains two rectangles sharing the same cell of the gridding $G_i$ that is neither wide nor tall, merge these two rectangles into a single new one, obtaining the  rectangle family $\cR_{i+1}$. 
Later, we show that $\cR_i$ must always contain such a pair of rectangles as otherwise, an occurrence of $\pi$ in $P$ is guaranteed.
Afterwards, the algorithm repeatedly merges all pairs of neighboring non-wide columns and non-tall rows in the gridding that violate invariant~\ref{cond:pairs}.
Note that one such merge might trigger a cascade of additional merges since the dimensions of the grid are decreasing and thus, other columns (resp. rows) might cease to be wide (resp. tall).
Once this process stops, invariant~\ref{cond:pairs} is satisfied and we set $G_{i+1}$ to be the obtained coarsening of~$G_i$.

\paragraph{Correctness.}
First, we verify that invariants \cref{cond:cell,cond:single,cond:pairs} are not violated in the $i$-th step.
Invariant~\ref{cond:cell} holds since the only new rectangle in $\cR_{i+1}$ is obtained by merging two rectangles inside a single cell of $G_i$, and $G_{i+1}$ is a coarsening of $G_i$.
Invariant~\ref{cond:single} is not violated since we are only merging consecutive rows and columns that violated condition~\ref{cond:pairs}, i.e., they together contained at most $t$ rectangles.
Finally, \ref{cond:pairs} is satisfied by the algorithm explicitly.

We now show that there must be two rectangles sharing a non-wide and non-tall cell in the $i$-th step of the algorithm. Suppose not.
Let us construct the point set \[M = \{(x,y) \in [p_i]\times [q_i] \mid \text{$\cR_i$ contains a rectangle in the $(x,y)$ cell}\}.\]
Observe that if $M$ contains a permutation $\pi$ then so does $P$ since each rectangle contains at least one point of $P$.
Let $W$ and $T$ denote the number of wide columns and tall rows in the gridding $G_i$, respectively.
We claim that $G_i$ contains at least $\lfloor\frac{p_i}{2}\rfloor  - W \ge \frac{p_i - 1}{2} - W$ pairs of neighboring non-wide columns.
This can be seen by partitioning the columns into $\lfloor \frac{p_i}{2}\rfloor$ neighboring pairs, since each wide column blocks at most one of the pairs.
Thus, condition \ref{cond:pairs} guarantees that there are strictly more than $ \left( \frac{p_i - 1}{2}  - W \right)\cdot t$ rectangles in non-wide columns.
At most $T\cdot t$ of these rectangles lie in tall rows by~\ref{cond:single}, and each remaining rectangle must occupy a cell on its own (by assumption).
Thus, we get $|M| > \left( \frac{p_i - 1}{2}  - W - T \right)\cdot t$.

Observe that $W < \frac{p_i}{C}$ since each wide column has width strictly larger than $\frac{C}{p_i}$ and thus $W \le \frac{p_i-1}{C}$, as $C$ is a positive integer.
Similarly, $T \le \frac{q_i-1}{C}$.
Putting these together:
\[|M| > \left(\frac{p_i-1}{2} - \frac{p_i-1}{C} - \frac{q_i - 1}{C}\right)\cdot t.\]
Assume w.l.o.g.\ that $p_i \ge q_i$ and $p_i \ge 2$.\footnote{If $p_i = q_i = 1$, the only cell is neither wide nor tall, so our claim follows trivially.} Then we have
\[|M| > \left(\frac{1}{2} - \frac{2}{C}\right) \cdot t \cdot (p_i-1) \ge \frac25 \cdot t \cdot \frac{p_i}{2} = c_\pi \cdot p_i.\]
Now \cref{p:marcus-tardos-with-limit-non-square} implies that $M$ contains  $\pi$, and thus $P$ contains $\pi$, a contradiction.

\paragraph{Width of the merge sequence.} Let $Q \in \cR_i$ be a rectangle in step $i$ of the construction. Each other rectangle in $\cR_i$ that is not homogeneous with $Q$ must share a row or column with $Q$. By \cref{cond:single,cond:cell}, there are at most $2t-2$ such rectangles. Hence, $\cR_i$ is in particular $2t = 10\,c_\pi$-wide.

\paragraph{Distance-balanced.}
Finally, let us show that the obtained merge sequence is distance-balanced. Let $d = 2t = 10\,c_\pi$.
First, observe that \cref{decomp:gridding} follows by a simple combination of \ref{cond:cell} and 
\ref{cond:single}.

Furthermore, any rectangle $S \in \cR_i$ corresponds either to an original point in $P$ and its width and height are both zero; or it was created by merging two rectangles $S, T$ in the $j$-th step for some $j < i$.
In that case, both $S$ and $T$ occupied the same cell of $G_j$ in a non-wide column and a non-tall row and thus, the width (resp.~height) of $S$ is at most $\frac{20}{p_j} \le \frac{20}{p_i}$ (resp.~$\frac{20}{q_j} \le \frac{20}{q_i}$).
Together with \ref{cond:cell}, this implies \cref{decomp:rect-size}.

Towards proving \cref{decomp:grid-size}, recall that $|\cR_i| = n - i + 1$.
The second inequality of \cref{decomp:grid-size} then follows from \ref{cond:single} since $|\cR_i| \le p_i \cdot t$ and $|\cR_i| \le q_i \cdot t$.
For the first inequality, recall that $G_i$ contains at most $\frac{p_i-1}{C}$ wide columns. Thus, there are at least $(\frac{1}{2} - \frac{1}{C})\cdot (p_i-1) = \frac{9}{20} \cdot (p_i-1)$ disjoint pairs of consecutive non-wide columns, so $|\cR_i| \ge \frac{9}{20} \cdot (p_i-1) \cdot t = \frac{9}{40} \cdot (p_i-1) \cdot d$ by \ref{cond:pairs}.
The same holds for $q_i$, and \cref{decomp:grid-size} follows.

Lastly, consider a column $x$ of $G_i$.
Either $x$ is already present in $G_1$ and thus, $x$ contains at most $t$ points of $P$, or $x$ was created by merging two non-wide columns in the $j$-th step for some $j < i$.
In the latter case, the width of $x$ is at most $2 \frac{C}{p_j} \le \frac{40}{p_j} \le \frac{40}{p_i}$.
The same argument on the rows of $G_i$ shows \cref{decomp:super-wide} and wraps up the proof.
\end{proof}

Note that the width of the resulting merge sequence can be improved at the cost of its balancedness.
Namely, for any integer $k > 0$, the proof above goes through when we set $t = (4 + \tfrac 1k) c_\pi$ and $C = 4 + 16k$.

Furthermore, observe that the proof is constructive and it describes a polynomial-time procedure that outputs the distance-balanced merge sequence. 
In fact, we claim that the algorithm of Guillemot and Marx~\cite{GM_PPM} can be adapted to compute the distance-balanced merge sequence in $\fO(n)$ time, given access to the points ordered by both $x$- and $y$-coordinates.
We choose to omit the discussion of necessary implementation details since we only use the existence of distance-balanced merge sequences to prove the existence of solutions of small cost.

Our construction of distance-balanced merge sequences does not directly work with bounded twin-width and instead requires the point set to actually avoid a given pattern. However, we note that since each point set with twin-width $d$ avoids some $(d+1) \times (d+1)$ grid permutation $\pi$, and $c_\pi \in 2^{\fO(d)}$~\cite{jfox},  \cref{thm:decomposition} extends to arbitrary point sets with bounded twin-width (albeit with exponential twin-width blowup).

\begin{corollary}
	Let $P$ be a point set of twin-width $d$. Then $P$ has a distance-balanced $2^{\fO(d)}$-wide merge sequence.
\end{corollary}

Since the full generality of distance-balanced merge sequences is not needed for our proofs, we isolate the required properties in two corollaries.
First, we show that the sum of the dimensions of created rectangles is logarithmic.

\begin{corollary}\label{p:bal-rect-size}
	Let $\cR_1, \dots, \cR_n$ be a distance-balanced $d$-wide merge sequence. For $i \in \{1, 2, \dots, n-1\}$, let $Q_i$ be the rectangle created in step $i$, and let $w_i, h_i$ be its width and height. Then
	\[ \sum_{i=1}^{n-1} w_i + h_i \in \fO( d \log n). \]
\end{corollary}
\begin{proof}
	By \cref{decomp:rect-size}, we have $w_i \le \frac{20}{p_{i+1}}$. By \cref{decomp:grid-size}, we have $\frac{1}{p_{i+1}} \le \frac{d}{2(n-i)}$ so $w_i \le \frac{10\,d}{n-i}$. Similarly, $h_i \le \frac{10\,d}{n-i}$. Overall, we get
	\[
		\sum_{i=1}^{n-1} w_i + h_i \le \sum_{i=1}^{n-1} \frac{20\,d}{n-i} = 20\,d \cdot H_{n-1} \in \fO( d \log n ),
	\]
	where $H_k$ denotes the $k$-th Harmonic number.
\end{proof}

Second, we use distance-balanced merge sequences to construct a ``balanced'' gridding.

\begin{corollary}\label{p:bal-gridding}
	Let $\pi$ be a non-trivial permutation with Füredi-Hajnal limit $c_\pi$ and let $P \subset [0,1]^2$ be a $\pi$-avoiding set of $n$ points.

	For each real $m \in [2, \frac{n}{5c_\pi}]$, there exists a gridding $G$ of $[0,1]^2$ such that
	\begin{enumerate}[(I)]
		\item The number of rows (columns) in $G$ is at least $m$ and at most $3m$.\label[ipart]{item:bal-gridding:num-rows}
		\item Each row or column contains at most $5c_\pi$ non-empty cells.\label[ipart]{item:bal-gridding:points-per-row}
		\item Each row or column of width more than $\frac{40}{m}$ contains at most $5c_\pi$ points.\label[ipart]{item:bal-gridding:extra-wide-cols}
	\end{enumerate}
\end{corollary}
\begin{proof}
	Let $\cR_1, \dots, \cR_n$ be a distance-balanced $10c_\pi$-wide merge sequence of $P$ with griddings $G_1, G_2, \dots, G_n$. Choose $i$ such that $|\cR_i| = n - i + 1 = \ceil{5 c_\pi \cdot m}$, and set $G = G_i$.
	
	Let $p$ be the number of columns in $G$ and let $q$ be the number of rows in $G$. Using~\cref{decomp:grid-size}, $m \ge 2$, and $c_\pi \ge 1$ (since $\pi$ is non-trivial), we get
	\begin{align*}
		& \max(p,q) \le \frac{4}{9 c_\pi} (n-i+1) + 1 \le \frac{20}{9} m + \frac{4}{9 c_\pi} + 1 \le \frac{53}{18} m \le 3m, \text{ and}\\
		& \min(p,q) \ge \frac{1}{5 c_\pi} (n-i+1) \ge m.
	\end{align*}
	Hence, \cref{item:bal-gridding:num-rows} holds.
	
	\Cref{item:bal-gridding:points-per-row} follows directly from \cref{decomp:gridding}.
	
	\Cref{decomp:super-wide} implies that each column of width more than $\frac{40}{p}$ contains at most $5c_\pi$ points. Since $p, q \ge m$ (as shown above), \cref{item:bal-gridding:extra-wide-cols} follows.
\end{proof}

\section{\texorpdfstring{$k$-server}{k-server} on the line}\label{sec:k-server}


Consider a request sequence $X = (x_1, x_2, \dots, x_n) \in [0,1]^n$. We associate to $X$ the point set $P_X = \{ (\frac{i}{n},x_i) \mid i \in [n] \} \subset [0,1]^{2}$, i.e., the value of the request is the $y$-coordinate of the corresponding point and serving the requests means moving through the point set from left to right. We generally assume $X$ (and therefore $P_X$) to be in general position, and treat $X$ and $P_X$ interchangeably; this allows us to apply the definitions of distance-balanced merge sequences from \cref{sec:dist-bal} to request sequences. 

We start with the proof of our main upper bound.

\restateKServerTwinWidth*
\begin{proof}
Let $d = 5 c_\pi$ and let $f_t(n)$ denote the 
largest possible cost 
for a request sequence of length $n$ by exactly $d^t$ servers.
If $k \neq d^t$ for any positive integer $t$, we simply let the first $d^{\floor{\log_d k}}$ servers serve all requests in time $f_{\floor{\log_d k}}(n)$.
We show by induction on $t$ that 
$f_t(n) \le b_t \cdot d^t\cdot n^{1/(t+1)}$ where $b_t \ge 1$ is a constant dependent on $t$, specified later.

For $t = 0$, linear cost is clearly sufficient to handle all requests by one server and we can set $b_0 = 1$.
Assume $t \ge 1$.
If $n < d^t \cdot n^{1/(t+1)}$, we conclude that the bound holds trivially.
Otherwise, we apply \cref{p:bal-gridding} with $m = n^{1/(t+1)} \ge d^{t/(t+1)} \ge 2$ to get a gridding $G$.

Our strategy is as follows.
Each column of $G$ contains at most $d$ non-empty cells.
We split the servers into $d$ groups of $d^{t-1}$ servers and let each group serve a different non-empty cell in the first column of the gridding.
After the requests in the first column are handled, we move each group to a different non-empty cell in the second column of $G$.
We continue in the same fashion until all requests are handled.
In this way, each non-empty cell of $G$ is handled locally by a group of $d^{t-1}$ servers and these groups move around the gridding only in between the columns of $G$.

First, let us bound the cost of moving servers around in between the columns of $G$.
Denote by $q$ the number of columns in $G$.
We can upper bound the cost as
\begin{equation}
	d^t \cdot q \le d^t \cdot 3m = 3 d^t n^{1/(t+1)}\label{eq:kserver-cost-1}
\end{equation}
where the first inequality follows from \cref{item:bal-gridding:num-rows} of \cref{p:bal-gridding}.

In order to bound the cost of serving the individual non-empty cells, we count their contributions separately depending on their height.
We say that a row or a cell is \emph{extra-tall} if its height is larger than $\frac{40}{m}$.

Let $h_1, \dots, h_r$ be the heights of the extra-tall rows in $G$.
By \cref{item:bal-gridding:extra-wide-cols} of \cref{p:bal-gridding}, each extra-tall row contains at most $d$ points.
Therefore, the cost inside each cell of the $j$-th extra-tall row is at most $h_j \cdot d$, even with one server.
This makes the total cost occurring in all the extra-tall cells at most
\begin{equation}
\label{eq:kserver-cost-2}
\sum_{j=1}^r h_j \cdot d \le d.
\end{equation}

Consider now the remaining cells.
Let $N$ be the number of non-empty cells in $G$ that are not extra-tall (i.e., have height $\le \frac{40}{m}$), and let $n_1, \dots, n_N$ be the number of points in each of these cells.
The contribution of the non-extra-tall cells is, by induction, at most
\begin{equation*}
	\sum_{i=1}^{N} \frac{40}{m} f_{t-1}(n_i)
	\le \frac{40}{m} \sum_{i=1}^{N} b_{t-1} \cdot d^{t-1} \cdot n_i^{1/t}.
\end{equation*}

The function $g(x) = x^{1/t}$ is concave on $(0,\infty)$ and thus, by Jensen's inequality:
\begin{equation*}
	\frac{40 \cdot b_{t-1} \cdot d^{t-1}}{m} \sum_{i=1}^{N} n_i^{1/t}
	\le \frac{40 \cdot b_{t-1} \cdot d^{t-1}}{m} \cdot N \cdot \left(\frac{n}{N}\right)^{1/t}
	= \frac{40 \cdot b_{t-1} \cdot d^{t-1} \cdot N^{(t-1)/t} \cdot n^{1/t}}{m}.
\end{equation*}

Finally, \cref{item:bal-gridding:num-rows,item:bal-gridding:points-per-row} of \cref{p:bal-gridding} imply that $N \le 3dm$. Plugging in this inequality and $m = n^{1/{t+1}}$ yields
\begin{align}
	& \frac{40 \cdot b_{t-1} \cdot d^{t-1} \cdot N^{(t-1)/t} \cdot n^{1/t}}{m}
	\le 40 \cdot b_{t-1} \cdot d^{t-1} \cdot 3^{(t-1)/t} \cdot d^{(t-1)/t} \cdot m^{-1/t} \cdot n^{1/t}\notag\\
	& \le 120 \cdot b_{t-1} \cdot d^t \cdot n^{1/(t+1)}.\label{eq:kserver-cost-3}
\end{align}

Summing the costs \eqref{eq:kserver-cost-1}, \eqref{eq:kserver-cost-2} and \eqref{eq:kserver-cost-3}, we obtain
\begin{multline*}
f_t(n) \le 3 \cdot d^t \cdot n^{1/(t+1)} + d + 120 \cdot b_{t-1}\cdot d^t \cdot n^{1/(t+1)}
\le (120 \cdot b_{t-1} + 4) \cdot d^t \cdot n^{1/(t+1)}.
\end{multline*}

Thus, the desired inequality holds when we set $b_t = 120 \cdot b_{t-1} + 4$.
This recurrence solves to $b_t \in  \Theta(120^t)=\Theta(d^{t \cdot \log_d 120})$ which implies the claimed upper bound on $f_t(n)$.
\end{proof}

\subsection{Special upper bounds}

\subsubsection{231-avoiding inputs}\label{sec:k-server-231-ub}

\restateKServerAvUB*

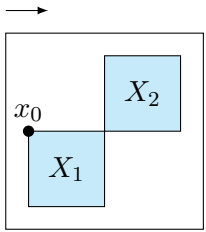
\begin{figure}[h]
	\centering
	\begin{tikzpicture}[scale=2,
			comp/.style=subrect
		]
		\def\mar{0.15}
		\draw (-\mar,-\mar) rectangle (1+\mar,1+\mar);
		\draw[-latex] (0-\mar,1+2*\mar) -- (0-\mar+0.28,1+2*\mar);
		\draw[comp] (0,0) rectangle node {$X_1$} (.5,.5);
		\draw[comp] (.5,.5) rectangle node {$X_2$} (1,1);
		\node[point] (x0) at (0,0.5) {};
		\node[above] at (x0) {$x_0$};
	\end{tikzpicture}
	\caption{Structure of a 231-avoiding sequence}\label{fig:231-av-structure}
\end{figure}

We use the specific structure of 231-avoiding inputs, as shown in \cref{fig:231-av-structure}. Let $X$ be a 231-avoiding sequence and let $x_0$ be the value of its first entry. Then there exist sequences $X_1, X_2$ such that $X = (x_0) \circ X_1 \circ X_2$, where $\circ$ denotes concatenation, and all entries in $X_1$ are at most $x_0$, and all entries in $X_2$ are at least $x_0$. We say $X$ \emph{decomposes} into $(x_0) \circ X_1 \circ X_2$.

An algorithm for serving a 231-avoiding sequence is given in \cref{alg:k-server-231}. The parameter $p$ is an arbitrary value that determines which sequences are small enough to be served with less than $k$ servers. The root call for a sequence of requests in $[0,1]$ will be $\Call{Serve}{k,X,\floor{|X|^{(k-1)/k}}, 0, 1}$, but note that 
$p$ is (re-)computed 
whenever we reduce the number of servers, i.e., call $\Call{Serve}{k-1, \dots}$.

\begin{figure}
	\begin{algorithmic}
		\Procedure{Serve}{$k, X, p, a, c$}
			\LComment{Serve $X$ with $k$ servers. Assume that all requests are in $[a,c]$. One server starts at $c$, the others at $a$. All servers end at $c$.}
			\If{$|X| = 0$}
				\State Move $k-1$ servers from $a$ to $c$
			\ElsIf{$k = 1$}
				\State Serve $X$, then move the single server to $c$
			\Else
				\State $b \gets $ first request in $X$
				\State Decompose $X$ into $(b) \circ X_1 \circ X_2$
				\If{$|X_1| \le p$}
					\State Move server from $a$ to $b$
					\State \Call{Serve}{$k-1, X_1, \floor{|X_1|^{(k-2)/(k-1)}}, a, b$}
				\Else
					\State Move server from $c$ to $b$
					\State \Call{Serve}{$k, X_1, p, a, b$}
					\State Move server from $b$ to $c$
				\EndIf
				\LComment{Now $k-1$ servers are stationed at $b$, and one server at $c$.}
				\State \Call{Serve}{$k, X_2, p, b, c$}
			\EndIf
		\EndProcedure
	\end{algorithmic}
	\caption{Algorithm for serving a 231-avoiding sequence with $k$ servers.}\label{alg:k-server-231}
\end{figure}

We now bound the total cost of a call to $\Call{Serve}{}$. Let $f_k(n,p,a,c)$ be the maximum cost of the algorithm when $X$ is a 231-avoiding sequence of length $n$ and the other parameters are as given. Observe that $f_k(n,p,a,c) = (c-a) f_k(n,p,0,1)$, hence we write $f_k(n,p) = f_k(n,p,0,1)$ and ignore the parameters $a$ and $c$ from now on.
We further write $g_k(n) = f_k(n, \floor{n^{(k-1)/k}})$.

The pseudocode directly yields the following:
\begin{align}
	& f_1(n,p) \le n+1,\label{eq:k-server-231-ub:base-k}\\
	& f_k(0,p) = (k-1).\label{eq:k-server-231-ub:base-n}
\end{align}

Moreover, $f_k(n,p)$ is bounded by the maximum of the following two terms over all $w_1, w_2, n_1, n_2$ with $w_1, w_2 \ge 0$, $w_1 + w_2 = 1$ and $n_1, n_2 \in \N$, $n_1 + n_2 = n$ (essentially, $w_1 = (b-a)$ and $w_2 = (c-b)$):
\begin{align}
	& w_1 + w_1\, g_{k-1}(n_1) + w_2\, f_k(n_2, p), & \text{ if } n_1 \le p,\label{eq:k-server-231-ub:step-unbal}\\
	& 2w_2 + w_1\, f_k(n_1, p) + w_2\, f_k(n_2, p), & \text{ if } n_1 > p.\label{eq:k-server-231-ub:step-bal}
\end{align}

Below, we show that $g_k(n) \in \fO( k^2 + k n^{1/k})$, thereby proving \cref{p:k-server-231-ub}. We first sketch the proof idea, ignoring constants and lower-order factors and treating $k$ as a constant.

Let $k \ge 2$. Compute an upper bound on $g_k(n) = f_k(n, \floor{n^{(k-1)/k}})$ by repeatedly applying \cref{eq:k-server-231-ub:step-unbal,eq:k-server-231-ub:step-bal} until all occurrences of $f_k$ are gone.
The occurrences of the first two terms $w_1 + w_1\, g_{k-1}(n_i)$ from \cref{eq:k-server-231-ub:step-unbal} add up to $\sum x_i + x_i\, g_{k-1}(n_i)$, where $\sum x_i = 1$ and $\sum n_i = n$. Using induction, the fact that $n_i \le p$, and Jensen's inequality, we can bound the result by $\fO(p^{1/(k-1)}) = \fO(n^{1/k})$.

It remains to compute the total contribution of the first term $w_2$ of \cref{eq:k-server-231-ub:step-bal}.
We simply bound $w_2 \le 1$ and argue that the term occurs at most $n/p \approx n^{1/k}$ times.
For this, consider the \emph{evaluation tree} of our calculation. This is a binary tree where each inner node is labeled $f_k(n',p)$ for some $n'$ and each leaf is labeled $g_{k-1}(n')$ for some $n'$.
A node $f_k(n',p)$ corresponding to an application of \cref{eq:k-server-231-ub:step-bal} must have two inner nodes as children, and must satisfy $n' \ge p$. It is easy to see that there can be no more than $n/p$ such nodes.

We now proceed with the formal proof.
For $n, p \in \N_+$, define a \emph{weighted $p$-bounded partition} of $n$ to be two sequences $(w_i)_{i \in [t]}$ and $(n_i)_{i \in [t]}$ such that $w_i \in [0,1]$ and $n_i \in [p]$ for all $i \in [t]$, as well as $\sum_{i=1}^t w_i = 1$ and $\sum_{i=1}^t n_i = n$.

Let $\phi_k(n,p)$ be the maximum value for $\sum_{i=1}^t w_i\, g_k(n_i)$ over all weighted $p$-bounded partitions $(w_i)_{i \in [t]}$, $(n_i)_{i \in [t]}$ of $n$.
\begin{observation}\label{p:k-server-231-ub:phi_obs}
	For $w_1, w_2 \in [0,1]$, $w_1 + w_2 = 1$ and $n_1 + n_2 = n$, we have $w_1 \phi_k(n_1,p) + w_2 \phi_k(n_2, p) \le \phi_k(n, p)$.
\end{observation}

\begin{lemma}\label{p:k-server-231-ub:ind1}
	For each $k \ge 2$ and $n, p \in \N$, we have $f_k(n,p) \le 2k + 2\frac{n}{p} - 2 + \phi_{k-1}(n, p)$.
\end{lemma}
\begin{proof}
	Fix $k$ and $p$. We prove the claim by induction on $n$. For $n = 0$, we have $f_k(0,p) = k-1$ by \cref{eq:k-server-231-ub:base-n} and are done.
	
	Let $n \ge 1$. Now either \cref{eq:k-server-231-ub:step-unbal} or \cref{eq:k-server-231-ub:step-bal} applies. Let $w_1, w_2 \in [0,1]$ with $w_1 + w_2 = 1$ and let $n_1 + n_2 = n$.
	In the first case, we have $n_1 \le p$ and
	\begin{align*}
		f_k(n,p) & = w_1 + w_1\, g_{k-1}(n_1) + w_2\, f_k(n_2, p)\\
			& \le w_1 + w_1\, g_{k-1}(n_1) + w_2 \cdot (2k + 2 \tfrac{n_2}{p} - 2) + w_2\, \phi_{k-1}(n_2, p) & (\text{induction})\\
			& \le 1 + w_2 \cdot (2k + 2 \tfrac{n_2}{p} - 3) + \phi_{k-1}(n, p) & (\text{\cref{p:k-server-231-ub:phi_obs}})\\
			& \le 1 + (2k + 2 \tfrac{n_2}{p} - 3) + \phi_{k-1}(n, p) & (2k - 3 \ge 0)\\
			& = 2k + 2 \tfrac{n_2}{p} - 2 + \phi_{k-1}(n, p).
	\end{align*}
	
	In the second case, for some $n_1 > p$, we have
	\begin{align*}
		f_k(n,p) & = 2 w_2 + w_1\, f_k(n_1, p) + w_2\, f_k(n_2, p)\\
			& \begin{aligned}
				\le 2 w_2 + 2k - 2 + 2 w_1\tfrac{n_1}{p} + w_1\, \phi_{k-1}(n_1, p)\\
				+ 2 w_2\tfrac{n_2}{p} + w_2\, \phi_{k-1}(n_2, p)
			\end{aligned} & (\text{induction})\\
			& \le 2k + w_1( 2 \tfrac{n_1}{p} - 2 ) + 2 w_2\tfrac{n_2}{p} + \phi_{k-1}(n, p) & (\text{\cref{p:k-server-231-ub:phi_obs}})\\
			& \le 2k + 2 \tfrac{n_1}{p} - 2 + 2 \tfrac{n_2}{p} + \phi_{k-1}(p, n) & (2\tfrac{n_1}{p} - 2 \ge 0)\\
			& = 2k + 2 \tfrac{n}{p} - 2 + \phi_{k-1}(n, p).\tag*{\qedhere}
	\end{align*}
\end{proof}

\begin{lemma}\label{p:k-server-231-ub:ind2}
	For each $k, n \in \N$, we have $g_k(n) \le 4 k n^{1/k} + k(k-1)$.
\end{lemma}
\begin{proof}
	For $k = 1$, we have $g_1(n) = n + 1 \le 4n$ from \cref{eq:k-server-231-ub:base-k}.
	
	For $k \ge 2$, we first bound $\phi_{k-1}(n,p)$ with $p = \floor{n^{(k-1)/k}}$. For some weighted $p$-bounded partition $(w_i)_{i \in [t]}$, $(n_i)_{i \in [t]}$ of $n$, we have
	\begin{align*}
		\phi_{k-1}(n,p) & \le \sum_{i=1}^t w_i \cdot g_{k-1}(n_i)\\
			& \le \sum_{i=1}^t w_i \cdot (4 (k-1) n_i^{1/(k-1)} + (k-1)(k-2)) & (\text{induction})\\
			& \le (k-1)(k-2) + 4 (k-1) \sum_{i=1}^t w_i \cdot n_i^{1/(k-1)}. & (\textstyle \sum w_i = 1)
	\end{align*}
	
	Since $\sum w_i = 1$ and $n_i \le p$ for each $i \in [t]$, Jensen's inequality yields
	\begin{align*}
		\sum_{i=1}^t w_i \cdot n_i^{1/(k-1)} \le \left( \sum_{i=1}^t w_i \cdot n_i \right)^{1/(k-1)} \le p^{1/(k-1)} \le n^{1/k}.
	\end{align*}
	
	Thus, we have $\phi_{k-1}(n,p) \le (k-1)(k-2) + 4 (k-1) n^{1/k}$.
	
	Together with $n/p \le 2 n^{1/k}$, we get
	\begin{align*}
		g_k(n) & = f_k(n, n^{(k-1)/k}) \le 2k + 2\tfrac{n}{p} - 2 + \phi_{k-1}(n, p) & (\text{\cref{p:k-server-231-ub:ind1}})\\
			& \le 2k + 4 n^{1/k} - 2 + (k-1)(k-2) + 4 (k-1) n^{1/k}\\
			& \le 4k n^{1/k} + k(k-1). \tag*{\qedhere}
	\end{align*}
\end{proof}

\Cref{p:k-server-231-ub:ind2} implies \cref{p:k-server-231-ub}.

\subsubsection{Proper subclasses of \texorpdfstring{$\Av(231)$}{Av(231)}}\label{sec:k-server-sep-subclasses}

\restateKServerSepSubUB*
\begin{proof}
	An algorithm for serving a $(231,\pi)$-avoiding sequence is given in \cref{alg:k-server-sep-subclasses-ind}.
	We prove that \textsc{Serve}($\pi, X, a, c$) handles any  $(231,\pi)$-avoiding sequence $X$ by $2^{|\pi|+1}$ servers with cost $(c-a) \cdot2^{|\pi|+2}$ by induction on $|X|+|\pi|$.
	Clearly, the claim holds if $|\pi| \le 1$ or $|X|=0$.
	Otherwise, let $X = (b) \circ X_1 \circ X_2$ be the decomposition of $X$ and let $\pi = (p) \circ \alpha \circ \beta$ be the decomposition of $\pi$.
	Note that $X_1$ avoids $\alpha$ or $X_2$ avoids $\beta$, otherwise $X$ contains $\pi$.
	
	\begin{figure}
		\begin{algorithmic}
			\Procedure{Serve}{$\pi, X, a, c$}
				\LComment{Serve $X$ that avoids $231$ and $\pi$ with $2^{|\pi|+1}$ servers. Assume that all requests are in $[a,c]$. The servers start and end evenly split between $a$ and $c$.}
				\If{$X \neq \emptyset$}
					\State $b \gets $ first request in $X$
					\State $p \gets $ first element of $\pi$
					\State Decompose $X$ into $(b) \circ X_1 \circ X_2$ and $\pi$ into $(p) \circ \alpha \circ \beta$
					\If{$X_1$ avoids $\alpha$}
						\State Move $2^{|\alpha|}$ servers from $a$ to $b$
						\State \Call{Serve}{$\alpha, X_1, a, b$}\Comment{$2^{|\alpha|}$ servers at $b$, $2^{|\pi|}-2^{|\alpha|} \ge 2^{|\alpha|}$ at $a$}
						\State Move $2^{|\pi|}-2^{|\alpha|}$ servers from $a$ to $b$
						\State \Call{Serve}{$\pi, X_2, b, c$}\Comment{$2^{|\pi|}$ servers at both $b$ and $c$}
						\State Move $2^{|\pi|}$ servers from $b$ to $a$
					\Else\Comment{$X_2$ avoids $\beta$}
						\State Move all $2^{|\pi|}$ servers from $c$ to $b$
						\State \Call{Serve}{$\pi, X_1, a, b$}\Comment{$2^{|\pi|}$ servers at both $a$ and $b$}
						\State Move $2^{|\beta|}$ servers from $b$ to $c$
						\State \Call{Serve}{$\beta, X_2, b, c$}\Comment{$2^{|\beta|}$ servers at $c$, $2^{|\pi|}-2^{|\beta|} \ge 2^{|\beta|}$ at $b$}
						\State Move $2^{|\pi|} -2^{|\beta|}$ servers from $b$ to $c$
					\EndIf
				\EndIf
			\EndProcedure
		\end{algorithmic}
		\caption{Algorithm for serving a (231, $\pi$)-avoiding sequence with $2^{|\pi|+1}$ servers.}\label{alg:k-server-sep-subclasses-ind}
	\end{figure}
	
	First assume that $X_1$ avoids $\alpha$.
	In that case, the call $\textsc{Serve}(\alpha, X_1, a, b)$ serves $X_1$ with cost $(b-a) \cdot2^{|\alpha|+2}$ by induction since $|\alpha|<|\pi|$.
	The call $\textsc{Serve}(\pi, X_2, b, c)$ serves $X_2$ with cost $(c-b) \cdot 2^{|\pi|+2}$ by induction since $|X_2|<|X|$.
	Moreover, the cost of moving servers between these calls is $(b-a) \cdot 2 \cdot 2^{|\pi|}$.
	Therefore, the total cost is at most
	\begin{align*}
		&(b-a) \cdot2^{|\alpha|+2} + (c-b) \cdot 2^{|\pi|+2} + (b-a) \cdot 2 \cdot 2^{|\pi|}\\
		&\le (b-a) \cdot2^{|\pi|+1} + (b-a) \cdot 2^{|\pi|+1} + (c-b) \cdot 2^{|\pi|+2} \\
		&\le (b-a) \cdot2^{|\pi|+2} + (c-b) \cdot 2^{|\pi|+2} \le (c - a)\cdot 2^{|\pi|+2}.
	\end{align*}
	
	An analogous argument proves the case when $X_2$ avoids $\beta$ and the claim follows.
\end{proof}

\subsubsection{\texorpdfstring{$t$}{t}-separable permutations}\label{sec:k-server-sep-ub}

\restateKServerSepUB*

A sequence $X$ of pairwise distinct requests is \emph{$t$-separable} if it is order-isomorphic to a $t$-separable permutation $\pi$. Then, $X_1 \circ X_2 \circ \dots \circ X_t = X$ is a \emph{decomposition} of $X$ if the subsequences $X_i$ correspond to the parts of the $t$-separable permutation $\pi$. It follows that, for each $i \neq j$, all values in $X_i$ are smaller than all values in $X_j$ or vice versa.

Fix some $t \ge 2$. We describe an algorithm $\Call{Serve}{\ell, X, p, a, b}$ to serve $t$-separable input sequence $X \in [a,b]^n$ with $2^\ell$ servers, $\ell \ge 0$ (see \cref{alg:k-server-sep}). Assume $\ceil{2^{\ell-1}}$ servers start and end at $a$ and $\floor{2^{\ell-1}}$ servers start and end at $b$.
Parameter $p$ is a threshold to lower the number of servers, like in \cref{sec:k-server-231-ub}. For the initial call, let $p = \floor{n^{\ell/(\ell+1)}}$.

\begin{figure}[h]
	\begin{algorithmic}
		\Procedure{Serve}{$\ell, X, p, a, b$}
			\LComment{Serve a $t$-separable request sequence $X \in [a,b]^n$, with $2^\ell \ge 1$ servers. Half of the servers each start at $a$ and at $b$. Half of the servers each end at $a$ and at $b$}
			\If{$\ell = 0$}
				\State Greedily serve $X$
			\Else
				\State Let $X_1 \circ X_2 \circ \dots \circ X_t$ be a decomposition of $X$
				\For{$i \in [t]$}
					\State Let $a_i = \min(X_i), b_i = \max(X_i)$.
					\If{$|X_i| > p$}
						\State Move $2^{\ell-1}$ servers to $a_i$ and $2^{\ell-1}$ servers to $b_i$
						\State \Call{Serve}{$\ell, X_i, p, a_i, b_i$}
					\Else
						\If{there is exactly one $i^*$ with $|X_{i^*}| > p$}
							\If{Values in $X_i$ are smaller than values in $X_{i^*}$}
								\State Move $2^{\ell-2}$ servers from $a$ to $a_i$ and $2^{\ell-2}$ servers from $a$ to $b_i$
							\Else
								\State Move $2^{\ell-2}$ servers from $b$ to $a_i$ and $2^{\ell-2}$ servers from $b$ to $b_i$
							\EndIf
						\Else
							\State Move $2^{\ell-1}$ servers to $a_i$ and $2^{\ell-1}$ servers to $b_i$
						\EndIf
						\State \Call{Serve}{$\ell-1, X_i, |X_i|^{\ell/(\ell+1)}, a_i, b_i$}¸
					\EndIf
				\EndFor
				\State Move $2^{\ell-1}$ servers to $a$ and $2^{\ell-1}$ servers to $b$
			\EndIf
		\EndProcedure
	\end{algorithmic}
	\caption{An algorithm for serving a $t$-separable request sequence.}\label{alg:k-server-sep}
\end{figure}

The idea of the algorithm is similar to the one in \cref{sec:k-server-231-ub}. For each block $X_i \in [a_i,b_i]^{n_i}$, we first move half of the servers to $a_i$ and half of the servers to $b_i$. If $|X_i| > p$, we serve $X_i$ recursively with all servers. Otherwise, we ignore half of the servers and call $\Call{Serve}{\ell-1, X_i, p_i, a_i, b_i}$, where $p_i = \floor{|X_i|^{\ell/(\ell+1)}}$.

Observe that we often move servers and end up not using them in the recursive call. Usually, this does not affect our cost bound, but there is one special case where we need to be more careful. If there is only one block $X_i$ with $|X_i| > p$, we handle the other blocks as follows. If the values in $X_j$ are smaller than the values in $X_i$, we only use the servers with positions below $X_i$ to serve $X_j$. Otherwise, we only use the servers above $X_i$. Observe that no servers cross $X_i$, except in the recursive call for $X_i$. Hence, we indeed always have the necessary $2^{\ell-1}$ servers below and above $X_i$.

We now bound the cost of the algorithm. Fix $t \ge 2$ and let $f_\ell(n,p)$ be the maximum cost of the algorithm when $X \in [0,1]^n$ and $p, \ell$ are as given.
As observed earlier in \cref{sec:k-server-231-ub}, the cost of $\Call{Serve}{\ell, X, p, a, b}$ is $(b-a) f_\ell(n,p)$.

We further write $g_\ell(n) = f_\ell(n, \floor{n^{\ell/(\ell+1)}})$, and let $\phi_k(n,p)$ be the maximum value for $\sum_{i=1}^t w_i\, g_k(n_i)$ over all weighted $p$-bounded partitions $(w_i)_{i \in [t]}$, $(n_i)_{i \in [t]}$ of $n$.

\begin{lemma}\label{p:k-server-sep-ub-ind1}
	$f_\ell(n,p) \le 2^\ell (t+1) (\frac{n}{p}+1) + \phi_{\ell-1}(n,p)$ for each $\ell \ge 1$.
\end{lemma}
\begin{proof}
	Let $\alpha_\ell = 2^\ell (t+1)$. We need to show that $f_\ell(n,p) \le \alpha_\ell + \alpha_\ell \tfrac{n}{p} + \phi_{\ell-1}(n,p)$.
	
	Consider a call \Call{Serve}{$\ell, X, p, a, b$}. Without loss of generality, let $a = 0$ and $b = 1$. Let $X_1 \circ X_2 \circ \dots \circ X_t$ be a decomposition of $X$, and let $n_i = |X_i|$, $w_i = \max(X_i)-\min(X_i)$.
	Let $I \subseteq [t]$ be the set of indices of \emph{small} blocks, i.e., let $I = \{i \in [t] \mid |X_i| \le p\}$. Let $J = [t] \setminus I$ be the set of \emph{large} blocks.
	
	Apart from recursive calls, the algorithm moves each server at most $(t+1)$ times (once for each block, and back to $a$ and $b$ at the end), for a total cost of at most $(t+1) \cdot 2^\ell = \alpha_\ell$.
	If there is exactly one large block $X_{i^*}$, then no server ever crosses $X_{i^*}$ (outside of recursive calls), so the total cost is at most $(1-w_{i^*}) \alpha_\ell$.
	
	We now prove our claim by induction on $n$. If $n \le p$, then there are no large blocks, so we have $f_\ell(n,p) = (t+1) \cdot 2^\ell + \phi_{\ell-1}(n,p)$ and are done.
	
	Now suppose $n > p$, so there may be large blocks.
	Let $h_\ell(i)$ be the cost of serving $X_i$ (after moving servers in position). If $n_i \le p$, then $h_\ell(i) = w_i \cdot g_{\ell-1}(n_i) \le w_i \cdot \phi_{\ell-1}(n_i)$. Otherwise, we have
	\begin{align*}
		h_\ell(i) = w_i f_\ell(n_i, p) & \le w_i \cdot \alpha_\ell (\tfrac{n_i}{p}+1) + w_i \cdot \phi_{\ell-1}(n,p) & (\text{induction})\\
			& \le w_i \alpha_\ell + \alpha_\ell \tfrac{n_i}{p} + (w_i-1)\alpha_\ell + w_i \cdot \phi_{\ell-1}(n,p). & (\tfrac{n}{p} \ge 1)
	\end{align*}
	
	If $J = \{i^*\}$, i.e., there is precisely one large block $X_{i^*}$, we have
	\begin{align*}
		f_\ell(n,p) & \le (1-w_{i^*}) \alpha_\ell + \sum_{i = 1}^t h_\ell(i)\\
			& \le (1-w_{i^*}) \alpha_\ell + w_{i^*} \alpha_\ell + w_{i^*} \alpha_\ell \tfrac{n_{i^*}}{p} + w_{i^*} \cdot \phi_{\ell-1}(n,p) + \sum_{i \in I} w_i \cdot \phi_{\ell-1}(n_i,p)\\
			& \le \alpha_\ell + \alpha_\ell \tfrac{n}{p} + \phi_{\ell-1}(n,p).
	\end{align*}
	
	Otherwise, let $W_J = \sum_{i \in J} w_i$, and we have
	\begin{align*}
		f_\ell(n,p) & \le \alpha_\ell + \sum_{i = 1}^t h_\ell(i)\\
			& \le \alpha_\ell + \sum_{i \in I} w_i \phi_{\ell-1}(n_i,p) + \sum_{i \in J} w_i \alpha_\ell + \alpha_\ell \tfrac{n_i}{p} + (w_i-1)\alpha_\ell + w_i \cdot \phi_{\ell-1}(n,p)\\
			& \le \alpha_\ell + \alpha_\ell \tfrac{n}{p} + \phi_{\ell-1}(n,p) + W_J \alpha_\ell + (W_J - |J|) \alpha_\ell\\
			& \le \alpha_\ell + \alpha_\ell \tfrac{n}{p} + \phi_{\ell-1}(n,p).
	\end{align*}
	The last inequality uses $2W_J \le |J|$, which follows from $|J| \neq 1$.
\end{proof}

\begin{lemma}\label{p:k-server-sep-ub-ind2}
	$g_\ell(n) \le 2^{\ell+2} (t+1) (n^{1/(\ell+1)} + 1)$ for all $\ell \ge 0$.
\end{lemma}
\begin{proof}
	We trivially have $g_0(n) = n$.
	
	Suppose now $\ell \ge 1$ and let $p = \floor{n^{\ell/(\ell+1)}}$.
	By \cref{p:k-server-sep-ub-ind1}, for some weighted $p$-bounded partition $(w_i)_{i \in [u]}$, $(n_i)_{i \in [u]}$ of $n$, we have
	\begin{align*}
		g_\ell(n) = f_\ell(n,p) & \le 2^\ell (t+1) ( \tfrac{n}{p} + 1 ) + \phi_{\ell-1}(n,p)\\
			& \le 2^\ell (t+1) + 2^\ell (t+1) \cdot 2 n^{1/(\ell+1)} + \sum_{i=1}^u w_i \cdot 2^{\ell+1} (t+1) (n_i^{1/\ell} + 1)\\
			& \le 2^\ell (t+1) + 2^{\ell+1} (t+1) n^{1/(\ell+1)} + 2^{\ell+1} (t+1) + 2^{\ell+1} (t+1) \cdot p^{1/\ell}\\
			& \le 2^{\ell+2} (t+1) + 2^{\ell+2} (t+1) n^{1/(\ell+1)}.
	\end{align*}
	The third inequality uses $n_i \le p$ and Jensen's inequality.
\end{proof}

\Cref{p:k-server-sep-ub-ind2} with $\ell = \floor{\log k}$ immediately implies \Cref{p:k-server-sep-ub}.

\subsection{Lower bounds}

In the following, if $\alpha \in \R$ and $X = (x_1, x_2, \dots, x_n)$, we write $\alpha X = (\alpha x_1, \alpha x_2, \dots, \alpha x_n)$ and $\alpha + X = (\alpha + x_1, \alpha + x_2, \dots, \alpha + x_n)$. If $X, Y$ are sequences, we write $X \circ Y$ for the concatenation of $X$ and $Y$.

\subsubsection{Bounded twin-width}\label{sec:k-server-tww-lb}

\begin{figure}
	\def\margin{.1}
	\def\ksbox{%
		\draw (0-\margin,0-\margin) rectangle (1+\margin,1+\margin);
		\draw[-latex] (0-\margin,1+2*\margin) -- (0-\margin+0.2,1+2*\margin);
	}
	\tikzset{
		point/.style={fill, circle, inner sep = 1pt}
	}
	
	\centering
	\begin{subfigure}{.25\textwidth}
		\centering
		\begin{tikzpicture}[scale=2.8]
			\ksbox
			\input{figs/k-server-lb-1.tex}
		\end{tikzpicture}
		\caption{$X_1^1(10)$}\label{sfig:k-server-lb-sep-1}
	\end{subfigure}%
	\begin{subfigure}{.25\textwidth}
		\centering
		\begin{tikzpicture}[scale=2.8, point/.append style={inner sep = 0.5pt}]
			\ksbox
			\input{figs/k-server-lb-2.tex}
		\end{tikzpicture}
		\caption{$X_2^1(5)$}\label{sfig:k-server-lb-sep-2}
	\end{subfigure}%
	\begin{subfigure}{.25\textwidth}
		\centering
		\begin{tikzpicture}[scale=2.8]
			\ksbox
			\input{figs/k-server-lb-tww-1.tex}
		\end{tikzpicture}
		\caption{$X_1^3(10)$}\label{sfig:k-server-lb-tww-1}
	\end{subfigure}%
	\begin{subfigure}{.25\textwidth}
		\centering
		\begin{tikzpicture}[scale=2.8, point/.append style={inner sep = 0.5pt}]
			\ksbox
			\input{figs/k-server-lb-3.tex}
		\end{tikzpicture}
		\caption{$X'_2(5)$}\label{sfig:k-server-lb-213}
	\end{subfigure}
	\caption{Lower bound constructions.}
\end{figure}
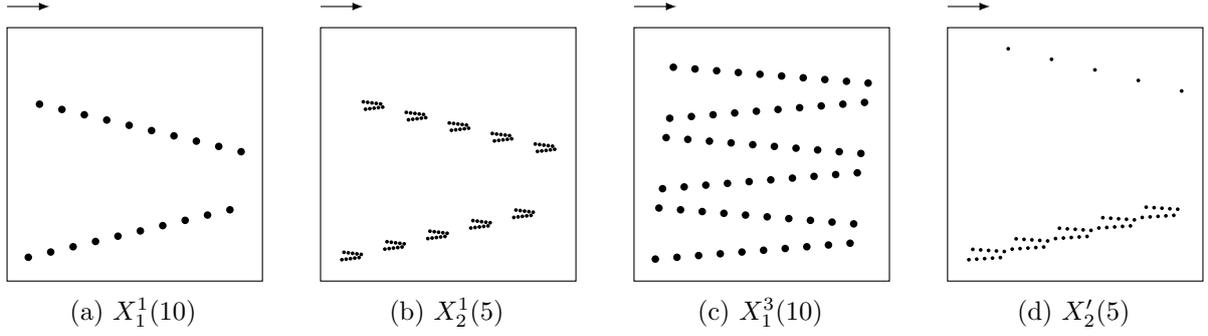

\restateKServerTwwLB*

We first illustrate the idea for $k \in \{1,2\}$ and $d=1$. \Cref{sfig:k-server-lb-sep-1} sketches a sequence $X$ with $n$ values (properly defined below).
The sequence is separable, i.e., has twin-width~1.
The cost for one server is clearly $\Theta(n)$, since the server has to switch between the left and right part in every step. On the other hand, two servers can serve the sequence with cost $\fO(1)$, by positioning one server at each part.

To construct a sequence that is hard for two servers, we take $X$ and replace each point with a copy of $X$ itself, scaled down by a factor of roughly $\frac{1}{n}$. The resulting sequence, shown in \cref{sfig:k-server-lb-sep-2}, is still separable (it is order-isomorphic to the inflation of $X$ with $2n$ copies of itself). Now consider serving it with two servers. For each small copy of $X$, we essentially have two choices. Either use both servers, which means that afterwards, at least one server has to be moved over to the other side for the next small copy of $X_1(n)$, for a cost of $\Omega(1)$; or use only one server, which costs $\Omega(\frac{1}{n})$ for each point, for a total cost of again $\Omega(1)$. Overall, the cost is $\Omega(n)$ for a point set of size $n^2$. For larger $k$, we recursively inflate the construction.

If $d > 1$, we can make the construction a little more ``efficient'' by using multiple interleaved copies of $X$ (\cref{sfig:k-server-lb-tww-1}) as the base construction, which is again recursively inflated.
We proceed with the formal definition and proof of \cref{p:k-server-tww-lb}.


Let $X$ be a sequence of reals in~$[0,1]$.
For $k,d \in \N$, we define the sequence $S_d(X,k)$ as follows. Let $\alpha = \frac{1}{4dk}$, let $Y^j_i = \alpha( 4kj + i + X )$, and let $Z^j_i = \alpha ( 4jk+3k - i - X )$ for $i \in \{0,1,\dots,k-1\}$ and $j \in \{0,1,\dots,d-1\}$.
We define
\begin{align*}
	S_d(X,k) &= W_0  \dots \circ W_{k-1} \text{, where}\\
	W_i &= Y^0_i \circ Z^0_i \circ \dots \circ Y^{d-1}_i \circ Z^{d-1}_i.
\end{align*}

We call each $W_i$ an \emph{epoch} and each $Y_i^j$, $Z_i^j$ a \emph{block} of $S_d(X,k)$. Each block is a copy of $X$ that is scaled down by $\alpha$ and shifted, so $S_d(X,k)$ contains $2dk|X|$ values in total.


\begin{lemma}
We have $\tww(S_d(X,k)) = \max(\tww(X),d)$.
\end{lemma}
\begin{proof}
It suffices to prove that $\tww(S_d(X,k)) = d$ when $|X| = 1$.
The claim then follows by the behavior of twin-width with respect to inflations (see \cref{obs:tww}).

Without loss of generality, let $X = (\frac{1}{2})$, so each block consists of a single value. Observe that $S_d(X,k)$ is order-isomorphic to a $k \times 2d$ grid permutation, where each epoch corresponds to a column, and the rows alternate between being increasing ($Y_0^j, Y_1^j, \dots, Y_{k-1}^j$) and decreasing ($Z_0^j, Z_1^j, \dots, Z_{k-1}^j$). See \cref{sfig:k-server-lb-tww-1} for an example.


We first merge each pair of neighboring rows into a single rectangle, gradually from the end of the sequence.
More precisely, we start by merging $Y_{k-1}^j$ and $Z_{k-1}^j$ into a rectangle $R_j$, for each $j \in \{0,1,\dots,d-1\}$. Then, for $i = k-2, k-3, \dots, 0$, we merge each $Z_{k-1}^j$ into $R_j$ and then each $Y_{k-1}^j$ into $R_j$.
At any time, each rectangle $R_j$ is homogeneous to every non-rectangle value (but may be non-homogeneous with any other rectangle $R_{j'}$), so the rectangle family is always $d$-wide.

In the end, we simply merge the $d$ rectangles (each corresponding to a double row) in any order.
\end{proof}

Let $X^d_1(n) = S_d(X,n)$, where $X = (\tfrac{1}{2})$. For $t \ge 2$, let $X^d_t(n) = S(X^d_{t-1}(n), n)$.

\begin{lemma}\label{p:k-server-tww-constr}
	Serving $X^d_t(n)$ with strictly less than $(2d)^t$ servers has cost at least $n / (8d)^t$.
\end{lemma}
\begin{proof}
Serving $X^d_1(n)$ with $(2d)^1-1=2d-1$ servers costs at least $1/(2d)$ per epoch, for a total of $\frac{1}{2d} n \ge n/(8d)$.
	
Now let $t \ge 2$ and consider a solution serving $X^d_t(n)$ with at most $(2d)^t-1$ servers.
We say that an epoch is \emph{saturated} if all of its blocks are touched by at least $(2d)^{t-1}$ servers.
Let $u$ be the number of saturated epochs.

If $u \leq \frac{1}{2} n$, then in at least $\frac{1}{2}n$ epochs there is a block served by less than $(2d)^{t-1}$ servers.
By induction, each of these blocks incurs a cost of $\frac{1}{4nd} \frac{n}{(8d)^{t-1}} = 1/(4d \cdot (8d)^{t-1})$, for a total cost of $n/(8d)^t$, as desired.

Otherwise, $u > \frac{1}{2} n$.
In each of the $u$ satured epochs, there is at least one server that touches two different blocks.
The distance between different blocks in the same epoch is at least $\frac{1}{4d}$ and thus, each saturated epoch incurs a cost of at least $\frac{1}{4d}$, for a total cost of $\frac{1}{4d} u > \frac{1}{8d} n \ge \frac{n}{(8d)^t}$.
\end{proof}

To prove \cref{p:k-server-tww-lb}, let $m = \frac{1}{2d} \floor{n^{1/t}}$ and $t = \floor{\log_{2d} k} + 1$.
Then $X^d_t(m)$ is of length $(2dm)^t \le n$, and serving it with $k < (2d)^t$ servers costs, by \cref{p:k-server-tww-constr},
\begin{align*}
	\frac{m}{(8d)^t} = \frac{1}{2d} \cdot \frac{1}{ (8d)^{\floor{\log_{2d} k} + 1}} \floor{n^{1/t}} \ge \frac{1}{16d^2 k^{1 + \log_{2d} 4}} \floor{n^{1/t}} \ge \frac{1}{16d^2 k^{3}} \floor{n^{1/t}}.
\end{align*}

\subsubsection{231-avoiding inputs}\label{sec:k-server-231-lb}

\restateKServerAvLB*

The construction from \cref{sec:k-server-tww-lb} with $d=1$ is separable (\cref{sfig:k-server-lb-sep-1,sfig:k-server-lb-sep-2}), but does not avoid any of the patterns 132, 213, 231, or 312. In fact, $X_k^1(2)$ contains every separable permutation of length~$k$. However, we can modify the construction to avoid one of those patterns, at the cost of a worse lower bound (which is necessary because of \cref{p:k-server-231-ub}). The recursive construction is modified as follows: Instead of inflating every value in the sequence, we only inflate the lower half (see \cref{sfig:k-server-lb-213}). A formal description follows.

Let $X$ be a sequence of reals in $[0,1]$. For $k \in \N$, define the sequence $S(X,k)$ as follows. Let $\alpha = \frac{1}{4k}$, let $Y_i = \alpha i + \alpha X$, and let $z_i = 1 - \alpha i$ for $i \in \{0,1,\dots,k-1\}$. Let
\begin{align*}
	S'(X,k) = Y_0 \circ (z_0) \circ \dots \circ Y_{k-1} \circ (z_{k-1}).
\end{align*}
Note that $S'(X,k)$ contains $k|X|$ values in $[0,\frac{1}{4}]$ and $k$ values in $[\frac{3}{4},1]$. As in \cref{sec:k-server-tww-lb}, call $Y_0, Y_1, \dots, Y_{k-1}$ \emph{blocks}.

We now argue that $S'(X,k)$ avoids 231 if $X$ avoids 231. Suppose $S'(X,k)$ contains an occurrence $(a,b,c)$ of 231. If any two of the three values are in a single block, then the third must be in the same block, so $X$ contains 231, a contradiction. If the three values are in two or three distinct blocks, then 231 must be the \emph{sum} of two smaller permutations, which is not true.
We also clearly cannot have $a = z_i$ or $c = z_i$ for any $z_i$.
The only remaining possibility is that $b = z_i$ for some $z_i$. Then $a$ is in some block $Y_j$ with $j \le i$, and $c$ is in some block $Y_\ell$ with $i < \ell$. But then $a < c$, a contradiction.

Let $X'_1(n) = S'(X,n)$, where $X = (\frac{1}{2})$. For $t \ge 2$, let $X'_t(n) = S(X'_{t-1}(n), n)$. By the discussion above, $X'_t(n)$ avoids 231 for all $t, n$.

\begin{lemma}\label{p:k-server-231-constr}
	Serving $X'_k(n)$ with at most $k$ servers has cost at least $n / 4^k$.
\end{lemma}
\begin{proof}
	Serving $X_1(n)$ with $2^1-1=1$ server costs at least $\frac{1}{2}$ per move, for a total of $\frac{1}{2} (2n-1) \ge n/4$.
	
	Now let $k \ge 2$ and consider a solution serving $X'_k(n)$ with $k$ servers.
	Let $u$ be the number of blocks that are touched by all $k$ servers. For each such block $Y_i$, immediately after serving the block, we have to move a server to $z_i$, with cost at least~$\frac12$.
	
	For each of the $(n-u)$ blocks that are touched by at most $k-1$ servers, by induction, the block incurs a cost of $\frac{1}{4n} \cdot n/4^{k-1}$. The overall cost is thus at least
	\begin{align*}
		& \frac12 \cdot u + \frac{1}{4n} \cdot \frac{n}{4^{k-1}} \cdot  (n-u)\\
		& = \left( \frac12 - \frac{1}{4^k} \right) u + \frac{n}{4^k} \ge \frac{n}{4^k}.\tag*{\qedhere}
	\end{align*}
\end{proof}

To prove \cref{p:k-server-312-lb}, let $m = \frac12 \floor{n^{1/k}}$. The sequence $X'_k(m)$ has length at most $(2m)^k \le n$, and serving it with $k$ servers costs (by \cref{p:k-server-231-constr}):
\begin{align*}
	\frac{1}{4^k} m = \frac12 \cdot \frac{1}{4^k} \floor{n^{1/k}}.
\end{align*}

\section{Euclidean TSP}\label{sec:mst}

 In this section, we show upper and lower bounds on the optimum euclidean TSP tour of a point set. As mentioned in \S\,\ref{sec1}, several characteristics of point sets are known to be within a constant factor of the TSP optimum, and it will be helpful to use them in our proofs when showing asymptotic bounds. 

For a point set $P$, let $\MST(P)$ be the cost of the \emph{euclidean minimum spanning tree} on $P$, i.e., the minimum spanning tree on $G_P$, where $G_P$ is the complete graph on $P$ where each edge has weight equal to the distance between its two endpoints.
Further, let $\NN(P) = \sum_{x \in P} d_x$, where $d_x$ is the minimum distance between $x$ and a different point from $P$ (i.e., the \emph{nearest neighbor}).
It is easy to see that \[ \tfrac12 \NN(P) \le \MST(P) \le \TSP(P) \le 2 \MST(P). \]
Let $\MStT(P)$ denote the minimum $\MST(P')$ over all supersets $P' \supseteq P$, i.e., the \emph{minimum euclidean Steiner tree}. We have
\[ \MStT(P) \le \MST(P) \le 1.22 \MStT(P), \] where the first inequality is trivial and the second is due to Chung and Graham~\cite{ChungGraham1985}.


\restateThmGeneralMST*
\begin{proof}
	Let $n$ be the size of $P$. We show $\MST(P) \in \fO( c_\pi \log n )$.
	
	Given a merge sequence $\cR_1, \dots, \cR_n$ of $P$, we construct a spanning tree as follows. Replay the merge sequence, and whenever two rectangles $Q_1$ and $Q_2$ are merged, connect an arbitrary point in $Q_1$ with an arbitrary point in $Q_2$. Observe that after step~$i$, the points in every rectangle $Q \in \cR_i$ are connected via a spanning tree. Thus, we obtain a spanning tree $T$ of $P$ at the end.
	
	To control the total length of $T$, use a distance-balanced $10 c_\pi$-wide merge sequence (Theorem~\ref{thm:decomposition}). 
	Edge $e_i$ added in step $i$ of the construction is contained in rectangle $S_i$ that is newly created in step $i$. Thus, the length of $e_i$ is bounded by the dimensions of $S_i$.
	\Cref{p:bal-rect-size} implies that the sum of the dimensions of $S_1, S_2, \dots, S_n$ is $\fO( c_\pi \log n )$.
\end{proof}

\subsection{231-avoiding point sets}

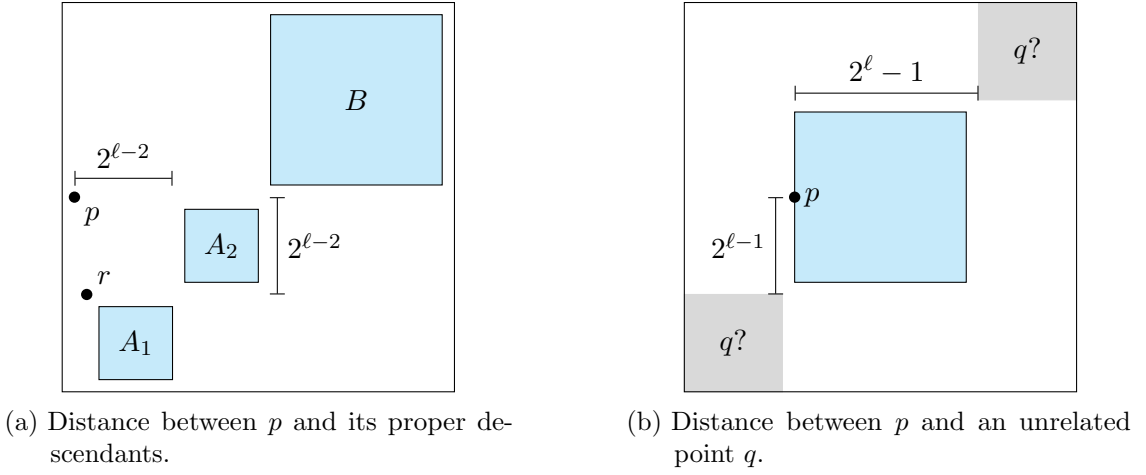
\begin{figure}
	\def\k{5}
	\pgfmathsetmacro\n{2^\k-1}
	\pgfmathsetmacro\na{2^(\k-1)-1}
	\pgfmathsetmacro\nb{2^(\k-2)-1}
	\newcommand{\measureDist}{\n*0.05}
	
	\hfill
	\begin{subfigure}{.4\textwidth}
		\centering
		\begin{tikzpicture}[
				scale=5/\n,
				yscale=-1,
				box/.style=subrect
			]
			\draw (0,0) rectangle (\n+1, \n+1);
			
			\node[point] (p) at (1,\na+1) {};
			\node[below right] at (p) {$p$};
			
			\node[point] (r) at (2,\na+\nb+2) {};
			\node[above right] at (r) {$r$};
			\draw[box] (3,\na+\nb+3) rectangle node {$A_1$} (\nb+2,\n);
			\draw[box] (\nb+3,\na+2) rectangle node {$A_2$} (\na+1,\na+\nb+1);
	
			\draw[box] (\na+2,1) rectangle node {$B$} (\n,\na);
			
			\draw[|-|] (1,\na+1-\measureDist) -- node[above] {$2^{\ell-2}$}
				(\nb+2,\na+1-\measureDist);
			
			\draw[|-|] (\na+1+\measureDist,\na+1) -- node[right] {$2^{\ell-2}$}
				(\na+1+\measureDist,\na+\nb+2);
		\end{tikzpicture}
		\caption{Distance between $p$ and its proper descendants.}\label{sfig:mst-sep-lb:desc}
	\end{subfigure}%
	\hfill
	\begin{subfigure}{.4\textwidth}
		\centering
		\begin{tikzpicture}[
				scale=5/\n,
				box/.style={draw, fill={white!80!cyan}},
				graybox/.style={fill, {white!85!black}}
			]
			
			\begin{scope}[shift={(\nb+1, \nb+1)}]
			\pgfmathsetmacro\ngap{\nb+1}
			\draw[box] (1,1) rectangle (\na,\na);
			\node[point] (p) at (1,\nb+1) {};
			\node[right] at (p) {$p$};
			
			\draw[graybox] (-\ngap,0) rectangle (0,-\ngap);
			\node at (-\ngap/2, - \ngap/2) {$q$?};
			\draw[|-|] (1-\measureDist,0) -- node[left] {$2^{\ell-1}$}
			(1-\measureDist,\nb+1);
			
			\draw[graybox] (\na+1+\ngap,\na+1) rectangle (\na+1,\na+1+\ngap);
			\node at (\na+1+\ngap/2, \na+1+\ngap/2) {$q$?};
			\draw[|-|] (1,\na+\measureDist) -- node[above] {$2^\ell - 1$} (\na+1,\na+\measureDist);
			\end{scope}
			
			\draw (0,0) rectangle (\n+1, \n+1);
		\end{tikzpicture}
		\caption{Distance between $p$ and an unrelated point $q$.}\label{sfig:mst-sep-lb:unrelated}
	\end{subfigure}
	\hfill
	\caption{Sketches for the proof of \cref{p:mst-sep-lb}.} 
\end{figure}

\restateMSTSepLowerBound*
\begin{proof}
	Let $k \in \N_+$. We recursively define a point set $P_k$ of size $n = 2^k-1$ on the integer grid $[n] \times [n]$ and an associated rooted tree $T_k$ with node set $P_k$.
	We show that $\NN(P_k) \ge k 2^{k-2} > \frac{1}{4} n \log n$, even in $L_\infty$. This implies the theorem (after scaling).
	
	$P_k$ consists of a single point $p_k = (1, 2^{k-1})$ and two shifted copies of $P_{k-1}$, one directly to the right and below $p$, and the other in the top right corner.
	More formally, let $P_1 = \{p_1\} = \{(1,1)\}$, with $T_1$ being the tree on the single node $p_1$. For $k \ge 2$, let $P_k = \{p_k\} \cup A \cup B$, where
	\begin{align*}
		& A = \{ (x+1, y) \mid (x,y) \in P_{k-1} \} \text{ and}\\
		& B = \{ (x+2^{k-1}, y+2^{k-1}) \mid (x,y) \in P_{k-1} \}.
	\end{align*}
	
	Observe that $P_k$ is in general position, avoids 231, and contains precisely $2^k-1$ points in $[2^k-1]^2$. Define $T_k$ as the tree with $p_k$ at its root, and the trees associated to $A$ and $B$ as subtrees.
	
	For a point $p \in P_k$, let the \emph{level} $\ell_p$ of $p$ in $T_k$ be the depth of the subtree of $T_k$ rooted at $p$. Observe that $p$ is the leftmost point in a shifted copy of $P_{\ell_p}$. Call that copy $Q_p$.
	
	We now give lower bounds for all distances between points.
	
	First, let $p, q \in P_k$ such that $q$ is a proper descendant of $p$, and let $\ell$ be the level of~$p$.
	\Cref{sfig:mst-sep-lb:desc} shows $Q_p$.
	Observe that $q \in Q_p$, and further, there is a child $r$ of $p$ such that $q \in Q_r$.
	Let $A_1, A_2$ be the point sets corresponding to the subtrees of $r$.
	If $q = r$, then $d(p,q) \ge 2^{\ell-2}$. The same is true if $q \in A_1 \cup A_2$. If $q \in B$, we even have $d(p,q) \ge 2^{\ell-1}$. In any case, $d(p,q) \ge 2^{\ell-2}$.
	
	Now suppose $p, q \in P_k$ are unrelated in $T_k$, and let $\ell$ be the level of $p$. Then $q$ is below and to the left of all points in $Q_p$, or $q$ is above and to the right of all points in $Q_p$. \Cref{sfig:mst-sep-lb:unrelated} clearly shows that then $d(p,q) \ge 2^{\ell-1}$.
	
	Our two observations imply that for each point $p$ of level $\ell$, the distance to the nearest neighbor in $P_k$ is at least $2^{\ell-2}$. The total sum of nearest neighbor distances is therefore
	\begin{align*}
	\sum_{\ell=1}^k 2^{k+1-\ell} \cdot 2^{\ell-2} = k \cdot 2^{k-2}.\tag*{\qedhere}
	\end{align*}
\end{proof}

\restateMSTSepSubclasses*
\begin{proof}
	Let $k = |\pi|$ and $n = |P|$.
	We slightly strengthen the claim to prove it by induction on $n+k$. For every point set $P \subset [0,w]\times[0,h]$ avoiding 231 and $\pi$, there is a spanning tree on a superset of $P' = P \cup \{(0,0), (w,0), (0,h), (w,h)\}$ of total weight at most $2k (w+h)$. This implies $\TSP(P) \le 3 \MStT(P) \le 12k$, as desired.
	
	If $k = 1$, then $P = \emptyset$, and thus $\MST(P') \le 2(w+h)$. If $n = 1$, then clearly $\MST(P) = 0$.
	Now suppose $k, n \ge 2$. Since $P$ and $\pi$ both avoid 231, they can be decomposed as shown in \cref{sfig:mst-231-av-struct}, where $P = \{(0,h_A)\} \cup A \cup B$ with $A \subset [0,w_A] \times [0,h_A]$ and $B \subset [w_A,w] \times [h_A,h]$, and $\pi = (p) \circ \alpha \circ \beta$.
	We additionally write $w_B = w-w_A$ and $h_B = h-h_A$. Note that $A$ avoids $\alpha$ or $B$ avoids $\beta$, otherwise $P$ contains $\pi$.\footnote{One of $\alpha$ and $\beta$ might be empty, in which case it cannot be avoided.}
	
	\begin{figure}
		\def\bboxDist{0.2}
		\def\measureDist{0.2}
		
		\begin{subfigure}{.3\textwidth}
			\centering
			\begin{tikzpicture}[
					yscale=-1,
					box/.style={draw, fill={white!80!cyan}},
					pointheap/.style={box, circle, inner sep=3mm},
					newpoint/.style={point, red},
					edge/.style={thick},
					newedge/.style={thick, red, densely dashed}
				]
				
				\begin{scope}[local bounding box=figA]
					\draw[box] (0,1) rectangle node {$A$} (1,2);
					\draw[box] (1,0) rectangle node {$B$} (2,1);
					
					\node[point] (root) at (0,1) {};
					
					\draw[|-|] (0,0-\measureDist) -- node[above] {$w_A$} (1,0-\measureDist);
					\draw[-|]  (1,0-\measureDist) -- node[above] {$w_B$} (2,0-\measureDist);
					
					\draw[|-|] (2+\measureDist,0) -- node[right] {$h_B$} (2+\measureDist,1);
					\draw[-|]  (2+\measureDist,1) -- node[right] {$h_A$} (2+\measureDist,2);
					
					\node at (0,-0.5) {};
					\node at (0,2) {};
				\end{scope}
				\draw ($(figA.north west)-(\bboxDist,\bboxDist)$) rectangle ($(figA.south east)+(\bboxDist,\bboxDist)$);
			\end{tikzpicture}
			\caption{A 231-avoiding point set.\\\strut}\label{sfig:mst-231-av-struct}
		\end{subfigure}%
		\hfill
		\begin{subfigure}{.6\textwidth}
			\centering
			\begin{tikzpicture}[
					yscale=-1,
					box/.style={draw, fill={white!80!cyan}},
					pointheap/.style={box, circle, inner sep=3mm},
					newpoint/.style={point, red},
					edge/.style={thick},
					newedge/.style={thick, red, densely dashed}
				]
				\begin{scope}[local bounding box=figB]
					\node[pointheap] (A) at (0.5,1.5) {};
					\node[pointheap] (B) at (1.5,0.5) {};
					
					\node[newpoint] (bl) at (0,0) {};
					\node[point]   (blA) at (0,1) {};
					\node[point]   (tlA) at (0,2) {};
					\node[point]   (blB) at (1,0) {};
					\node[point]   (c)   at (1,1) {};
					\node[point]   (trA) at (1,2) {};
					\node[point]   (brB) at (2,0) {};
					\node[point]   (trB) at (2,1) {};
					\node[newpoint] (tr) at (2,2) {};
					
					\draw (A) -- (blA) (A) -- (tlA) (A) -- (c) (A) -- (trA);
					\draw (B) -- (blB) (B) -- (c) (B) -- (brB) (B) -- (trB);
					
					\draw[newedge]
							(bl) -- node[above] {$w_A$} (blB)
							(tr) -- node[right] {$h_A$} (trB);
					
					\node at (0,-0.5) {};
					\node at (0,2) {};
				\end{scope}
				\draw ($(figB.north west)-(\bboxDist,\bboxDist)$) rectangle ($(figB.south east)+(\bboxDist,\bboxDist)$);
			\end{tikzpicture}
			\caption{Recursively constructing a spanning tree if $A$ avoids~$\alpha$. Newly added edges and points are colored \textcolor{red}{red}.}\label{sfig:mst-231-av-constr}
		\end{subfigure}
		\caption{Sketches for \cref{p:mst-sep-subclasses-ind}.}
	\end{figure}
	
	We recursively construct spanning trees on $A' = A \cup \{(0,0), (0,h_A), (w_A,0), (w_A,h_A)\}$ and $B' = B \cup \{(w_A,h_A), (w_A,h), (w,h_A), (w,h)\}$. 
	 Together, these two trees already form a spanning tree on a subset of $P'$ excluding the two points $(w,0)$ and $(0,h)$.
	To connect the two remaining points, we distinguish between two cases:
	\begin{itemize}
		\item $A$ avoids $\alpha$. Then we connect $(w,0)$ to $(w,h_A)$ and $(0,h)$ to $(w_A,h)$. The total cost of the spanning trees is (by induction, and considering $|\alpha| \le |\pi|-1$):
		\begin{align*}
			& 2 |\alpha| (w_A+h_A) + 2 |\pi| (w_B+h_B) + (w_A+h_A) \le 2 |\pi|(w+h).
		\end{align*}
		\item $B$ avoids $\beta$. Then we connect $(w,0)$ to $(w_B,0)$ and $(0,h)$ to $(0,h_A)$. We obtain the desired bound similarly to the previous case.\qedhere
	\end{itemize}
\end{proof}

\subsection{Twin-width lower bound}

Next, we show that there are point sets $P$ such that the constant factor hidden in $\fO(\log n)$ is at least $\Omega(d / \log d)$ where $d$ is the twin-width of $P$.

\restateMSTTwwLB*


For $d \ge 2$, let $G_d$ be the point set with coordinates $((i-1)\frac{1}{d}+(d-j+1)\frac{1}{d^2}, (j-1)\frac{1}{d} + i\frac{1}{d^2})$ for $i,j\in [d]$, i.e., the $d \times d$ canonical grid scaled to fit $[0,1]^2$ (see \cref{sfig:can-grid}). By \cref{lem:canonical-grid-tww} we have $\tww(G_d) = d$.

First, we show that any Steiner tree of $G_d$ must have weight at least linear in $d$. Indeed, since for every pair $p_1, p_2 \in G_d$ we have the distance $d(p_1,p_2) \geq \frac{1}{d}$, it follows that $\NN(G_d) \geq d^2 \cdot \frac{1}{d} = d$. Thus, $\MStT(P) \geq \frac{0.5}{1.22} \cdot \NN(P) \geq 0.4 d$.

Now we inductively define for each $d \ge 2$ and $t \ge 1$ a point set $P^d_t$.
Let $P^d_1 = G_d$ and for $t \ge 2$, let $P^d_t$ be the point set obtained by inflating each point of~$G_d$ with a properly scaled-down copy of $P^d_{t-1}$.
Formally, we set $P^d_t = \bigcup_{p \in G_d} p + \frac{1}{d^2}P^d_{t-1}$.

\begin{lemma}\label{p:mst-tww-constr}
	For every $d \ge 40$ and $t \ge 1$, we have $\MStT(P^d_t) \ge \frac{d}{5} t$.
\end{lemma}
\begin{proof}
	\newcommand{\rI}{\mathrm{I}}
	\newcommand{\rO}{\mathrm{O}}
	We proceed by induction on~$t$.
	When $t=1$,  we have $\MStT(P^d_1) \ge 0.4 d > \frac{d}{5}$.
 

	Now let $t > 1$ and let $T$ be a minimum Steiner tree of $P^d_t$.
	For every $p \in G_d$, let $A_p$ be the box $[0,\frac{1}{d^2}]$ translated such that its bottom left corner coincides with~$p$, i.e., $ A_p = p + [0,\frac{1}{d^2}]^2$.
	Observe that each $A_p$ contains precisely one small copy of $P^d_{t-1}$. 
	We split all the edges of $T$ into parts that lie inside and outside of $\bigcup_{p \in G_d} A_p$.
	Let $T_\rI$ be the part inside $\bigcup_{p \in G_d} A_p$ and let $W_\rI$ be its total weight.
	Similarly, let $T_\rO$ be the part outside $\bigcup_{p \in G_d} A_p$ and let $W_\rO$ be its total weight.

	First, let us bound $W_\rO$.
	If we add to $T_\rO$ the boundary of each $A_p$ for $p \in G_d$, we obtain a 
 Steiner tree of $G_d$.\footnote{Or, technically, a supergraph of a Steiner tree.}
	The weight of the added edges is exactly $d^2 \frac{4}{d^2} = 4$, so  
 by the above observation 
 we get
	$W_\rO \ge 0.4d - 4$.
	We proceed to bound $W_\rI$.
	Inside each $A_p$, we get a Steiner tree $T_p$ of the point set $\frac{1}{d^2} P^d_{t-1}$ by adding the boundary of $A_p$ to the part of $T_\rI$ that lies inside $A_p$.
	By induction, the weight of $T_p$ must be at least $\frac{1}{d^2}\frac{d}{5} (t-1) = \frac{t-1}{5 d}$.
	Summing over $A_p$ for every $p \in G_d$ and subtracting the total length of their boundaries, we get
	\[W_\rI \ge d^2 \cdot \frac{t-1}{5 d} - d^2 \frac{4}{d^2} \ge \frac{d}{5} (t-1) - 4.\]

	The total weight of $T$ is then
	\[W_\rI + W_\rO \ge \frac{d}{5} (t-1) + \frac{2d}{5} - 8 \ge \frac{d}{5} (t-1) + \frac{2d}{5} - \frac{d}{5} = \frac{d}{5} t,\]
	where the second inequality holds since $d \ge 40$.
\end{proof}

To prove \cref{p:mst-tww-lb}, consider two separate cases.
If $d < 40$, apply \cref{p:mst-sep-lb} to obtain a separable  point set $P$ such that $\TSP(P) \in \Omega(\log n)$.
Otherwise $d \ge 40$ and we apply \cref{p:mst-tww-constr} with $t = \floor{\frac12 \log_d n} = \left\lfloor\frac{\log n}{2 \log d}\right\rfloor$.
Then $P^d_t$ is a point set of size $d^{2t} \le n$ and twin-width~$d$ such that $\TSP(P^d_t) \in \Omega(\frac{d}{\log d} \log n)$.

\section{Open questions}
\label{sec8}

In this paper we initiated the study of optimization problems with pattern-avoiding input, with three central problems as case studies. Our work raises several open questions and directions; we list the ones we find the most interesting. 

\paragraph{BST}

\begin{itemize}
\item Our upper bound $\fO(n \, {c_\pi^2})$ and lower bound $\Omega(n\log{c_\pi})$ for $\OPT$ hold for each avoided pattern $\pi$. While we know special families of patterns where the upper bound is not tight~\cite{Isaac18}, we have no examples where the lower bound is not tight. 
\item Our upper bound holds for the offline optimum of serving BST access sequences. Can similar bounds be shown for online algorithms, such as splay or greedy trees? 
\end{itemize}

\paragraph{k-server}
\begin{itemize}
\item We give tight bounds on the $k$-server cost of pattern-avoiding access sequences. Yet, a full characterization in terms of the concrete avoided pattern $\pi$ is still missing (gaps between $n^{\Omega(1/k)}$ and $n^{\fO(1/\log k)}$). We do establish that the cost can have at least three different growth rates depending on $\pi$: $\fO(1)$ if $\pi$ is monotone; $n^{\Theta(1/k)}$ if $\pi$ is non-monotone and has length three; and $n^{\Theta(1/ \log k)}$ if $\pi$ is non-separable.


\item Our results concern the offline cost. Does the online \emph{competitive ratio} $k$ improve under pattern-avoidance? Note that a lower bound of $k$ holds even for the simplest $k$-point metrics~\cite{Manasse}. Yet, when viewed on the line, the lower bound constructions are emphatically \emph{not} pattern-avoiding, suggesting a possible improvement. 
\item Does the effect of pattern-avoidance extend to 
more general metric spaces, e.g., trees (the extension of double-coverage (DC) is $k$-competitive on trees). This would require the definition of more general pattern-concepts, which may be interesting in itself. 
\item By the known competitive results, our bounds immediately transfer to the DC algorithm (with a factor $k$ cost-increase). DC is simple and has an intuitive potential-based analysis~\cite{Chrobak, Borodin}. This analysis does not hint at why DC is adaptive to avoided patterns. Can this be seen directly through the analysis of DC or other competitive $k$-server algorithms?
\end{itemize}

\paragraph{TSP}
\begin{itemize}
	\item Our $\fO(\log{n})$ bound for TSP raises the question for which special cases this cost may be $\fO(1)$. We have a complete characterization for principal classes (families defined by the avoidance of a single pattern). A broader characterization, e.g., in terms of \emph{path-width} and other parameters and \emph{merging}~\cite{gridwidth} may be possible.
	
	%
	%
\end{itemize}

\paragraph{Further questions}
\begin{itemize}
\item Our result for BST holds for bounded twin-width, regardless of avoided patterns. For $k$-server and TSP, our main bounds depend on $c_{\pi}$, for a pattern $\pi$ avoided in the input. Recall that distance-balanced merge sequences exist for individual permutations with bounded twin-width, but with exponential width blowup. 
We leave open whether this dependence can be improved. 

\item Which other online or geometric problems benefit from pattern-avoidance in the input? Online problems with sequential structure that may be good candidates include list labeling, scheduling, online matching, or interval coloring~\cite{Borodin}. 
\item 
Pattern-avoidance is sensitive to perturbation, whereas the optimal cost typically changes smoothly. (Note however, that \emph{twin-width} affords some robustness, e.g., a constant number of pattern-occurrences can increase twin-width only by a constant.) Do other, more robust concepts of pattern  (e.g.,~\cite{Newman}) relate to complexity?

\item The relation between $\tww(\Av(\pi))$ and $c_{\pi}$, i.e., the largest possible twin-width attainable by a $\pi$-avoiding permutation and the Füredi-Hajnal limit of $\pi$, is not yet understood.
For instance, it is possible that they are polynomially related (which would imply that $\tww(\Av(\pi))$ is exponential in $|\pi|$ for almost all $\pi$), but it is also possible that $\tww(\Av(\pi))$ is polynomially bounded in $|\pi|$.

\end{itemize}

\paragraph{Acknowledgement.} The work evolved from discussions at Dagstuhl Seminar 23121. We thank Yaniv Sadeh for helpful comments on an earlier version of this paper. We also thank anonymous reviewers for their thoughtful feedback. 

\newpage
\small
\bibliographystyle{alphaurl}
\bibliography{main}

\end{document}

%% file: figs/ass-big.tex
\draw[bigboxBG] (0, 0) rectangle (8.6,2);
\draw[boxBG] (0, 1.4) rectangle (1.1,2);
\draw[boxBG] (7.1, 0) rectangle (8.6,0.5);
\draw[boxBG] (1.7, 3) rectangle (2.6,3.6);
\draw[boxBG] (4.8, 0.8) rectangle (6,2.3);
\draw[boxBG] (5.5, 1.2) rectangle (8.2,2.6);
\def\minx{0-\bboxDist}
\def\maxx{8.6+\bboxDist}
\def\miny{0-\bboxDist}
\def\maxy{3.6+\bboxDist}
\draw[gridline] (0,\miny) -- (0,\maxy);
\draw[gridline] (1.1,\miny) -- (1.1,\maxy);
\draw[gridline] (1.7,\miny) -- (1.7,\maxy);
\draw[gridline] (2.6,\miny) -- (2.6,\maxy);
\draw[gridline] (4.8,\miny) -- (4.8,\maxy);
\draw[gridline] (5.5,\miny) -- (5.5,\maxy);
\draw[gridline] (6,\miny) -- (6,\maxy);
\draw[gridline] (7.1,\miny) -- (7.1,\maxy);
\draw[gridline] (8.2,\miny) -- (8.2,\maxy);
\draw[gridline] (8.6,\miny) -- (8.6,\maxy);
\draw[gridline] (\minx,0) -- (\maxx,0);
\draw[gridline] (\minx,0.5) -- (\maxx,0.5);
\draw[gridline] (\minx,0.8) -- (\maxx,0.8);
\draw[gridline] (\minx,1.2) -- (\maxx,1.2);
\draw[gridline] (\minx,1.4) -- (\maxx,1.4);
\draw[gridline] (\minx,2) -- (\maxx,2);
\draw[gridline] (\minx,2.3) -- (\maxx,2.3);
\draw[gridline] (\minx,2.6) -- (\maxx,2.6);
\draw[gridline] (\minx,3) -- (\maxx,3);
\draw[gridline] (\minx,3.6) -- (\maxx,3.6);
\draw[bigboxFG] (0, 0) rectangle node {$Q$} (8.6,2);
\draw[boxFG] (1.7, 3) rectangle (2.6,3.6);
\draw[boxFG] (4.8, 0.8) rectangle (6,2.3);
\draw[boxFG] (5.5, 1.2) rectangle (8.2,2.6);
\draw[boxFG] (0, 1.4) rectangle node {$Q_1$} (1.1,2);
\draw[boxFG] (7.1, 0) rectangle node {$Q_2$} (8.6,0.5);
\node[newpoint] at (0,0) {};
\node[newpoint] at (0,0.5) {};
\node[newpoint] at (0,0.8) {};
\node[newpoint] at (0,1.2) {};
\node[oldpoint] at (0,1.4) {};
\node[oldpoint] at (0,2) {};
\node[newpoint] at (1.1,0) {};
\node[newpoint] at (1.1,0.5) {};
\node[newpoint] at (1.1,0.8) {};
\node[newpoint] at (1.1,1.2) {};
\node[oldpoint] at (1.1,1.4) {};
\node[oldpoint] at (1.1,2) {};
\node[newpoint] at (1.7,0) {};
\node[newpoint] at (1.7,0.5) {};
\node[newpoint] at (1.7,0.8) {};
\node[newpoint] at (1.7,1.2) {};
\node[newpoint] at (1.7,1.4) {};
\node[newpoint] at (1.7,2) {};
\node[oldpoint] at (1.7,3) {};
\node[oldpoint] at (1.7,3.6) {};
\node[newpoint] at (2.6,0) {};
\node[newpoint] at (2.6,0.5) {};
\node[newpoint] at (2.6,0.8) {};
\node[newpoint] at (2.6,1.2) {};
\node[newpoint] at (2.6,1.4) {};
\node[newpoint] at (2.6,2) {};
\node[oldpoint] at (2.6,3) {};
\node[oldpoint] at (2.6,3.6) {};
\node[newpoint] at (4.8,0) {};
\node[newpoint] at (4.8,0.5) {};
\node[oldpoint] at (4.8,0.8) {};
\node[oldpoint] at (4.8,1.2) {};
\node[oldpoint] at (4.8,1.4) {};
\node[oldpoint] at (4.8,2) {};
\node[oldpoint] at (4.8,2.3) {};
\node[newpoint] at (5.5,0) {};
\node[newpoint] at (5.5,0.5) {};
\node[oldpoint] at (5.5,0.8) {};
\node[oldpoint] at (5.5,1.2) {};
\node[oldpoint] at (5.5,1.4) {};
\node[oldpoint] at (5.5,2) {};
\node[oldpoint] at (5.5,2.3) {};
\node[oldpoint] at (5.5,2.6) {};
\node[newpoint] at (6,0) {};
\node[newpoint] at (6,0.5) {};
\node[oldpoint] at (6,0.8) {};
\node[oldpoint] at (6,1.2) {};
\node[oldpoint] at (6,1.4) {};
\node[oldpoint] at (6,2) {};
\node[oldpoint] at (6,2.3) {};
\node[oldpoint] at (6,2.6) {};
\node[oldpoint] at (7.1,0) {};
\node[oldpoint] at (7.1,0.5) {};
\node[newpoint] at (7.1,0.8) {};
\node[oldpoint] at (7.1,1.2) {};
\node[oldpoint] at (7.1,1.4) {};
\node[oldpoint] at (7.1,2) {};
\node[oldpoint] at (7.1,2.3) {};
\node[oldpoint] at (7.1,2.6) {};
\node[oldpoint] at (8.2,0) {};
\node[oldpoint] at (8.2,0.5) {};
\node[newpoint] at (8.2,0.8) {};
\node[oldpoint] at (8.2,1.2) {};
\node[oldpoint] at (8.2,1.4) {};
\node[oldpoint] at (8.2,2) {};
\node[oldpoint] at (8.2,2.3) {};
\node[oldpoint] at (8.2,2.6) {};
\node[oldpoint] at (8.6,0) {};
\node[oldpoint] at (8.6,0.5) {};
\node[newpoint] at (8.6,0.8) {};
\node[newpoint] at (8.6,1.2) {};
\node[newpoint] at (8.6,1.4) {};
\node[newpoint] at (8.6,2) {};

%% file: figs/k-server-lb-1.tex
\node[point] at (0.00000,0.01250) {};
\node[point] at (0.05263,0.73750) {};
\node[point] at (0.10526,0.03750) {};
\node[point] at (0.15789,0.71250) {};
\node[point] at (0.21053,0.06250) {};
\node[point] at (0.26316,0.68750) {};
\node[point] at (0.31579,0.08750) {};
\node[point] at (0.36842,0.66250) {};
\node[point] at (0.42105,0.11250) {};
\node[point] at (0.47368,0.63750) {};
\node[point] at (0.52632,0.13750) {};
\node[point] at (0.57895,0.61250) {};
\node[point] at (0.63158,0.16250) {};
\node[point] at (0.68421,0.58750) {};
\node[point] at (0.73684,0.18750) {};
\node[point] at (0.78947,0.56250) {};
\node[point] at (0.84211,0.21250) {};
\node[point] at (0.89474,0.53750) {};
\node[point] at (0.94737,0.23750) {};
\node[point] at (1.00000,0.51250) {};

%% file: figs/k-server-lb-2.tex
\node[point] at (0.00000,0.00125) {};
\node[point] at (0.01010,0.03625) {};
\node[point] at (0.02020,0.00375) {};
\node[point] at (0.03030,0.03375) {};
\node[point] at (0.04040,0.00625) {};
\node[point] at (0.05051,0.03125) {};
\node[point] at (0.06061,0.00875) {};
\node[point] at (0.07071,0.02875) {};
\node[point] at (0.08081,0.01125) {};
\node[point] at (0.09091,0.02625) {};
\node[point] at (0.10101,0.74875) {};
\node[point] at (0.11111,0.71375) {};
\node[point] at (0.12121,0.74625) {};
\node[point] at (0.13131,0.71625) {};
\node[point] at (0.14141,0.74375) {};
\node[point] at (0.15152,0.71875) {};
\node[point] at (0.16162,0.74125) {};
\node[point] at (0.17172,0.72125) {};
\node[point] at (0.18182,0.73875) {};
\node[point] at (0.19192,0.72375) {};
\node[point] at (0.20202,0.05125) {};
\node[point] at (0.21212,0.08625) {};
\node[point] at (0.22222,0.05375) {};
\node[point] at (0.23232,0.08375) {};
\node[point] at (0.24242,0.05625) {};
\node[point] at (0.25253,0.08125) {};
\node[point] at (0.26263,0.05875) {};
\node[point] at (0.27273,0.07875) {};
\node[point] at (0.28283,0.06125) {};
\node[point] at (0.29293,0.07625) {};
\node[point] at (0.30303,0.69875) {};
\node[point] at (0.31313,0.66375) {};
\node[point] at (0.32323,0.69625) {};
\node[point] at (0.33333,0.66625) {};
\node[point] at (0.34343,0.69375) {};
\node[point] at (0.35354,0.66875) {};
\node[point] at (0.36364,0.69125) {};
\node[point] at (0.37374,0.67125) {};
\node[point] at (0.38384,0.68875) {};
\node[point] at (0.39394,0.67375) {};
\node[point] at (0.40404,0.10125) {};
\node[point] at (0.41414,0.13625) {};
\node[point] at (0.42424,0.10375) {};
\node[point] at (0.43434,0.13375) {};
\node[point] at (0.44444,0.10625) {};
\node[point] at (0.45455,0.13125) {};
\node[point] at (0.46465,0.10875) {};
\node[point] at (0.47475,0.12875) {};
\node[point] at (0.48485,0.11125) {};
\node[point] at (0.49495,0.12625) {};
\node[point] at (0.50505,0.64875) {};
\node[point] at (0.51515,0.61375) {};
\node[point] at (0.52525,0.64625) {};
\node[point] at (0.53535,0.61625) {};
\node[point] at (0.54545,0.64375) {};
\node[point] at (0.55556,0.61875) {};
\node[point] at (0.56566,0.64125) {};
\node[point] at (0.57576,0.62125) {};
\node[point] at (0.58586,0.63875) {};
\node[point] at (0.59596,0.62375) {};
\node[point] at (0.60606,0.15125) {};
\node[point] at (0.61616,0.18625) {};
\node[point] at (0.62626,0.15375) {};
\node[point] at (0.63636,0.18375) {};
\node[point] at (0.64646,0.15625) {};
\node[point] at (0.65657,0.18125) {};
\node[point] at (0.66667,0.15875) {};
\node[point] at (0.67677,0.17875) {};
\node[point] at (0.68687,0.16125) {};
\node[point] at (0.69697,0.17625) {};
\node[point] at (0.70707,0.59875) {};
\node[point] at (0.71717,0.56375) {};
\node[point] at (0.72727,0.59625) {};
\node[point] at (0.73737,0.56625) {};
\node[point] at (0.74747,0.59375) {};
\node[point] at (0.75758,0.56875) {};
\node[point] at (0.76768,0.59125) {};
\node[point] at (0.77778,0.57125) {};
\node[point] at (0.78788,0.58875) {};
\node[point] at (0.79798,0.57375) {};
\node[point] at (0.80808,0.20125) {};
\node[point] at (0.81818,0.23625) {};
\node[point] at (0.82828,0.20375) {};
\node[point] at (0.83838,0.23375) {};
\node[point] at (0.84848,0.20625) {};
\node[point] at (0.85859,0.23125) {};
\node[point] at (0.86869,0.20875) {};
\node[point] at (0.87879,0.22875) {};
\node[point] at (0.88889,0.21125) {};
\node[point] at (0.89899,0.22625) {};
\node[point] at (0.90909,0.54875) {};
\node[point] at (0.91919,0.51375) {};
\node[point] at (0.92929,0.54625) {};
\node[point] at (0.93939,0.51625) {};
\node[point] at (0.94949,0.54375) {};
\node[point] at (0.95960,0.51875) {};
\node[point] at (0.96970,0.54125) {};
\node[point] at (0.97980,0.52125) {};
\node[point] at (0.98990,0.53875) {};
\node[point] at (1.00000,0.52375) {};

%% file: figs/k-server-lb-tww-1.tex
\node[point] at (0.00000,0.00417) {};
\node[point] at (0.01695,0.24583) {};
\node[point] at (0.03390,0.33750) {};
\node[point] at (0.05085,0.57917) {};
\node[point] at (0.06780,0.67083) {};
\node[point] at (0.08475,0.91250) {};
\node[point] at (0.10169,0.01250) {};
\node[point] at (0.11864,0.23750) {};
\node[point] at (0.13559,0.34583) {};
\node[point] at (0.15254,0.57083) {};
\node[point] at (0.16949,0.67917) {};
\node[point] at (0.18644,0.90417) {};
\node[point] at (0.20339,0.02083) {};
\node[point] at (0.22034,0.22917) {};
\node[point] at (0.23729,0.35417) {};
\node[point] at (0.25424,0.56250) {};
\node[point] at (0.27119,0.68750) {};
\node[point] at (0.28814,0.89583) {};
\node[point] at (0.30508,0.02917) {};
\node[point] at (0.32203,0.22083) {};
\node[point] at (0.33898,0.36250) {};
\node[point] at (0.35593,0.55417) {};
\node[point] at (0.37288,0.69583) {};
\node[point] at (0.38983,0.88750) {};
\node[point] at (0.40678,0.03750) {};
\node[point] at (0.42373,0.21250) {};
\node[point] at (0.44068,0.37083) {};
\node[point] at (0.45763,0.54583) {};
\node[point] at (0.47458,0.70417) {};
\node[point] at (0.49153,0.87917) {};
\node[point] at (0.50847,0.04583) {};
\node[point] at (0.52542,0.20417) {};
\node[point] at (0.54237,0.37917) {};
\node[point] at (0.55932,0.53750) {};
\node[point] at (0.57627,0.71250) {};
\node[point] at (0.59322,0.87083) {};
\node[point] at (0.61017,0.05417) {};
\node[point] at (0.62712,0.19583) {};
\node[point] at (0.64407,0.38750) {};
\node[point] at (0.66102,0.52917) {};
\node[point] at (0.67797,0.72083) {};
\node[point] at (0.69492,0.86250) {};
\node[point] at (0.71186,0.06250) {};
\node[point] at (0.72881,0.18750) {};
\node[point] at (0.74576,0.39583) {};
\node[point] at (0.76271,0.52083) {};
\node[point] at (0.77966,0.72917) {};
\node[point] at (0.79661,0.85417) {};
\node[point] at (0.81356,0.07083) {};
\node[point] at (0.83051,0.17917) {};
\node[point] at (0.84746,0.40417) {};
\node[point] at (0.86441,0.51250) {};
\node[point] at (0.88136,0.73750) {};
\node[point] at (0.89831,0.84583) {};
\node[point] at (0.91525,0.07917) {};
\node[point] at (0.93220,0.17083) {};
\node[point] at (0.94915,0.41250) {};
\node[point] at (0.96610,0.50417) {};
\node[point] at (0.98305,0.74583) {};
\node[point] at (1.00000,0.83750) {};

%% file: figs/k-server-lb-3.tex
\node[point] at (0.00000,0.00125) {};
\node[point] at (0.01852,0.05000) {};
\node[point] at (0.03704,0.00375) {};
\node[point] at (0.05556,0.04750) {};
\node[point] at (0.07407,0.00625) {};
\node[point] at (0.09259,0.04500) {};
\node[point] at (0.11111,0.00875) {};
\node[point] at (0.12963,0.04250) {};
\node[point] at (0.14815,0.01125) {};
\node[point] at (0.16667,0.04000) {};
\node[point] at (0.18519,1.00000) {};
\node[point] at (0.20370,0.05125) {};
\node[point] at (0.22222,0.10000) {};
\node[point] at (0.24074,0.05375) {};
\node[point] at (0.25926,0.09750) {};
\node[point] at (0.27778,0.05625) {};
\node[point] at (0.29630,0.09500) {};
\node[point] at (0.31481,0.05875) {};
\node[point] at (0.33333,0.09250) {};
\node[point] at (0.35185,0.06125) {};
\node[point] at (0.37037,0.09000) {};
\node[point] at (0.38889,0.95000) {};
\node[point] at (0.40741,0.10125) {};
\node[point] at (0.42593,0.15000) {};
\node[point] at (0.44444,0.10375) {};
\node[point] at (0.46296,0.14750) {};
\node[point] at (0.48148,0.10625) {};
\node[point] at (0.50000,0.14500) {};
\node[point] at (0.51852,0.10875) {};
\node[point] at (0.53704,0.14250) {};
\node[point] at (0.55556,0.11125) {};
\node[point] at (0.57407,0.14000) {};
\node[point] at (0.59259,0.90000) {};
\node[point] at (0.61111,0.15125) {};
\node[point] at (0.62963,0.20000) {};
\node[point] at (0.64815,0.15375) {};
\node[point] at (0.66667,0.19750) {};
\node[point] at (0.68519,0.15625) {};
\node[point] at (0.70370,0.19500) {};
\node[point] at (0.72222,0.15875) {};
\node[point] at (0.74074,0.19250) {};
\node[point] at (0.75926,0.16125) {};
\node[point] at (0.77778,0.19000) {};
\node[point] at (0.79630,0.85000) {};
\node[point] at (0.81481,0.20125) {};
\node[point] at (0.83333,0.25000) {};
\node[point] at (0.85185,0.20375) {};
\node[point] at (0.87037,0.24750) {};
\node[point] at (0.88889,0.20625) {};
\node[point] at (0.90741,0.24500) {};
\node[point] at (0.92593,0.20875) {};
\node[point] at (0.94444,0.24250) {};
\node[point] at (0.96296,0.21125) {};
\node[point] at (0.98148,0.24000) {};
\node[point] at (1.00000,0.80000) {};